\documentclass[11pt,letterpaper]{article}

\usepackage[margin=1in]{geometry}

\usepackage[utf8]{inputenc} %
\usepackage[T1]{fontenc}    %
\usepackage[colorlinks=true, linkcolor=blue, citecolor=green!60!black]{hyperref}       %
\usepackage{url}            %
\usepackage{booktabs}       %
\usepackage{amsfonts}       %
\usepackage{microtype}      %
\usepackage{xcolor}         %
\usepackage{nicefrac}
\usepackage[ruled,noend]{algorithm2e}
\usepackage{amssymb,amsthm,amsmath,amssymb,wrapfig,dsfont,authblk,mathrsfs,enumerate}
\usepackage{natbib}
\usepackage{parskip}

\usepackage{booktabs} %
\usepackage[ruled]{algorithm2e} %

\SetAlFnt{\small}
\SetAlCapFnt{\small}
\SetAlCapNameFnt{\small}
\SetAlCapHSkip{0pt}
\IncMargin{-\parindent}

\newtheorem{theorem}{Theorem}[section]
\newtheorem{lemma}[theorem]{Lemma}

\newtheorem{corollary}[theorem]{Corollary}

\newtheorem{definition}[theorem]{Definition}

\SetCommentSty{mycommfont}

\usepackage[nameinlink]{cleveref}
\usepackage{graphicx}
\usepackage{multirow,hhline,tabularx,xcolor,wrapfig,tikz,cancel,nicefrac,mathtools}
\usepackage{thm-restate}
\usetikzlibrary{shapes.geometric}
\usepackage{wrapfig}
\usepackage{bbm}

\newcommand{\1}[1]{\mathbbm{1}\left[#1\right]} %

\newcommand{\cut}{\mathsf{Cut}}
\newcommand{\induced}{\mathsf{Induced}}
\newcommand{\util}{\mathsf{quality}}
\newcommand{\D}{\mathcal{D}} %
\newcommand{\Exp}{\mathsf{Exponential}}
\newcommand{\signalspace}{\mathcal{S}} %
\newcommand{\indicator}[1]{\mathbbm{1}\left[#1\right]}

\newcommand{\eps}{\epsilon}
\newcommand{\Ex}[2]{\operatorname*{\mathbb{E}}_{#1}\left[#2\right]} %

\newcommand{\OPT}{\mathsf{OPT}}

\newcommand{\OPTIC}{\OPT^{\mathsf{stable}}}
\newcommand{\OPTIR}{\OPT^{\mathsf{IR}}}
\newcommand{\barOPT}{\overline{\mathsf{OPT}}}
\newcommand{\barW}{\overline{W}}
\newcommand{\IC}{\mathsf{stable}}
\newcommand{\IR}{\mathsf{IR}}
\newcommand{\PoS}{\mathsf{PoS}}
\newcommand{\Cost}{\mathsf{Cost}}
\newcommand{\LHS}{\mathrm{LHS}}
\newcommand{\poly}{\mathrm{poly}}
\newcommand{\pr}[2]{\Pr_{#1}\left[#2\right]} %
\newcommand{\R}{\mathbb{R}} %
\newcommand{\RHS}{\mathrm{RHS}}
\newcommand{\rmd}{\mathrm{d}}
\newcommand{\Sum}{\mathsf{sum}}
\newcommand{\btheta}{\boldsymbol{\theta}}

\newcommand{\bt}{\boldsymbol{t}}
\newcommand{\bs}{\boldsymbol{s}}
\newcommand{\bu}{\boldsymbol{u}}
\newcommand{\ba}{\boldsymbol{a}}
\newcommand{\DS}{{\mathsf{DS}}}
\newcommand{\DStilde}{\widetilde{\DS}}
\newcommand{\IS}{{\mathsf{IS}}}
\newcommand{\IStilde}{\widetilde{\IS}}

\newcommand{\suchthat}{\text{ s.t. }}
\newcommand{\sym}{\textsf{Sym}}
\newcommand{\unif}{\textsf{Unif}}

\newcommand{\tildephi}{\widetilde{\varphi}}
\newcommand{\tildemu}{\widetilde{\mu}}

\newcommand{\Contrib}{\textsf{Contrib}}
\newcommand{\Numtheta}{\textsf{Num}}
\newcommand{\prior}{\tau}
\newcommand{\ISsize}{\beta}
\newcommand{\vecone}{\mathbf{1}}
\newcommand{\veczero}{\mathbf{0}}
\newcommand{\dualvar}{\boldsymbol{\phi}}

\newcommand{\stepa}[1]{\overset{\rm (a)}{#1}}
\newcommand{\stepb}[1]{\overset{\rm (b)}{#1}}
\newcommand{\stepc}[1]{\overset{\rm (c)}{#1}}
\newcommand{\stepd}[1]{\overset{\rm (d)}{#1}}

\crefname{definition}{definition}{definitions}
\Crefname{definition}{Definition}{Definitions}
\crefname{prop}{proposition}{propositions}
\Crefname{Prop}{Proposition}{Propositions}
\Crefname{cor}{Corollary}{Corollaries}
\crefname{lemma}{Lemma}{Lemmas}
\crefname{section}{Section}{Sections}
\crefname{subsubsubsection}{Section}{Sections}
\crefname{remark}{Remark}{Remarks}
\crefname{figure}{Figure}{Figures}
\crefname{table}{Table}{Tables}
\Crefname{lemma}{Lemma}{Lemmas}
\crefname{theorem}{Theorem}{Theorems}
\Crefname{theorem}{Theorem}{Theorems}
\crefname{algo}{Algorithm}{Algorithms}

\begin{document}
\title{Platforms for Efficient and Incentive-Aware Collaboration}
\author{Nika Haghtalab} %
\author{Mingda Qiao} %
\author{Kunhe Yang} %

\affil{University of California, Berkeley\\{\small\texttt{\{nika,mingda.qiao,kunheyang\}@berkeley.edu}}}
\date{}

\maketitle

\allowdisplaybreaks
\begin{abstract}
    Collaboration is crucial for reaching collective goals. However, its potential for effectiveness is often undermined by the strategic behavior of individual agents --- a fact that is captured by a high Price of Stability (PoS) in recent literature~\citep{blumOneOneAll2021}.  Implicit in the traditional PoS analysis is the assumption that agents have full knowledge of how their tasks relate to one another. 
We offer a new perspective on bringing about efficient collaboration across strategic agents using information design. 
Inspired by the increasingly important role collaboration plays in machine learning (such as platforms for collaborative federated learning and data cooperatives), we propose a framework in which the platform possesses more information about how the agents' tasks relate to each other than the agents themselves. Our results characterize how and to what degree such platforms can leverage their information advantage and steer strategic agents towards efficient collaboration.

Concretely, we consider collaboration networks in which each node represents a task type held by one agent, and each task benefits from contributions made in their inclusive neighborhood of tasks.
 This network structure is known to the agents and the platform. 
 On the other hand, the real location of each agent in the network is known to the platform only --- from the perspective of the agents, their location is determined by a uniformly random permutation.
We employ the framework of private Bayesian persuasion and design two families of persuasive signaling schemes that the platform can use to guarantee a small total workload when agents follow the signal. 
The first family
{aims to achieve}
the minmax optimal approximation ratio compared to the total workload in the optimal collaboration, which is {shown to be $\Theta(\sqrt{n})$} for unit-weight graphs,
{$\Theta(n^{\frac{2}{3}})$} for graphs with edge weights lower bounded by $\Omega(1)$,
and $O(n^{\frac{3}{4}})$ for general weighted graphs. The second family ensures per-instance strict improvement in the total workload compared to scenarios with full information disclosure.
    \end{abstract}
    \clearpage
\tableofcontents
\clearpage

    \section{Introduction}

Collaboration is the cornerstone of modern achievements across various disciplines.
{Effective implementation of collaborative systems has substantially increased what can be accomplished by the limited capabilities of individual agents. For instance, the global collaboration between agencies and institutions under the Genome-Wide Association Studies (GWAS) has enabled the decoding of  genetic foundations of diseases~\citep{bergen2012genome}; collaborations across hundreds of mathematicians have led to resolving longstanding open problems~\citep{steingart2012group}; collaboratively maintained virus signature databases~\citep{OTX} have made technology safer; 
and collaborative federated learning~\citep{mcmahan2017communication} has led to the training of models with superior performance on a large range of applications.

However, for every highly successful collaboration, there are many others that never came to fruition.
This is in part due to the %
inherently strategic nature of participants in a collaboration that can be a} significant barrier to the realization of optimal collaboration. On the one hand, achieving optimal collaboration has been shown to typically require levels of effort from {some} participants that surpass what is individually rational (the amount of work required when working alone) or what is considered {stable} (the amount of work deemed reasonable given others' contributions).
On the other hand, stable collaboration systems where everyone is satisfied with their assigned effort, if exist, often suffer from significant inefficiency.

The above issues are evident even in a simple double-star network shown in \Cref{fig:double-star} where each node represents an agent {and their task}, and each agent{'s task} benefits from both their own and neighboring agents' contributions\footnote{A variant of this structure was used by \citet{blumOneOneAll2021} for establishing a lower bound on the PoS in several abstractions of collaborative federated learning.}.
For collaboration to be feasible, each agent's {task requires at least one unit of contribution.}
Ideally, in the optimal collaboration that minimizes total workload, the two central nodes {take on the entire workload} by each contributing one unit of effort to support all their leaf neighbors. However, when incentives are factored in, both center nodes have the incentive to {unilaterally} reduce their contribution, as they have already received enough support from each other. 
In fact, in any {stable} collaboration {in this network}, the total workload inevitably scales linearly with the number of agents, which implies that the benefits gained from collaboration are merely marginal.

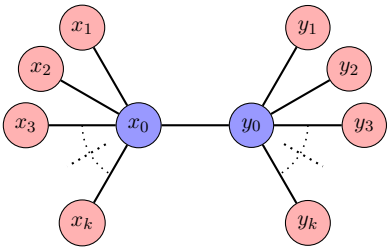
\begin{figure}[!htbp]
    \centering
    \begin{tikzpicture}[scale=0.75, transform shape]
        \def \radius {2cm}
        
        \node[draw, circle, fill = blue!40] at (0,0) (center1) {$x_0$};
        \foreach \i in {1,2,3}{
          \node[draw, circle, fill = red!30] at (\i*30+90:\radius) (left\i) {$x_{\i}$};
          \draw[thick] (center1)--(left\i);
        }
        \node[draw, circle, fill = red!30] at (240:\radius) (leftk) {$x_{k}$};
        \draw[thick] (center1)--(leftk);

        \node[circle] at ({210}:\radius) (aux1) {\phantom{$u_{5}$}};
        \draw[dotted, thick, shorten >=1mm, shorten <=2mm] (center1)--(aux1);

        \draw[dotted, semithick] (180:\radius/2) arc[start angle=180, end angle=240, radius=\radius/2];
        
        \node[draw, circle, fill = blue!40] at (\radius,0) (center2) {$y_0$};
        \foreach \i in {1,2,3}{
            \node[draw, circle, fill = red!30] at ([shift={(-\i*30+90:\radius)}]center2) (right\i) {$y_{\i}$};
          \draw[thick] (center2)--(right\i);
        }
        \node[draw, circle, fill = red!30] at ([shift={(-60:\radius)}]center2) (rightk) {$y_{k}$};
        \draw[thick] (center2)--(rightk);

        \node[circle] at ([shift={(-30:\radius)}]center2) (aux2) {\phantom{$u_{5}$}};
        \draw[dotted, thick, shorten >=1mm, shorten <=2mm] (center2)--(aux2);

        \draw[dotted, semithick] ([shift={(0:\radius/2)}]center2) arc[start angle=0, end angle=-60, radius=\radius/2];
        
        \draw[thick] (center1) -- (center2);
        
        \end{tikzpicture}

        \caption{\footnotesize Double Star Graph. 
        Formally, let $\theta_v$ denote the contribution of node $v$. The social optimum is achieved when $\theta_{x_0}=\theta_{y_0}=1$ with all other $\theta_v=0$.  {\emph{Feasibility} requires that each agent $v$ receives a total contribution of $\theta_v + \sum_{u: \{u,v\}\in E} \theta_u \geq 1$.}
        \emph{Stability} requires no node can unilaterally reduce their contribution without hurting their own feasibility. Following are typical stable solutions: (1) One-sided leaf contribution, where $\theta_v=1$ for $v\in\{x_1, x_2, \ldots, x_k, y_0\}$ and $\theta_v=0$ for the rest, and (2) Full leaf contribution, with $\theta_v=1$ for all leaf nodes  $v\in\{x_1, x_2, \ldots, x_k, y_1, y_2, \ldots, y_k\}$, while $\theta_{x_0}=\theta_{y_0}=0$. Both solutions suffer from a total workload of $\Omega(k)$.
        }

        \label[figure]{fig:double-star}
\end{figure}

This example highlights the inefficiency of collaboration in the presence of strategic agents, which is quantitatively reflected by a high Price of Stability (PoS)~\citep{PoS,blumOneOneAll2021}, a well-established concept in game theory measuring the ratio between the total workload in the best {stable} solution and the socially optimal solution\footnote{
In the original definition, PoS is the ratio between the cost of the best Nash equilibrium and that of the social optimum. In our collaborative context, the Nash equilibrium analog is a stable solution conditioned on being feasible.}.

However, the traditional notion of PoS may be overly pessimistic for analyzing collaboration networks for two primary reasons. First, it assumes that agents have full knowledge of their positions and roles within the collaborative system. This overlooks real-life and practical uncertainties that often exist in collaborations centered on data gathering and usage; for instance, in data cooperatives and federated learning, participants are often uncertain of how their data contribute to the bigger picture, including their specific uses, their broader implications and value, and how they can help solve tasks for other participants. 
Collaborations are often hosted on platforms, yet, platforms lack the leverage to steer the agents toward more efficient collaborative outcomes in the traditional notion of PoS. As a result, the system is severely affected by the misalignment between the strategic agents' incentives and the platform's goal of improving the overall efficiency.

{Inspired by the increasingly important role collaboration plays in machine learning, such as platforms for collaborative federated learning and data cooperatives, we revisit the problem of efficient and stable collaboration under a new framework. To address the aforementioned} issues, it is crucial to recognize that while agents typically have some understanding of the types of tasks in the system and the general structure relating different task types, they often lack precise knowledge about how their specific tasks relate to others.
This precise knowledge is usually exclusive to the platform, which highlights the platform's crucial role in communicating this knowledge to the agents. 
This level of uncertainty challenges the traditional PoS analysis that assumes complete knowledge of one's position and role. 
Yet, it also opens up new possibilities: by strategically distributing information about agents' task types, {not only} platforms can address these challenges but also {they can} create opportunities for efficient collaboration that was previously thought to be impossible.

Reflecting this into the aforementioned double-star structure in \cref{fig:double-star}, {we now consider a scenario where agents are assigned to nodes of this network by a uniformly random permutation. That is, while the network structure is known to the agents, they are unaware of their exact positions} within the graph.
Not knowing whether one's task is represented by a central node or a leaf node
significantly influences agents' decisions and collaboration outcomes.  
This is where the platform's role of strategically sharing information becomes clear.
If the platform withholds all information, every agent will contribute constant effort because they face {maximum and independent uncertainty} about whether {their task is sufficiently covered by contributions of others. This high level of uncertainty will lead agents to contribute an unnecessary amount of effort.
}
On the other extreme, full information disclosure {removes all uncertainty}, leading central nodes to reduce their contributions (as discussed above) and pushing leaf nodes into a less efficient pattern of contribution, thus significantly increasing the PoS.

The crux of steering strategic agents into efficient collaboration lies in \emph{creating the right level of correlated uncertainty}. 
By controlling the flow of information, the platform can {negatively correlate the beliefs of}
central and leaf nodes,
{while maintaining a desirable level of uncertainty}.
{The negative correlation between agents' beliefs ensures that agents would not perform unnecessary amount of work simultaneously. Moreover, the remaining uncertainty (as well as negative correlation) leverages the central agents' fear that their tasks might have not received sufficient contributions from others, thereby encouraging them to increase their effort.
}

This leads us to the main question of our study: Is it possible for the platform to always ensure efficient, {stable} collaboration among strategic agents by controlling the information flow about their types? Specifically, we seek to determine whether a general information structure can be designed to:
\begin{enumerate}[(1)]
    \item Guarantee a sublinear\footnote{Recall that worst-case PoS can be linear in the number of agents, as demonstrated by \cref{fig:double-star}} upper bound on the PoS for all collaboration instances;
    \item Consistently improve the PoS compared to a scenario of full information revelation.
\end{enumerate}

To address these questions, our approach builds on the framework of the private Bayesian persuasion model in information design~\citep{kamenica2011bayesian,arieli2019private}. 
We allow the platform to credibly commit to a signaling scheme, which is a probabilistic mapping from the realization of joint types to private signals that are sent to each agent in the system. These private signals distill the information about joint types into direct recommendations on how much each agent should contribute. Key to this framework is the signaling scheme's {\emph{persuasiveness}} -- it guarantees that agents are incentivized to follow the recommendations. The platform facilitates efficient collaboration by committing to persuasive signaling schemes with low total workload when agents follow the recommendations.

\subsection{Our Contribution}

We initiate the study of improving the efficiency of {stable} collaboration through the lens of information design. We formalize the private Bayesian persuasion problem in a linear collaborative network and characterize the structural properties of persuasive signaling schemes. For {a} special class of {\emph{binary}} signaling schemes, we identify how its persuasiveness relates to graph characteristics such as cuts and induced subgraphs. 

Based on these structural insights, we design two families of signaling schemes that achieve the aforementioned goals (1) and (2), respectively. 
We will use $n$ to denote the number of agents (nodes) in the network, and use benchmarks $\OPT$, $\OPTIR$ and $\OPTIC$ to respectively quantify the total workload under the optimal feasible collaboration as well as the optimal feasible collaboration which also satisfies individual rationality (IR) and stability constraints.

Toward our first goal of achieving a nontrivial approximation ratio to $\OPT$ in the worst case, we design a binary signaling scheme that achieves the approximation ratio of $O(\sqrt{n})$ in any unit-weight graphs. In general weighted graphs, we argue that $\OPTIR$ serves as a more appropriate benchmark than $\OPT$, because $\OPT$ might require some agents to contribute more than their individually rational amount, which rational agents cannot be persuaded to do. However, we show that binary signaling schemes fail to achieve this goal. To address this, we introduce a third signal to achieve the optimal approximation ratios of $O(n^{\frac{2}{3}})$ for graphs with lower bounded edge weights,
and $O(n^{\frac{3}{4}})$ for general weighted graphs. Our results are summarized in \cref*{thm:unit-ub-opt,thm:weighted-ternary-opt}.

\begin{theorem}
    \label[theorem]{thm:unit-ub-opt}
    In any unit-weight graph, there exists a binary signaling scheme that is persuasive and has cost $O(\sqrt{n}\cdot\OPT)$. Moreover, the $O(\sqrt{n})$ approximation ratio is tight for certain graphs. 
\end{theorem}

\begin{theorem}
    \label[theorem]{thm:weighted-ternary-opt}
    In any weighted graph, there exists a ternary signaling scheme that is persuasive and has cost $O(n^{3/4}\cdot(\OPTIR)^{1/2})$.
    In addition, for graphs with all edge weights bounded below by $\delta$, the cost is improved to $O((n\cdot \OPTIR)^{2/3} \cdot \delta^{-1/3})$, with the $O(n^{2/3})$ approximation ratio being tight on a family of graphs with $\delta = 1/2$. 
\end{theorem}

We introduce a different family of signaling schemes that achieve our second goal of strictly improving $\OPTIC$. This approach uses only two signals in unit-weight graphs and at most $n+1$ signals in general weighted graphs. We summarize the results in \cref{thm:unit-improve-ic,thm:weighted-improve-ic}. 
Despite the substantial increase in the signal space required for weighted graphs, we demonstrate the infeasibility of achieving this goal through binary or ternary signaling schemes, and leave the identification of the minimum number of signals needed as an open problem discussed in \Cref{sec:discussion}.

\begin{theorem}
    \label[theorem]{thm:unit-improve-ic}
    In any unit-weight graph, whenever $\OPT<\OPTIC$, there exists a binary signaling scheme that is persuasive and has a cost strictly lower than $\OPTIC$.
\end{theorem}

\begin{theorem}
    \label[theorem]{thm:weighted-improve-ic}
    In any general weighted graph, whenever $\OPTIR<\OPTIC$, there exists a signaling scheme with at most $n+1$ signals that is persuasive and has a cost strictly lower than $\OPTIC$.
\end{theorem}

\subsection{Related Works}
\label{sec:related}

\paragraph{Collaborative federated learning.}
Our work connects to the broader topic of collaborative and federated learning among strategic agents, where incentives and fairness have been extensively explored~\citep{donahueA2021model,donahueB2021optimality,donahueC2021models,blumOneOneAll2021}. Most related to us is the work of \citet{blumOneOneAll2021} {that introduces a formalism for collaborative learning among network of agents and highlights} the inefficiency in collaboration among strategic agents as captured by a high PoS. Our work introduces a novel approach to reducing the PoS under a more realistic scenario where agents face uncertainty regarding the types of their tasks.

Recent studies have adopted a \emph{mechanism design} approach in federated learning, targeting objectives such as incentivizing data contribution and preventing free-riding~\citep{huang2023evaluating,karimireddy2022mechanisms,chen2023mechanism,xu2021gradient}, encouraging participation~\citep{hu2023federated,xu2023fair,cohen2023incentivized,han2023effect},
and promoting fairness~\citep{lin2023fair,lyu2020collaborative}. See recent surveys~\citep{zhan2021survey,tu2022incentive} for a more comprehensive overview.
Additionally, research on federated bandits \citep{wei2023incentivized} explores strategies to motivate agents to share their data. However, these papers primarily rely on either monetary incentives or penalization through reduced model accuracy and usually operate under the condition that agents know how their data or tasks relate to those of others. 
In contrast, we assume that the platform possesses more information about the agents' types but without a direct control over the distribution of data, gradients, or models.

In addition, there is research focusing on incentives and mechanism design in \emph{data markets} where consumers purchase data from the market~\citep{agarwal2019marketplace,agarwal2020towards,jia2019towards}, but they do not focus on incentivizing collaboration among participants.

Another line of literature focuses on the communication cost~\citep{blum2021communication} or sample efficiency of collaborative learning without incentives~\citep{blum2017collaborative,nguyen2018improved,chen2018tight,haghtalab2022demand,awasthi2023open,zhang2023optimal,peng2023sample}.
In contrast, we achieve \emph{both efficiency and incentive-awareness} through the lens of information design.

\paragraph{Information design.}
Bayesian persuasion, a seminal work of \citet{kamenica2011bayesian} rooted in the context of incomplete information games~\citep{aumann1995repeated}, is a canonical model for information design that studies how an informed sender can strategically share information to influence the actions of a single uninformed receiver.
This framework has been expanded to multi-receiver settings~\citep{bergemann2016bayes} and explored through algorithmic perspectives~\citep{dughmi2016algorithmic,dughmi2017algorithmic}, 
supporting its applications in diverse problems. These applications include but are not limited to routing games~\citep{das2017reducing,wu2019information}, coordination games~\citep{wu2019information}, voting~\citep{schnakenberg2015expert,alonso2016persuading,bardhi2018modes,wang2013bayesian,arieli2019private}, and incentivizing exploration~\citep{mansour2020bayesian,mansour2022bayesian,kremer2014implementing}.
We refer the readers to survey papers~\citep{candogan2020information,kamenica2019bayesian} for a comprehensive overview.

The problem we study involves designing signaling schemes in a specific multi-receiver setting, where the literature offers two general approaches: characterizing and selecting \emph{Bayes Correlated Equilibria (BCE)} based on sender utility~\citep{bergemann2013robust,bergemann2016bayes,bergemann2019information,taneva2019information}, and belief-based optimization~\citep{mathevet2020information}.
However, applying these approaches directly in our context is computationally intractable
due to the exponentially large state space.
Indeed, 
finding optimal signal schemes is proved hard even in some simple multi-receiver settings ~\citep{bhaskar2016hardness,rubinstein2017honest,dughmi2017algorithmic}, emphasizing the importance of leveraging our collaborative network's structural properties.

Our work also contributes to the recent line of work on \emph{information design in networks and platforms}.
In the context of social networks, \citet{candogan2017optimal,candogan2019persuasion,egorov2019persuasion} design optimal signaling schemes to maximize overall engagement, while \citet{arieli2019private} studies the setting when the sender's utility is more general and non-additive.
\citet{mathevet2022organized} studies vertical and horizontal information transmission, identifying scenarios where \emph{single-meeting schemes} --- signals shared with and observed by a specific subset of agents --- are optimal.
Our analysis of signaling schemes in the collaboration network also focuses on optimality but diverges in its definition: we assess optimality not by the ability to implement all BCE, as previous studies do, but by achieving an optimal asymptotic approximation ratio compared to the socially optimal solution when agents are non-strategic. To the best of our knowledge, we are the first work that focuses on the perspective of multi-receiver Bayesian persuasion to enhance and encourage collaboration between strategic agents.

In the context of online platforms with strategic agents, \citet{papanastasiou2018crowdsourcing} studies information provision policy for sequentially arriving consumers who choose among alternative products or services,
\citet{bergemann2022calibrated,chen2023bayesian} focuses on maximizing revenues in click-through auctions when the platform has exclusive knowledge about the true click-through rates,
among many other studies in this field.
Our work complements the literature by studying the efficiency of collaboration platforms.

    \subsection{Structure of the paper}
We formally introduce our model in \Cref{sec:model}. In \Cref{sec:technical-overview}, we discuss structural results that characterize the persuasiveness of signaling schemes, and provide a technical overview of our main results. In \Cref{sec:discussion}, we discuss a few open problems and future research directions. In \Cref{sec:unit-weight-upper}, we focus on unit-weight graphs and prove the upper bounds in \Cref{thm:unit-ub-opt,thm:unit-improve-ic}. We prove the lower bound side of \Cref{thm:unit-ub-opt} in \Cref{sec:unit-lower}. For weighted graphs, we provide the upper bounds (\Cref{thm:weighted-ternary-opt,thm:weighted-improve-ic}) in \Cref{sec:weighted-upper}, and the lower bound (second half of \Cref{thm:weighted-ternary-opt}) in \Cref{sec:weighted-lower}.
    
    \section{Model}
\label{sec:model}

\subsection{Collaboration system}
\paragraph{Task-based contribution.}
Let $V$ be the set of task types in the collaboration system, with $|V|=n$. Let $\theta_v\in\R_{\ge0}$ be the effort contributed directly towards solving task $v$, and let $\btheta=(\theta_v)_{v\in V}$ be its vector form. We assume that the total contribution that each type receives is a linear combination of the direct contributions that are put into all tasks in the system, with coefficients {$W_{u,v}$ that quantifies}
the extent to which efforts dedicated to a type-$u$ task help complete a type-$v$ task. We assume that $W_{u,v}=W_{v,u}\in[0,1]$ and $W_{u,u}=1$ for all pairs $(u,v)\in V^2$, i.e., $W$ is an $n\times n$ symmetric matrix with coefficients $W_{u,v}$ and has a diagonal of all ones. For $v\in V$, we use $u_v$ to denote the total contribution entering a type-$v$ task, and let $\bu=(u_v)_{v\in V}$ be its vector form. We then have
\begin{align*}
    u_v(\btheta)=\sum_{v'\in V}W_{v,v'}\theta_{v'},
    \quad\text{and}\quad
    \bu(\btheta)=W \btheta.
\end{align*}
\paragraph{Graph representation.}
The task structure can be equivalently represented as a weighted undirected graph $G=(V, E, w)$, where each vertex is a task type in $V$. Each edge $\{v_1,v_2\}\in E$ represents that efforts towards task $v_2$ positively impact task $v_1$. In other words, $E = \{\{v_1, v_2\}: v_1, v_2 \in V, v_1 \ne v_2, W_{v_1,v_2} > 0\}$. The weight function $w: E \to [0, 1]$ maps each edge $\{v_1, v_2\} \in E$ to $W_{v_1,v_2}$ and reflects the coefficients between tasks.
We also use $N(v)$ to denote the open neighborhood of vertex $v\in V$.

\paragraph{Strategic agents.}
Let there be $n$ agents in the collaboration system. Each agent $i$ is associated with a unique task type $t_i\in V$ in the network. Task types for all agents form a type profile $\bt=(t_1,\ldots,t_n)$, which is a permutation of the types in $V$, i.e., $\bt\in \sym(V)$ with $\sym(V)$ being the symmetric group on $V$. 
Additionally, each agent $i$ can take an action $a_i\in\R_{\ge0}$ {that is the effort they} put into solving their own task. The action profile $\ba$ and task type profile $\bt\in \sym(V)$ naturally induce a type-based direct contribution profile $\btheta(\ba;\bt)=(\theta_v)_{v\in V}$, where $a_i=\theta_{t_i}$, or more succinctly $\btheta(\ba;\bt)=\Pi^{-1}\ba$. In the above equation, $\Pi=(\pi_{i,j})$ is the $n\times n$ permutation matrix corresponding to permutation $\bt$, where $\pi_{i,j}=1$ if $j=t_i$ and $0$ otherwise. We abbreviate $\btheta(\ba; \bt)$ to $\btheta$ when it is clear from context.

For each agent $i$, {the quality of a collaborative solution}, denoted with $\util_i$, is a function of both $\ba$ and $\bt$ and equals the total contribution entering the task of type $t_i$. 
We also use $\util=(\util_1,\cdots,\util_n)$ to denote the agent-based utility vector.
Formally, we have
\[
    \util_i(\ba; \bt)=u_{t_i}(\btheta(\ba; \bt)),
    \quad\text{and}\quad
    \util(\ba; \bt)=\Pi \cdot \bu(\btheta(\ba; \bt)).
\]
Strategic agents have three main concerns about the collaboration outcome: \emph{feasibility, individual rationality (IR), and stability}. 
Given a type profile $\bt$, an action profile $\ba$, and the type-based action profile $\btheta$ they induce, we say that a solution is {feasible} if all agents secure a {quality} goal of $1$, i.e., $\util(\ba;\bt)\ge\vecone$ or equivalently $\bu(\btheta)\ge\vecone$ coordinate-wise --- 
For example, when specializing our model to the context of collaborative learning, $\util_i$ represents the accuracy of the final model on agent $i$'s learning task, and feasibility requires the model to achieve sufficient accuracy across all agents' tasks.
IR is satisfied if participating in the collaboration system is more beneficial for each agent than completing the task independently. Note that a solution satisfies IR iff $\ba\le\vecone$ (or equivalently $\btheta\le\vecone$) coordinate-wise.
This is because the matrix $W$ has an all-one diagonal, implying that completing a task independently would require a unit of effort.
In addition, a feasible solution is stable if, given the contribution of others, no agent can unilaterally reduce their contribution without compromising their own feasibility. Formally, this requires
\[
\forall i,\ a_i=\min\left\{x\ge0\ \left|\ \util_i(x,\ba_{-i};\bt)\ge1\right.\right\}
\iff
\forall v,\ \theta_v=\min\left\{x\ge0\ \left|\ u_v(x,\btheta_{-v})\ge1\right.\right\}.
\]

\subsection{Benchmarks and the price of stability}\label{sec:benchmarks-PoS}
Our benchmarks are the {minima} of the total workload (i.e., $\|\ba\|_1=\|\btheta\|_1$) {among} solutions that satisfy different combinations of the feasibility, IR, and {stability} requirements.
We introduce three benchmarks: $\OPT$ as the optimal workload in feasible collaborations, and $\OPTIR$, $\OPTIC$ as the optimal workload subject to additional IR and stability constraints. Formally, they are solutions to the following programs:
\begin{align*}
    \OPT=\min_{\btheta}\|\btheta\|_1 \quad 
    \suchthat \ &
    \bu(\btheta)\ge\vecone,\ \btheta\ge\veczero;\\
    \OPTIR=\min_{\btheta}\|\btheta\|_1 \quad \suchthat \ &
    \bu(\btheta)\ge\vecone,\ \btheta\ge\veczero,\ \btheta\le\vecone;\\
    \OPTIC=\min_{\btheta}\|\btheta\|_1 \quad \suchthat \ &
    \bu(\btheta)\ge\vecone,\ \btheta\ge\veczero,\ \theta_v=\min\left\{x\ge0\ \left|\ u_v(x,\btheta_{-v})\ge1\right.\right\}\ (\forall v\in V).
    \end{align*}
The \emph{price of stability (PoS)} is formally defined as the gap between the socially optimal solution and the optimal stable solution, i.e.,
$
    \PoS\triangleq{\OPTIC}/{\OPT}.$
\subsection{Information design and signaling schemes}\label{sec:signaling-schemes}
\paragraph{Prior distribution and information asymmetry.}
We now introduce the information design perspective of our model.
Let $\prior$ be the prior distribution from which the type profile $\bt$ is drawn. We model $\prior$ as the uniform distribution over the symmetric group, i.e., $\prior=\unif(\sym(V))$. As a result, for each fixed agent $i$, the marginal distribution on its type $t_i$ is uniform over $V$. 
We assume that the graph structure $G$ and prior distribution $\prior$ are common knowledge, but the realization of the true types $\bt$ is exclusive knowledge held by the platform. The platform uses the private Bayesian persuasion protocol to strategically communicate this exclusive knowledge to the agents.

\paragraph{Private Bayesian persuasion.}
A \emph{signaling scheme}, denoted with $\varphi:\sym(V)\to\Delta(\signalspace^V)$, is a mapping from the task type profile $\bt\in\sym(V)$ to a correlated distribution over private signals $\bs=(s_v)_{v\in V}\in\signalspace^V$, where $\signalspace\subset[0,1]$ is {a finite} space of signals, and each $s_v$ is the signal value sent to the agent with type $v$.\footnote{In the original definition, the signals should be sent to each \emph{agent} rather than each \emph{type}, but the two definitions are equivalent as the platform can reassign the signals to agents according to $\Pi \cdot \bs$.}
The interaction protocol between the platform and the agents is as follows: initially, both the platform and the agents start with a common prior $\prior$, and the platform credibly commits to a signaling scheme $\varphi$, which also becomes common knowledge upon commitment. 
Then, the true types $\bt\sim\prior$ are realized from the prior distribution and observed by the platform. With the knowledge of $\bt$, the platform generates a set of signals $\bs=(s_v)_{v\in V}\sim\varphi(\bt)$, and sends $s_{t_i}$ privately to agent $i$ for all $i\in[n]$, which represents the recommended action to be taken. When an agent $i$ receives the private signal $s_{t_i}$, they form a posterior belief about the type profile $\bt$ and the signals sent to other types $\bs_{-t_i}$, which we denote with $\mu_i$:
\begin{align}
    \mu_i(\bt,\bs_{-t_i}\mid s_{t_i})=\frac{\prior(\bt)\cdot\pr{\varphi(\bt)}{s_{t_i},\bs_{-t_i}}}{
    \sum_{\bt',\bs_{-t_i}'}\prior(\bt')\cdot\pr{\varphi(\bt')}{s_{t_i},\bs_{-t_i}'}
    }.
    \label{eq:posterior}
\end{align}

\paragraph{Persuasiveness.} We say that a signaling scheme is \emph{persuasive} if no agent has the incentive to {unilaterally} deviate from the signal, assuming that all the other agents follow their signals. 
Formally, this requires that for every agent $i$, the following two conditions hold:
\begin{itemize}
    \item (Feasibility) For every $\theta \in \signalspace$, conditioning on receiving a private signal $s_{t_i}=\theta$, taking action $\theta$ is feasible in expectation:
    \[
        \Ex{(\bt,\bs_{-t_i})\sim\mu_i(\cdot\mid \bs_{t_i}=\theta)}{u_{t_i}(\theta,\bs_{-t_i})}\ge1.
    \]
    \item (Stability) For every $\theta \in \signalspace$, conditioning on receiving a private signal of value $\theta$, the agent cannot contribute strictly {less} than $\theta$ {effort} while still {meeting the aforementioned feasibility condition:}
    \[
    \theta=\min\left\{x\ge0\ \left|\ 
    \Ex{(\bt,\bs_{-t_i})\sim\mu_i(\cdot\mid \bs_{t_i}=\theta)}{u_{t_i}(x,\bs_{-t_i})}\ge1
    \right.\right\}.
    \]
    Equivalently, the feasibility condition must be tight whenever $\theta > 0$. 
\end{itemize}
The \emph{cost} of a persuasive signaling scheme $\varphi$ is the expected total contribution assuming that all agents follow the signals. Formally,
\[
    \Cost(\varphi)\triangleq\Ex{\bt\sim\prior,\bs\sim\varphi(\bt)}{\|\bs\|_1}.
\]

\paragraph{Identity-independent signaling schemes.}
Note that in the above formulation, since $\varphi$ is a function of permutations of $V$, the signal sent to each task type could depend on the identity of the agent performing the task. Nevertheless, we claim that it is sufficient to consider signaling schemes that are \emph{identity-independent}: one such that the signal distribution is independent of the realized permutation, i.e., there is a distribution $\D_\varphi\in\Delta(\signalspace^V)$ such that $\ \varphi(\bt)=\D_\varphi$ holds for all $\bt \in \sym(V)$. We prove the following lemma in \Cref{app:identity-independent}.
\begin{lemma}[Identity-independent signaling schemes]
    \label[lemma]{lem:simplify-signal}
    For any persuasive signaling scheme $\varphi$, there exists an identity-independent signaling scheme $\tildephi$ that is persuasive and has $\Cost(\tildephi)=\Cost(\varphi)$.
\end{lemma}

Since identity-independent signaling schemes are equivalent to joint distributions on $\signalspace^V$,
this lemma sets the groundwork for interpreting a signaling scheme as a ``random labeling'' or ``random assignment'' that directly assigns values in $\signalspace$ to types in $V$.
In the rest of the paper, we will adopt this terminology and refer to the signaling process in task structure graphs as ``labeling'' or ``assigning'' randomized values to its vertices. Furthermore, \Cref{lem:simplify-signal} shows that, to lower bound the cost of all persuasive signaling schemes on a particular instance, it suffices to prove such a lower bound against identity-independent schemes.
    
    \section{Technical Overview and Structural Results}
\label{sec:technical-overview}

{This section outlines our approach for proving all of our main results. In \Cref{sec:structural}, we give a sufficient and necessary condition for a signaling scheme to be persuasive and show that for \emph{binary} signaling schemes this condition reduces to a simple inequality involving cuts and induced sub-graphs in the graph representation of the collaboration system. This characterization is useful both for verifying the persuasiveness of the signaling schemes that we design, and for proving lower bounds against persuasive schemes.
With these tools in hand, we
give an overview of our upper bounds for unit-weight graphs (\Cref{sec:overview_unit}) and general weighted graphs (\Cref{sec:weighted-overview}) starting from concrete illustrative examples. We end the section by explaining how we use the LP duality framework to prove our lower bounds.
}

\subsection{Structural results}
\label{sec:structural}
Consider an identity-independent signaling scheme $\D_\varphi \in \Delta(\signalspace^{V})$ on a finite signal space $\signalspace\subset[0,1]$. For each $\theta\in\signalspace$, we introduce the concept of \emph{slack} for receiving signal $\theta$ as follows:
\begin{definition}[Slack]
    \label[definition]{def:slack}
    For a signaling scheme $\D_\varphi \in \Delta(\signalspace^{V})$ and any $\theta\in\signalspace$, 
    we define $\Contrib_{\theta}$ as the expected total contribution entering vertices receiving $\theta$, and $\Numtheta_{\theta}$ as the expected number of vertices receiving $\theta$:
    \begin{align*}
        \Contrib_\theta\triangleq \Ex{\bs \sim \D_\varphi}{\sum_{v \in V}\1{s_v = \theta} {\sum_{v'\in N(v)}W_{v,v'}s_{v'}}};
        \quad
        \Numtheta_\theta\triangleq \Ex{\bs \sim \D_\varphi}{\sum_{v \in V}\1{s_v = \theta}},
    \end{align*}
where $N(v)$ stands for the open neighborhood of $v$.
    The \emph{slack} of receiving signal $\theta$, denoted with $\Delta_\theta$, is the difference between $\Contrib_\theta$ and $\Numtheta_\theta$ scaled by the factor $(1-\theta)$:
    \begin{align*}
        \Delta_\theta\triangleq
        \Contrib_\theta
        - 
        {(1 - \theta) \cdot }
        \Numtheta_\theta.
    \end{align*}
\end{definition}
\begin{lemma}[Persuasiveness of general signaling schemes]
    \label[lemma]{lemma:persuasive-general}
    An identity-independent signaling scheme $\D_\varphi \in \Delta(\signalspace^{V})$ with finite $\signalspace$ is persuasive if and only if the following conditions are met:
    \begin{itemize}
        \item For every $\theta \in \signalspace$ where $\theta>0$, the slack $\Delta_\theta$  must equal zero ($\Delta_\theta=0$).
        \item If the signal space includes zero ($0\in\signalspace$), the slack $\Delta_0$ must be non-negative ($\Delta_0\ge0$).
    \end{itemize}
\end{lemma}
\begin{proof}[Proof sketch]
Recall from the definition of persuasiveness, $\D_\varphi$ is persuasive if and only if $\forall\theta\in\signalspace$,
$\Ex{(\bt,\bs_{-t_i})\sim\mu_i(\cdot\mid \bs_{t_i}=\theta)}{u_{t_i}(\theta,\bs_{-t_i})}\ge1$,
with this inequality being tight for $\theta>0$. 
Therefore, to prove \cref{lemma:persuasive-general}, it suffices to show that 
\begin{align}
    \Ex{(\bt,\bs_{-t_i})\sim\mu_i(\cdot\mid \bs_{t_i}=\theta)}{u_{t_i}(\theta,\bs_{-t_i})}
    =\theta+\frac{\Contrib_\theta}{\Numtheta_\theta}.
    \label{eq:tmp-cond-exp-slack}
\end{align}
This is because $\theta+\frac{\Contrib_\theta}{\Numtheta_\theta} {\geq 1}\Leftrightarrow{\Contrib_\theta}-(1-\theta)\cdot{\Numtheta_\theta}\ge0$ when $\Numtheta_\theta>0$, which is true because it is without loss of generality to assume that every $\theta\in\signalspace$ is realized with a non-zero probability.
We establish \Cref{eq:tmp-cond-exp-slack} by straightforward application of Bayes' rule in \Cref{app:persuasive-general}.
\end{proof}

\paragraph{A special class of binary signals}
Consider a binary scheme $\varphi$ with signal space $\signalspace=\{0,\alpha\}$ for $\alpha\in(0,1]$. 
Recall from \cref{lem:simplify-signal} that $\D_\varphi$ can be interpreted as a random labeling of $V$ using either $0$ or $\alpha$.
It is useful to view the random labeling $\bs$ as first selecting a random subset $S\subseteq V$ from a distribution $\D$ over all subsets of $V$ (namely $2^V$), and then assigning $s_v=\alpha\cdot\1{v\in S}$ for each $v\in V$.
Using this view, we can rewrite $\Contrib_0$ and $\Contrib_\alpha$ using the random subset $S\sim\D$:
\begin{align*}
    \Contrib_\alpha=&\Ex{\bs \sim \D_\varphi}{\sum_{v \in V}\1{s_v = \alpha}\sum_{v' \in N(v)}W_{v,v'}s_{v'}}=\alpha\cdot\Ex{S\sim \D}{\sum_{v\in S}\sum_{v'\in S\setminus\{v\}}W_{v,v'}};\\
    \Contrib_0=&\Ex{\bs \sim \D_\varphi}{\sum_{v \in V}\1{s_v = 0}\sum_{v' \in N(v)}W_{v,v'}s_{v'}}=\alpha\cdot\Ex{S\sim \D}{\sum_{v\not\in S}\sum_{v'\in S}W_{v,v'}}.
\end{align*}

The above two quantities have graph-theoretical interpretation using the cut $(S,V\setminus S)$ and the subgraph induced by $S$. We define their total weight respectively as follows.
\begin{definition}[Weights of Cut and Induced Subgraph]
    \label[definition]{def:cut-induced}
    The weight of the cut between $S$ and its complement $V \setminus S$, denoted $\cut(S, V\setminus S)$, and the weight of the subgraph induced by $S$, denoted $\induced(S)$, are given by
    \begin{align*}
        &\cut(S, V\setminus S)\triangleq\sum_{u\in S}\sum_{v\in(V\setminus S)}W_{u,v};\\
        &\induced(S)\triangleq\sum_{u\in S}\sum_{v\in S}W_{u,v}=|S|+2\sum_{\{u,v\}\subseteq S} W_{u,v}.
    \end{align*}
    Here, $\cut(S, V\setminus S)$ measures the total weight of edges crossing the cut from $S$ to $V \setminus S$, and $\induced(S)$ measures the total weight of all edges within $S$, including self-loops with $W_{v,v}=1$, and counts each undirected internal edge twice.
\end{definition}
 
Leveraging the characteristics of $S$ defined in \Cref{def:cut-induced}, we provide a graph-theoretical characterization of the persuasiveness of a binary signaling scheme in the next lemma. The detailed proof mirrors the arguments used in \cref{lemma:persuasive-general} and is deferred to \Cref{app:persuasive-binary}.

\begin{lemma}[Persuasiveness of binary signaling schemes]
\label[lemma]{lemma:persuasive-binary}
    Let $\D\in\Delta(2^V)$ be a distribution over subsets of $V$. If $\D$ satisfies the following inequality,\footnote{When $\D$ is the degenerate distribution at $V$, the left-hand side is treated as $+\infty$.}
    \begin{align*}
        \frac{\Ex{S \sim \D}{\cut(S, V \setminus S)}}{\Ex{S \sim \D}{|V \setminus S|}} \ge \frac{\Ex{S \sim \D}{\induced(S)}}{\Ex{S \sim \D}{|S|}},
    \end{align*}
    for $\alpha=\frac{\Ex{S \sim \D}{|S|}}{\Ex{S \sim \D}{\induced(S)}}$, there exists an identity-independent persuasive signaling scheme $\D_\varphi\in\Delta\left(\{0,\alpha\}^V\right)$ that assigns $s_v=\alpha\cdot\indicator{v\in S}$ according to a random subset $S\sim\D$.
    The cost of this signaling scheme is $\Cost(\varphi)=\frac{(\Ex{S \sim \D}{|S|})^2}{\Ex{S \sim \D}{\induced(S)}} \le \Ex{S \sim \D}{|S|}$.
\end{lemma}

\begin{corollary}[Cost of the no-information scheme]
    \label[corollary]{cor:degenerate-signal}
    When $\D$ in \Cref{lemma:persuasive-binary} is the degenerate distribution at $V$, the resulting signaling scheme is persuasive for $\alpha=\frac{n}{\induced(V)}$ and incurs a cost of $\frac{n^2}{\induced(V)}$. Using $m=\sum_{\{v_1,v_2\}\in E} W_{v_1,v_2}$ to denote the total edge weights in the graph, the cost simplifies to $\frac{n^2}{n+2m}=O(\frac{n^2}{m})$.
    This scheme, by not differentiating between vertices, essentially conveys no information. Therefore, its cost is equal to the total contribution of agents when no signal is sent.
\end{corollary}

\subsection{Overview of upper bounds in unit-weight graphs}
\label{sec:overview_unit}

We start by sketching the proof of \Cref{thm:unit-ub-opt}, which gives a binary signaling scheme with an $O(\sqrt{n}\cdot\OPT)$ cost. Recall that by \Cref{lem:simplify-signal}, a signaling scheme is equivalent to a randomized labeling of vertices with values between $0$ and $1$.

\paragraph{A signaling scheme for the double-star graph.} Our proof is best illustrated by the double-star example from \Cref{fig:double-star}, on which the signaling scheme simply randomizes between the following two plans for some small $\eps > 0$:
\begin{itemize}
    \item Plan A: Label the two centers with $1 - \eps$. Label the $2k$ leaves with $0$.
    \item Plan B: Label the two centers with $0$. Label the $2k$ leaves with $1 - \eps$.
\end{itemize}

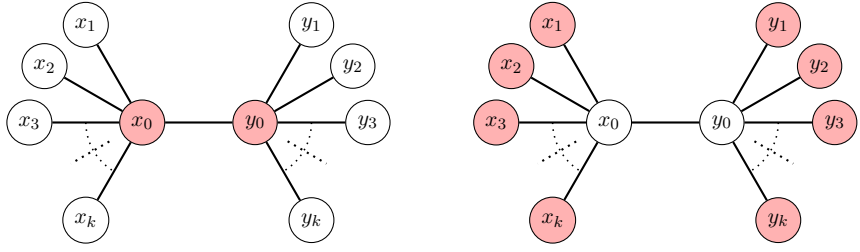
\begin{figure}[!htbp]
    \centering
    \begin{tikzpicture}[scale=0.75, transform shape]
        \def \radius {2cm}
        
        \node[draw, circle, fill = red!30] at (0,0) (center1) {$x_0$};
        \foreach \i in {1,2,3}{
          \node[draw, circle] at (\i*30+90:\radius) (left\i) {$x_{\i}$};
          \draw[thick] (center1)--(left\i);
        }
        \node[draw, circle] at (240:\radius) (leftk) {$x_{k}$};
        \draw[thick] (center1)--(leftk);

        \node[circle] at ({210}:\radius) (aux1) {\phantom{$u_{5}$}};
        \draw[dotted, thick, shorten >=1mm, shorten <=2mm] (center1)--(aux1);

        \draw[dotted, semithick] (180:\radius/2) arc[start angle=180, end angle=240, radius=\radius/2];
        
        \node[draw, circle, fill = red!30] at (\radius,0) (center2) {$y_0$};
        \foreach \i in {1,2,3}{
            \node[draw, circle] at ([shift={(-\i*30+90:\radius)}]center2) (right\i) {$y_{\i}$};
          \draw[thick] (center2)--(right\i);
        }
        \node[draw, circle] at ([shift={(-60:\radius)}]center2) (rightk) {$y_{k}$};
        \draw[thick] (center2)--(rightk);

        \node[circle] at ([shift={(-30:\radius)}]center2) (aux2) {\phantom{$u_{5}$}};
        \draw[dotted, thick, shorten >=1mm, shorten <=2mm] (center2)--(aux2);

        \draw[dotted, semithick] ([shift={(0:\radius/2)}]center2) arc[start angle=0, end angle=-60, radius=\radius/2];
        
        \draw[thick] (center1) -- (center2);
        
    \end{tikzpicture}
    \hspace{24pt}
    \begin{tikzpicture}[scale=0.75, transform shape]
        \def \radius {2cm}
        
        \node[draw, circle] at (0,0) (center1) {$x_0$};
        \foreach \i in {1,2,3}{
          \node[draw, circle, fill = red!30] at (\i*30+90:\radius) (left\i) {$x_{\i}$};
          \draw[thick] (center1)--(left\i);
        }
        \node[draw, circle, fill = red!30] at (240:\radius) (leftk) {$x_{k}$};
        \draw[thick] (center1)--(leftk);

        \node[circle] at ({210}:\radius) (aux1) {\phantom{$u_{5}$}};
        \draw[dotted, thick, shorten >=1mm, shorten <=2mm] (center1)--(aux1);

        \draw[dotted, semithick] (180:\radius/2) arc[start angle=180, end angle=240, radius=\radius/2];
        
        \node[draw, circle] at (\radius,0) (center2) {$y_0$};
        \foreach \i in {1,2,3}{
            \node[draw, circle, fill = red!30] at ([shift={(-\i*30+90:\radius)}]center2) (right\i) {$y_{\i}$};
          \draw[thick] (center2)--(right\i);
        }
        \node[draw, circle, fill = red!30] at ([shift={(-60:\radius)}]center2) (rightk) {$y_{k}$};
        \draw[thick] (center2)--(rightk);

        \node[circle] at ([shift={(-30:\radius)}]center2) (aux2) {\phantom{$u_{5}$}};
        \draw[dotted, thick, shorten >=1mm, shorten <=2mm] (center2)--(aux2);

        \draw[dotted, semithick] ([shift={(0:\radius/2)}]center2) arc[start angle=0, end angle=-60, radius=\radius/2];
        
        \draw[thick] (center1) -- (center2);
        
        \end{tikzpicture}

        \caption{Plan~A (left) and Plan~B (right) for the double-star graph. The shaded vertices are labeled with $1 - \eps$. The empty vertices are labeled with $0$.}

        \label[figure]{fig:binary-scheme}
\end{figure}

Plan~A is almost the same as the socially optimal solution (without the stability constraint), in which each center contributes a unit amount. While Plan~A has a low cost, we should not always follow it, for two different reasons. First, it violates the stability condition: conditioning on receiving signal $1 - \eps$, the agent knows for sure that one of their neighbors will play $1 - \eps$, and thus has an incentive to deviate and play a much lower value. Second, it violates feasibility: conditioning on receiving $0$, the agent expects a total contribution of $1 - \eps$ from their neighbors, so following the signal would not satisfy their demand. On the other hand, while Plan~B appears extremely inefficient at first glance, it does remedy both issues of Plan~A by reducing the slack when signal $1 - \eps$ is received, and increasing the total contributions from neighbors when the signal is $0$.

With the notation from \Cref{def:slack} (i.e., $\Delta_\theta$ denotes the amount of slack corresponding to signal $\theta$), always following Plan~A gives
\[
    \Delta_{1-\eps} = 2\cdot (1 - \eps) - \eps\cdot 2 = +\Theta(1), \quad \Delta_0 = 2k(1 - \eps) - 1\cdot 2k = -\Theta(k\eps),
\]
whereas following Plan~B gives
\[
    \Delta_{1-\eps} = 0 - \eps\cdot 2k = -\Theta(k\eps), \quad \Delta_0 = 2k(1-\eps) - 1\cdot 2 = +\Theta(k).
\]
A simple calculation shows that, for some $\eps, p = \Theta(1/\sqrt{k})$, we can ensure $\Delta_{1-\eps} = 0$ and $\Delta_0 \ge 0$ by following Plan~A with probability $1 - p$ and following Plan~B with probability $p$. The resulting cost would then be $(1 - p)\cdot 2\cdot (1 - \eps) + p\cdot 2k\cdot (1 - \eps) = O(\sqrt{k}) = O(\sqrt{n})$ as desired.

\paragraph{A binary signal perspective.} To generalize this result to all unit-weight graphs, it is helpful to view the signaling scheme above through the lens of \Cref{lemma:persuasive-binary}, our characterization of persuasive binary signaling schemes. Recall that the lemma states that a distribution $\D \in \Delta(2^V)$ gives a persuasive binary signaling scheme if the following inequality holds:
\begin{equation}\label{eq:binary-signal-persuasiveness}
    \frac{\Ex{S \sim \D}{\cut(S, V \setminus S)}}{\Ex{S \sim \D}{|V \setminus S|}} \ge \frac{\Ex{S \sim \D}{\induced(S)}}{\Ex{S \sim \D}{|S|}}.
\end{equation}
Furthermore, the cost of the scheme is at most $\Ex{S \sim \D}{|S|}$.

In the double-star example, let $\DS$ be the set of the two centers, and $\IS$ be the set of the $2k$ leaves. (The names are justified since $\DS$ is a dominating set of the graph, and $\IS$ is an independent set.) When $\D$ is the degenerate distribution at $\DS$, \Cref{eq:binary-signal-persuasiveness} reduces to
$    \frac{2k}{2k} \ge \frac{4}{2},
$
which does not hold. In contrast, when $\D$ is the degenerate distribution at $\IS$, \Cref{eq:binary-signal-persuasiveness} gives
$
    \frac{2k}{2} \ge \frac{2k}{2k},
$
which not only holds, but holds with a large margin! Again, it follows from an elementary calculation that, \Cref{eq:binary-signal-persuasiveness} can be satisfied by setting $\D(\DS) = 1 - p$ and $\D(\IS) = p$ for some $p = \Theta(1/\sqrt{k})$, and the resulting cost is, as expected, bounded by $\Ex{S \sim \D}{|S|} = (1 - p)\cdot|\DS| + p\cdot|\IS| = O(\sqrt{n})$.

\paragraph{Signaling scheme for general unit-weight graphs.} In general, we choose $\DS$ as a minimum dominating set of $G$, and choose $\IS$ as a maximal independent set of the sub-graph induced by $V \setminus \DS$. As in the double-star graph, $\DS$ gives a good approximation of the socially optimal cost $\OPT$, but does not guarantee persuasiveness. Concretely, $|\DS| / \OPT$ is upper bounded by the integrality gap of the LP relaxation of the dominating set problem, which is $O(\log n)$ (e.g., \citep[Section 1.7]{WS11}). To see that $\DS$ might not be persuasive, note that the left-hand side of \Cref{eq:binary-signal-persuasiveness}, $\frac{\cut(\DS, V \setminus \DS)}{|V \setminus \DS|}$, can be as small as $1$, while the right-hand side $\frac{\induced(\DS)}{|\DS|}$ can be as large as $|\DS|$, when the induced sub-graph of $\DS$ is a clique.

Once again, the other set $\IS$ comes to the rescue of persuasiveness. By the independence of $\IS$, the right-hand side of \Cref{eq:binary-signal-persuasiveness} gives $\frac{\induced(\IS)}{|\IS|} = 1$. For the left-hand side, we note that
\[
    \cut(\IS, V \setminus \IS)
=   \cut(\IS, \DS) + \cut(\IS, V \setminus (\IS \cup \DS))
\ge |\IS| + |V \setminus (\IS \cup \DS)|
=   n - |\DS|,
\]
\begin{wrapfigure}{r}{0.5\textwidth}
    \begin{center}
        \begin{tikzpicture}[scale=0.75, transform shape]
    \def \radius {0.75cm}
    \def \isstart {1.4cm}
    \def \dsdx {0.22cm}
    \def \dsdy {0.6cm}
    \draw[rotate=75, dashed] (-2,0) circle(2cm and 1cm);
    \draw[dashed] (2.5,-4) circle(2cm and 1cm);
    \draw[dashed] (3,-1.2) circle(1.5cm and 1cm);
    \node at (3,-0.8) (others) {$V\setminus(\DS\cup\IS)$};
    \node[draw, circle, fill=red!30] at (\isstart+3*\radius, -4) (IS0){};
    \node[draw, circle, fill=red!30] at (\isstart+2*\radius, -4) (IS1){};
    \node[draw, circle, fill=red!30] at (\isstart+\radius, -4) (IS2){};
    \node[draw, circle, fill=red!30] at (\isstart, -4) (IS3){};
    \node at (2.4,-4.7) (DS) {$\IS$};
    \node[draw, circle, fill=blue!40] at (-0.75cm, -2.5cm) (DS1){};
    \node[draw, circle, fill=blue!40] at (-0.75cm+\dsdx, -2.5cm+\dsdy) (DS2){};
    \node[draw, circle, fill=blue!40] at (-0.75cm+2*\dsdx, -2.5cm+2*\dsdy) (DS3){};
    \node at (-0.9,-3.5) (DS) {$\DS$};
    \draw[thick, color=blue!60] (IS3) -- (DS1);
    \draw[thick, color=blue!60] (IS2) -- (DS2);
    \draw[thick, color=blue!60] (IS1) -- (DS2);
    \draw[thick, color=blue!60] (IS0) -- (DS3);
    \node[draw, circle] at (3,-1.2) (other1) {};
    \node[draw, circle] at (2.4,-1.5) (other2) {};
    \node[draw, circle] at (3.7,-1.5) (other3) {};
    \draw[thick, color=red!60] (IS3) -- (other2);
    \draw[thick, color=red!60] (IS0) -- (other3);
    \draw[thick, color=red!60] (IS1) -- (other1);
    \end{tikzpicture}
    \end{center}
    \caption{\footnotesize $\cut(\IS,V\setminus \IS)$. Every vertex in $\IS$ has a neighbor in $\DS$; every vertex in $V \setminus (\DS \cup \IS)$ has a neighbor in $\IS$.}
    \label[figure]{fig:DS-IS-example}
\end{wrapfigure}
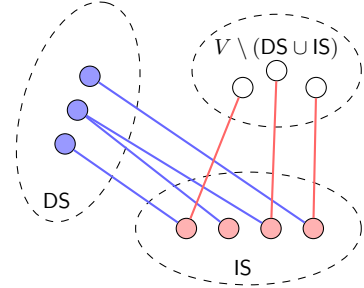
The first step holds since $\DS$ and $V \setminus (\IS \cup \DS)$ partitions $V \setminus \IS$ (see \Cref{fig:DS-IS-example} for a pictorial illustration). In the second step, $\cut(\IS, \DS) \ge |\IS|$ since $\DS$ is a dominating set of $G$, which guarantees that every vertex in $\IS \subseteq V \setminus \DS$ has a neighbor in $\DS$. Meanwhile, the maximality of $\IS$ implies that every vertex in $V \setminus (\DS \cup \IS)$ has a neighbor in $\IS$. This gives $\cut(\IS, V \setminus (\IS \cup \DS)) \ge |V \setminus (\IS \cup \DS)|$. Therefore, the left-hand side of \Cref{eq:binary-signal-persuasiveness}, $\frac{\cut(\IS, V \setminus \IS)}{|V \setminus \IS|}$, is at least $\frac{n - |\DS|}{n - |\IS|}$, which is typically strictly larger than $1$. This shows that \Cref{eq:binary-signal-persuasiveness} holds with a margin for $\IS$.

By carefully randomizing between $\DS$ and $\IS$, we obtain a persuasive binary signaling scheme with cost $O(\sqrt{n})\cdot|\DS|$, which is an $O(\sqrt{n}\cdot\log n)$-approximation of $\OPT$. To shave this extra $\log n$ factor, our actual proof replaces $\DS$ with a distribution over sets obtained from an independent rounding of the optimal solution. This essentially preserves all the desirable properties of $\DS$, while reducing the cost from $|\DS|$ to $\OPT$.

\paragraph{Strict improvement upon $\OPTIC$.} We sketch the proof of \Cref{thm:unit-improve-ic}, which states that there is a binary signaling scheme with a strictly lower cost than $\OPTIC$, whenever $\OPT < \OPTIC$. A key ingredient of the proof is the construction of a binary scheme with cost exactly $\OPTIC$, which is based on the following structural result regarding stable solutions in unit-weight graphs. 

\begin{lemma}\label[lemma]{lemma:IC-solution-decomp}
    Let $\btheta \in [0, 1]^V$ be an arbitrary stable solution on a unit-weight graph $G = (V, E)$. There exists a distribution $\D \in \Delta(2^V)$ supported over the independent sets of $G$, such that for every $v \in V$, 
$\pr{S \sim \D}{v \in S} = \theta_v$.
\end{lemma}

We prove \Cref{lemma:IC-solution-decomp} in \Cref{sec:unit-weight-upper} %
by rounding the fractional solution $\btheta$ (which satisfies stability) into an integral solution (i.e., a subset of $V$), via a rounding scheme known as \emph{competing exponential clocks}. Given the lemma, it then follows quite easily that such a distribution $\D$ satisfies \Cref{eq:binary-signal-persuasiveness}, and gives a binary signaling scheme of cost exactly $\OPTIC$. To see this, note that since every set in the support of $\D$ is an independent set, we have $\Ex{S \sim \D}{\induced(S)} = \Ex{S \sim \D}{|S|}$, i.e., the right-hand side of \Cref{eq:binary-signal-persuasiveness} is $1$. Then, using the assumption that $\btheta$ is a stable solution, along with the fact that the marginal of $\D$ exactly matches $\btheta$, we can prove that $\Ex{S \sim \D}{\cut(S, V\setminus S)} \ge n - \|\btheta\|_1 = \Ex{S \sim \D}{|V \setminus S|}$. %
This lower bounds the left-hand side of \Cref{eq:binary-signal-persuasiveness} by $1$.

The other ingredient of our proof is an argument based on LP duality, which shows that, whenever $\OPT < \OPTIC$, there exists an optimal stable solution $\btheta \in [0, 1]^V$ such that the rounding of $\btheta$ actually satisfies \Cref{eq:binary-signal-persuasiveness} with a positive margin. This allows us to randomize between this scheme (with cost $\OPTIC$) and an independent rounding of the optimal solution (with cost $\OPT < \OPTIC$), and achieve a cost strictly below $\OPTIC$.

\subsection{Overview of upper bounds in weighted graphs} \label{sec:weighted-overview}
Recall that for unit-weight graphs, our binary signaling scheme is based on sending a positive signal to either $\DS$ (a minimum dominating set) or $\IS$ (a maximal independent set of the sub-graph induced by $V \setminus \DS$), each with a carefully chosen probability. Sending signal $1$ to the vertices in $\DS$ gives a low cost, guarantees feasibility, but might not be stable. The purpose of the set $\IS$ is to reduce the slack introduced by $\DS$.

What goes wrong when we repeat this strategy on weighted graphs? The argument about $\DS$ generalizes easily. On unit-weight graphs, a dominating set is equivalent to an integral solution, in which the contribution from each vertex is either $0$ or $1$. Naturally, for weighted graphs, we choose $\DS$ as the set corresponding to the optimal integral solution. Towards satisfying \Cref{eq:binary-signal-persuasiveness}, the feasibility of the integral solution guarantees $\cut(\DS, V \setminus \DS) \ge |V \setminus \DS|$, i.e., the left-hand side of \Cref{eq:binary-signal-persuasiveness} is lower bounded by $1$. Under mild assumptions on the edge weights, we can still show that $|\DS|$ is not much larger than $\OPTIR$, the optimal fractional solution subject to IR.

Unfortunately, as we demonstrate in the example below, it might be impossible to find an analogue of $\IS$ for our purpose.

\paragraph{A concrete example where binary schemes fail.} The graph is shown in \Cref{fig:binary-schemes-fail}. Every edge in the graph has a weight of $1/2$. The graph contains two ``centers'' and $k$ triangles (consisting of $3k$ ``leaves''). Each center is connected to {all} other {vertices} in the graph.

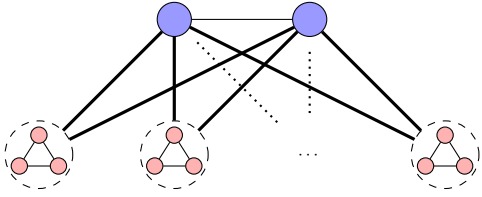
\begin{figure}
  \begin{center}
    \begin{tikzpicture}[scale=0.6, transform shape]
        \def \radius {1.5cm}
        
        \node[draw, circle, fill = blue!40] at (-1*\radius,0) (center1) {\phantom{$c_1$}};
        \node[draw, circle, fill = blue!40] at (1*\radius, 0) (center2) {\phantom{$c_2$}};
        \draw (center1) -- (center2);

        \node[draw, circle, minimum size = \radius, dashed] at (- 3*\radius, -2 * \radius) (triangle1){};
        \draw[very thick] (center1) -- (triangle1);
        \draw[very thick] (center2) -- (triangle1);

        \node[draw, circle, fill = red!30] at (-3*\radius+0*\radius, -2*\radius+0.289*\radius) (leaf11) {};
        \node[draw, circle, fill = red!30] at (-3*\radius+0.289*\radius, -2*\radius-0.167*\radius) (leaf12) {};
        \node[draw, circle, fill = red!30] at (-3*\radius-0.289*\radius, -2*\radius-0.167*\radius) (leaf13) {};
        \draw (leaf11) -- (leaf12);
        \draw (leaf12) -- (leaf13);
        \draw (leaf13) -- (leaf11);

        \node[draw, circle, minimum size = \radius, dashed] at (-1*\radius, -2 * \radius) (triangle2){};
        \draw[very thick] (center1) -- (triangle2);
        \draw[very thick] (center2) -- (triangle2);

        \node[draw, circle, fill = red!30] at (-1*\radius+0*\radius, -2*\radius+0.289*\radius) (leaf21) {};
        \node[draw, circle, fill = red!30] at (-1*\radius+0.289*\radius, -2*\radius-0.167*\radius) (leaf22) {};
        \node[draw, circle, fill = red!30] at (-1*\radius-0.289*\radius, -2*\radius-0.167*\radius) (leaf23) {};
        \draw (leaf21) -- (leaf22);
        \draw (leaf22) -- (leaf23);
        \draw (leaf23) -- (leaf21);

        \node[circle, minimum size = \radius] at (1*\radius, -2 * \radius) (triangle3){$\cdots$};
        \draw[dotted, thick, shorten >=1mm, shorten <=2mm] (center1)--(triangle3);
        \draw[dotted, thick, shorten >=1mm, shorten <=2mm] (center2)--(triangle3);

        \node[draw, circle, minimum size = \radius, dashed] at (3*\radius, -2 * \radius) (trianglek){};
        \draw[very thick] (center1) -- (trianglek);
        \draw[very thick] (center2) -- (trianglek);

        \node[draw, circle, fill = red!30] at (3*\radius+0*\radius, -2*\radius+0.289*\radius) (leafk1) {};
        \node[draw, circle, fill = red!30] at (3*\radius+0.289*\radius, -2*\radius-0.167*\radius) (leafk2) {};
        \node[draw, circle, fill = red!30] at (3*\radius-0.289*\radius, -2*\radius-0.167*\radius) (leafk3) {};
        \draw (leafk1) -- (leafk2);
        \draw (leafk2) -- (leafk3);
        \draw (leafk3) -- (leafk1);
        
        \end{tikzpicture}
  \end{center}
  \caption{\footnotesize A weighted graph on which all binary schemes fail. There are $k$ triangles in total. Each thick line indicates the edges between a center and all the three vertices in a triangle. Every edge is of weight $1/2$.}
        
        \label[figure]{fig:binary-schemes-fail}
\end{figure}

On this graph, the optimum subject to IR, $\OPTIR = 2$, is achieved when each center plays $1$. As discussed earlier, we pick $\DS$ as the set of the two centers. If we, as in unit-weight graphs, pick $\IS$ as a maximal independent set of the induced sub-graph of $V \setminus \DS$, $\IS$ would contain exactly one vertex in each triangle. However, it can be verified that no distribution $\D$ supported over $\{\DS, \IS\}$ satisfies \Cref{eq:binary-signal-persuasiveness}. Therefore, we cannot derive a persuasive binary signaling scheme following the previous approach.

In fact, this is not even a problem with this specific choice of $\DS$ and $\IS$ --- in \Cref{thm:weighted-binary-signals-fail}, we show that \emph{every} persuasive binary signaling scheme has a cost of $\Omega(n)$ on the instance above!

\paragraph{Solution: a third signal.} While binary schemes are insufficient for weighted graphs, a minimal fix --- adding a third signal --- would solve the problem. Concretely, for the graph in \Cref{fig:binary-schemes-fail}, there is a non-trivial signaling scheme via randomizing among the three plans below:
\begin{itemize}
    \item Plan A: Label the centers with $1 - \eps$. Label the leaves with $0$.
    \item Plan B: Label one vertex from each triangle with $1 - \eps$. Label all other vertices with $0$.
    \item Plan C: Label the leaves with $\alpha$. Label the centers with $0$.
\end{itemize}
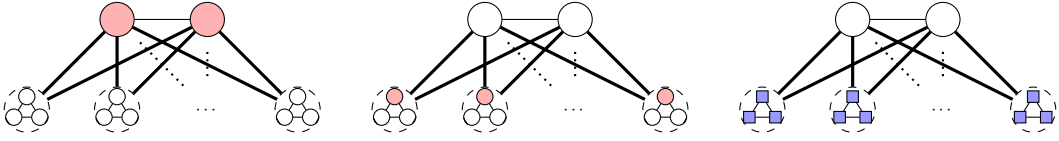
\begin{figure}[!htbp]
    \centering
    \begin{tikzpicture}[scale=0.6, transform shape]
        \def \radius {1cm}
        
        \node[draw, circle, fill = red!30] at (-1*\radius,0) (center1) {\phantom{$c_1$}};
        \node[draw, circle, fill = red!30] at (1*\radius, 0) (center2) {\phantom{$c_2$}};
        \draw (center1) -- (center2);

        \node[draw, circle, minimum size = \radius, dashed] at (- 3*\radius, -2 * \radius) (triangle1){};
        \draw[very thick] (center1) -- (triangle1);
        \draw[very thick] (center2) -- (triangle1);

        \node[draw, circle] at (-3*\radius+0*\radius, -2*\radius+0.289*\radius) (leaf11) {};
        \node[draw, circle] at (-3*\radius+0.289*\radius, -2*\radius-0.167*\radius) (leaf12) {};
        \node[draw, circle] at (-3*\radius-0.289*\radius, -2*\radius-0.167*\radius) (leaf13) {};
        \draw (leaf11) -- (leaf12);
        \draw (leaf12) -- (leaf13);
        \draw (leaf13) -- (leaf11);

        \node[draw, circle, minimum size = \radius, dashed] at (-1*\radius, -2 * \radius) (triangle2){};
        \draw[very thick] (center1) -- (triangle2);
        \draw[very thick] (center2) -- (triangle2);

        \node[draw, circle] at (-1*\radius+0*\radius, -2*\radius+0.289*\radius) (leaf21) {};
        \node[draw, circle] at (-1*\radius+0.289*\radius, -2*\radius-0.167*\radius) (leaf22) {};
        \node[draw, circle] at (-1*\radius-0.289*\radius, -2*\radius-0.167*\radius) (leaf23) {};
        \draw (leaf21) -- (leaf22);
        \draw (leaf22) -- (leaf23);
        \draw (leaf23) -- (leaf21);

        \node[circle, minimum size = \radius] at (1*\radius, -2 * \radius) (triangle3){$\cdots$};
        \draw[dotted, thick, shorten >=1mm, shorten <=2mm] (center1)--(triangle3);
        \draw[dotted, thick, shorten >=1mm, shorten <=2mm] (center2)--(triangle3);

        \node[draw, circle, minimum size = \radius, dashed] at (3*\radius, -2 * \radius) (trianglek){};
        \draw[very thick] (center1) -- (trianglek);
        \draw[very thick] (center2) -- (trianglek);

        \node[draw, circle] at (3*\radius+0*\radius, -2*\radius+0.289*\radius) (leafk1) {};
        \node[draw, circle] at (3*\radius+0.289*\radius, -2*\radius-0.167*\radius) (leafk2) {};
        \node[draw, circle] at (3*\radius-0.289*\radius, -2*\radius-0.167*\radius) (leafk3) {};
        \draw (leafk1) -- (leafk2);
        \draw (leafk2) -- (leafk3);
        \draw (leafk3) -- (leafk1);
        
        \end{tikzpicture}
    \hspace{12pt}
    \begin{tikzpicture}[scale=0.6, transform shape]
        \def \radius {1cm}
        
        \node[draw, circle] at (-1*\radius,0) (center1) {\phantom{$c_1$}};
        \node[draw, circle] at (1*\radius, 0) (center2) {\phantom{$c_2$}};
        \draw (center1) -- (center2);

        \node[draw, circle, minimum size = \radius, dashed] at (- 3*\radius, -2 * \radius) (triangle1){};
        \draw[very thick] (center1) -- (triangle1);
        \draw[very thick] (center2) -- (triangle1);

        \node[draw, circle, fill = red!30] at (-3*\radius+0*\radius, -2*\radius+0.289*\radius) (leaf11) {};
        \node[draw, circle] at (-3*\radius+0.289*\radius, -2*\radius-0.167*\radius) (leaf12) {};
        \node[draw, circle] at (-3*\radius-0.289*\radius, -2*\radius-0.167*\radius) (leaf13) {};
        \draw (leaf11) -- (leaf12);
        \draw (leaf12) -- (leaf13);
        \draw (leaf13) -- (leaf11);

        \node[draw, circle, minimum size = \radius, dashed] at (-1*\radius, -2 * \radius) (triangle2){};
        \draw[very thick] (center1) -- (triangle2);
        \draw[very thick] (center2) -- (triangle2);

        \node[draw, circle, fill = red!30] at (-1*\radius+0*\radius, -2*\radius+0.289*\radius) (leaf21) {};
        \node[draw, circle] at (-1*\radius+0.289*\radius, -2*\radius-0.167*\radius) (leaf22) {};
        \node[draw, circle] at (-1*\radius-0.289*\radius, -2*\radius-0.167*\radius) (leaf23) {};
        \draw (leaf21) -- (leaf22);
        \draw (leaf22) -- (leaf23);
        \draw (leaf23) -- (leaf21);

        \node[circle, minimum size = \radius] at (1*\radius, -2 * \radius) (triangle3){$\cdots$};
        \draw[dotted, thick, shorten >=1mm, shorten <=2mm] (center1)--(triangle3);
        \draw[dotted, thick, shorten >=1mm, shorten <=2mm] (center2)--(triangle3);

        \node[draw, circle, minimum size = \radius, dashed] at (3*\radius, -2 * \radius) (trianglek){};
        \draw[very thick] (center1) -- (trianglek);
        \draw[very thick] (center2) -- (trianglek);

        \node[draw, circle, fill = red!30] at (3*\radius+0*\radius, -2*\radius+0.289*\radius) (leafk1) {};
        \node[draw, circle] at (3*\radius+0.289*\radius, -2*\radius-0.167*\radius) (leafk2) {};
        \node[draw, circle] at (3*\radius-0.289*\radius, -2*\radius-0.167*\radius) (leafk3) {};
        \draw (leafk1) -- (leafk2);
        \draw (leafk2) -- (leafk3);
        \draw (leafk3) -- (leafk1);
        
        \end{tikzpicture}
    \hspace{12pt}
    \begin{tikzpicture}[scale=0.6, transform shape]
        \def \radius {1cm}
        
        \node[draw, circle] at (-1*\radius,0) (center1) {\phantom{$c_1$}};
        \node[draw, circle] at (1*\radius, 0) (center2) {\phantom{$c_2$}};
        \draw (center1) -- (center2);

        \node[draw, circle, minimum size = \radius, dashed] at (- 3*\radius, -2 * \radius) (triangle1){};
        \draw[very thick] (center1) -- (triangle1);
        \draw[very thick] (center2) -- (triangle1);

        \node[draw, rectangle, fill = blue!40] at (-3*\radius+0*\radius, -2*\radius+0.289*\radius) (leaf11) {};
        \node[draw, rectangle, fill = blue!40] at (-3*\radius+0.289*\radius, -2*\radius-0.167*\radius) (leaf12) {};
        \node[draw, rectangle, fill = blue!40] at (-3*\radius-0.289*\radius, -2*\radius-0.167*\radius) (leaf13) {};
        \draw (leaf11) -- (leaf12);
        \draw (leaf12) -- (leaf13);
        \draw (leaf13) -- (leaf11);

        \node[draw, circle, minimum size = \radius, dashed] at (-1*\radius, -2 * \radius) (triangle2){};
        \draw[very thick] (center1) -- (triangle2);
        \draw[very thick] (center2) -- (triangle2);

        \node[draw, rectangle, fill = blue!40] at (-1*\radius+0*\radius, -2*\radius+0.289*\radius) (leaf21) {};
        \node[draw, rectangle, fill = blue!40] at (-1*\radius+0.289*\radius, -2*\radius-0.167*\radius) (leaf22) {};
        \node[draw, rectangle, fill = blue!40] at (-1*\radius-0.289*\radius, -2*\radius-0.167*\radius) (leaf23) {};
        \draw (leaf21) -- (leaf22);
        \draw (leaf22) -- (leaf23);
        \draw (leaf23) -- (leaf21);

        \node[circle, minimum size = \radius, dashed] at (1*\radius, -2 * \radius) (triangle3){$\cdots$};
        \draw[dotted, thick, shorten >=1mm, shorten <=2mm] (center1)--(triangle3);
        \draw[dotted, thick, shorten >=1mm, shorten <=2mm] (center2)--(triangle3);

        \node[draw, circle, minimum size = \radius, dashed] at (3*\radius, -2 * \radius) (trianglek){};
        \draw[very thick] (center1) -- (trianglek);
        \draw[very thick] (center2) -- (trianglek);

        \node[draw, rectangle, fill = blue!40] at (3*\radius+0*\radius, -2*\radius+0.289*\radius) (leafk1) {};
        \node[draw, rectangle, fill = blue!40] at (3*\radius+0.289*\radius, -2*\radius-0.167*\radius) (leafk2) {};
        \node[draw, rectangle, fill = blue!40] at (3*\radius-0.289*\radius, -2*\radius-0.167*\radius) (leafk3) {};
        \draw (leafk1) -- (leafk2);
        \draw (leafk2) -- (leafk3);
        \draw (leafk3) -- (leafk1);
        
        \end{tikzpicture}

        \label[figure]{fig:ternary-scheme}
        \caption{\footnotesize The three plans (left: Plan A; middle: Plan B; right: Plan C) of signaling for the graph from \Cref{fig:binary-schemes-fail}. Shaded circle vertices are labeled with $1 - \eps$. Empty circle vertices are labeled with $0$. Shaded square vertices are labeled with $\alpha$.}
\end{figure}

To ensure persuasiveness, we go back to the characterization in \Cref{lemma:persuasive-general}. The three plans, when we follow each of them alone, give
\begin{itemize}
    \item Plan A: $\Delta_{1 - \eps} = (1 - \eps) - \eps \cdot 2 = +\Theta(1)$ and  $\Delta_0 = 3k(1-\eps) - 1\cdot 3k =  -\Theta(k\eps)$.
    \item Plan B: $\Delta_{1 - \eps} = 0 - \eps\cdot k = -\Theta(k\eps)$ and $\Delta_0 = 2k(1-\eps) - 1\cdot (2k+2) = -\Theta(1 + k\eps)$.
    \item Plan C: $\Delta_0 = 3k\alpha - 1\cdot 2 = 3k\alpha - 2$ and $\Delta_\alpha = 3k\alpha - (1 - \alpha)\cdot 3k = 6k\alpha - 3k$.
\end{itemize}
To satisfy the constraint $\Delta_{\alpha} = 0$, we set $\alpha = 1/2$. The effect of Plan~C is then simplified to
\[
    \Delta_0 = +\Theta(k)
\quad \text{and} \quad
    \Delta_{1/2} = 0.
\]

In hindsight, the introduction of Plan~C and the third signal $\alpha$ is natural. Our previous strategy, which corresponds to randomizing between only Plans A~and~B, cannot guarantee both $\Delta_{1-\eps}=0$ and $\Delta_0 \ge 0$ --- enforcing $\Delta_{1-\eps} = 0$ inevitably results in $\Delta_0 < 0$. The purpose of Plan~C is then to bring $\Delta_0$ back to $0$, without introducing additional slacks.

Formally, it can be verified that, for some $\eps, p, q = \Theta(1/\sqrt{k})$, following the three plans with probability $1 - p - q$, $p$, and $q$ respectively gives a persuasive scheme. The resulting cost is given by
\[
    (1 - p - q)\cdot 2(1-\eps) + p \cdot k(1-\eps) + q \cdot 3k\alpha = O(\sqrt{k}) = O(\sqrt{n}).
\]

\paragraph{Upper bound for general graphs.} For general weighted graphs, generalizing the ternary signaling scheme above gives a cost of $n|\DS| / \sqrt{|\IS|}$, where $\IS$ is {an independent set of the graph}. 
{In particular, the first two plans still correspond to assigning signal $1 - \eps$ to $\DS$ and $\IS$, respectively, while Plan~C assigns a different non-zero value to $V \setminus \DS$, in the hope of making the slack at $0$ non-negative.}
When we can find a large $\IS$, this would give a non-trivial approximation of $|\DS|$, {which in turn approximates $\OPTIR$.} %
What if there are no large independent sets, e.g., when the graph is dense?

Assuming a lower bound on the non-zero edge weights, this issue can be resolved via a win-win argument. Suppose that the edge weights are at least $\delta > 0$. Let $m$ be the total edge weight in the graph. If $m$ is large, by \Cref{cor:degenerate-signal}, we have a persuasive signaling scheme with cost $O(n^2/m)$. If $m$ is small, by our assumption on the edge weights, there are at most $m / \delta$ edges in the graph. Intuitively, this sparse graph must have a large independent set. Indeed, we can lower bound $|\IS|$ by $\Omega(\delta n^2/m)$. Then, regardless of the value of $m$, the better between the two schemes gives an $O((n\cdot\OPTIR)^{2/3} \cdot \delta^{-1/3})$ cost.

Without such a lower bound $\delta$, we still have a non-trivial approximation of $\OPTIR$. The key observation is that $\IS$ does not need to be an independent set --- the same proof strategy goes through as long as, in the sub-graph induced by $\IS$, every vertex has a (weighted) degree of $O(\eps)$. Going one step further, we do not need $\IS$ to be a deterministic set at all! In our actual proof, we replace $\IS$ with a random set $\IStilde \subseteq V$, obtained from including each vertex in the graph with a fixed, carefully chosen probability. This gives the $O(n^{3/4}\cdot\left(\OPTIR\right)^{1/2})$ bound for the general case.

\subsection{Overview of lower bounds}
\label{sec:lower_overview}
To prove the tightness of our approximation guarantees, we start by reverse-engineering the proofs of the upper bounds, and identifying graphs on which the analyses are tight. For unit-weight graphs, the hard instance is exactly the double-star graph in \Cref{fig:double-star}. For graphs with edge weights bounded by $\Omega(1)$, we modify the construction in \Cref{fig:binary-schemes-fail} by replacing the $\Theta(n)$ triangles with $\Theta(n^{2/3})$ copies of a clique of size $\Theta(n^{1/3})$.

The more difficult part is, of course, to lower bound the cost of \emph{all} persuasive signaling schemes on these graphs. To this end, we revisit the characterization of persuasiveness in \Cref{lemma:persuasive-general}: $\D_\varphi \in \Delta(\signalspace^{V})$ is persuasive if and only if the induced slacks satisfy $\Delta_0 \ge 0$ and $\Delta_\theta = 0$ for all $\theta \in \signalspace \setminus \{0\}$. By \Cref{def:slack}, each $\Delta_\theta$ can be expressed as an expectation over the distribution $\D_\varphi$. Therefore, the family of persuasive signaling schemes with signal space $\signalspace$ is exactly a subset of $\Delta(\signalspace^V)$ defined by finitely many linear constraints. Consequently, the minimum-cost persuasive scheme is characterized by a linear program.

This simple observation leads us to lower bound the cost of persuasive signaling schemes by constructing a feasible solution of the dual LP. Concretely, suppose that we could find a function $f:[0, 1] \to \mathbb{R}$ such that $f(0) \ge 0$ and, for any $\signalspace$ and degenerate distribution at $s \in \signalspace^V$ (i.e., a deterministic labeling of $V$ using values in $\signalspace$), the resulting slacks satisfy
\begin{equation}\label{eq:dual-feasibility-overview}
    \|s\|_1 \ge \sum_{\theta \in \signalspace}f(\theta)\cdot\Delta_\theta + C.
\end{equation}
Then, for any persuasive signaling scheme $\D_\varphi \in \Delta(\signalspace^V)$, the linearity of expectation gives
\begin{align*}
    \Ex{s \sim \D_\varphi}{\|s\|_1}
\ge \sum_{\theta \in \signalspace}f(\theta)\cdot\Delta_{\theta} + C
\ge C.
\end{align*}
The last step above holds since when $\theta = 0$, $f(0)\cdot\Delta_0 \ge 0$ holds as both $f(0)$ and $\Delta_0$ are non-negative. When $\theta \ne 0$, $\Delta_\theta = 0$ implies $f(\theta)\cdot\Delta_\theta = 0$.

Our proof of the $\Omega(n^{2/3})$ lower bound (for graphs with edge weights lower bounded by $\Omega(1)$) proceeds by carefully choosing the function $f(\theta)$, and proving \Cref{eq:dual-feasibility-overview} for a sufficiently large $C$ via an involved case analysis. For unit-weight graphs, while the $\Omega(\sqrt{n})$ lower bound admits a simpler proof, in the proof we implicitly consider a dual solution, in which $f$ is a piece-wise constant function $f$ that takes a negative value with a large magnitude when $\theta$ is close to $1$.
    
    \section{Discussion and Future Directions}\label{sec:discussion}
In this section, we highlight a few concrete open problems and discuss a few natural extensions of our model.

\paragraph{The power of simple signaling schemes.} A recurring theme in our positive results is the surprising effectiveness of \emph{simple} signaling schemes that only use a few different signal values. The only exception is \Cref{thm:weighted-improve-ic}: we use $\Omega(n)$ different signals to achieve a cost below $\OPTIC$. Can we reduce the number of signals to $O(1)$? In \Cref{sec:discussion-details-match-ic}, we explain the technical difficulties in proving such a result, and give instances which suggest that the binary and ternary schemes that we consider are insufficient in general.

\paragraph{Tight approximation ratio for weighted graphs.} In terms of worst-case approximation guarantees, the only result that is not shown to be optimal is the $O(n^{3/4})$ approximation on general weighted graphs. In \Cref{sec:discussion-details-approx}, we give a concrete instance on which the $n^{3/4}$ ratio is conjectured to be tight, and discuss why the conjecture does not following easily from our current proof strategy.

\paragraph{Alternative collaboration systems.} Our work focuses on collaboration systems in which the quality of each agent's task is a linear combination of the agent's own contribution and the amounts contributed by the neighbors. In the context of collaborative federated learning, this corresponds to the \emph{random discovery} model proposed by~\citet{blumOneOneAll2021}. %
A non-linear version of the collaboration system is another model of~\citet{blumOneOneAll2021}, termed \emph{random coverage}, in which each agent is associated with a (discrete) data distribution, and the accuracy of each agent's task is linear in the total probability mass of the elements sampled by the agent themself and their neighbors. This formulation brings more structure to the relation between each pair of agents' tasks and might allow us to circumvent some hard instances in the linear model. For instance, agent $i$'s data are maximally effective for another agent $j$ only if they share the same distribution, and this property would be transitive. In contrast, in our current model, there might exist task types $i$, $j$, and $k$ such that $W_{i,j} = W_{j,k} = 1$ yet $W_{i,k} = 0$.

\paragraph{Alternative models of incentives.} Another modeling assumption that we made is regarding the agents' incentives. Implicit in the definition of persuasiveness (\Cref{sec:signaling-schemes}) is that each agent prioritizes satisfying feasibility, and exactly minimizes their own effort subject to the feasibility constraint. A natural relaxation is to allow \emph{approximate stability}, i.e., a signaling scheme is considered stable as long as no agent has the incentive to decrease their action by some $\eps > 0$. We may also consider alternative models in which each agent maximizes the expected quality of their task minus a penalization term that depends on their own workload.

    \section{Upper Bounds for Unit-Weight Graphs}\label{sec:unit-weight-upper}

In this section, we prove the upper bounds in \Cref{thm:unit-ub-opt,thm:unit-improve-ic}. To this end, we first state a few bounds on the cuts and induced sub-graphs in the unit-weight graph representation of the collaboration system. We then present a binary signaling scheme that is shown to be persuasive and achieve an $O(\sqrt{n})$ approximation of $\OPT$. Towards proving \Cref{thm:unit-improve-ic}, we give a binary scheme with a cost of exactly $\OPTIC$, which, under the additional assumption that $\OPT < \OPTIC$, easily implies a strict improvement upon $\OPTIC$.

\subsection{Weights of cuts and induced sub-graphs}
\label{sec:weight-induced-size}
We start by stating several elementary bounds on the cuts and induced edge weights (from \Cref{def:cut-induced}) for dominating sets and maximal independent sets in a graph. These bounds will be used towards constructing binary signaling schemes via \Cref{lemma:persuasive-binary}. We defer the proofs to \Cref{sec:unit-weight-upper-omitted}.

The first bound is regarding maximal independent sets in an induced sub-graph obtained from removing a dominating set.

\begin{lemma}\label[lemma]{lemma:IS-cut}
    Let $\DS$ denote a dominating set of a unit-weight graph $G = (V, E)$, and $\IS$ be an arbitrary maximal independent set of the sub-graph induced by $V \setminus \DS$. Then,
    $\cut(\IS, V \setminus \IS) \ge |V| - |\DS|.$
\end{lemma}

The second lemma considers a random set $\DStilde$ obtained from the independent rounding of a feasible solution that satisfies the IR requirement as defined in \Cref{sec:benchmarks-PoS}. This lemma holds for weighted graphs as well.

\begin{lemma}\label[lemma]{lemma:IR-sol-rounding}
    Let $G = (V, E, w)$ be a weighted graph, and $W$ be the corresponding matrix. Let $\btheta \in [0, 1]^V$ denote a feasible solution subject to IR, i.e., $W\btheta \ge 1$ and $\btheta\le1$ hold coordinate-wise. Let $\D \in \Delta(2^V)$ denote the distribution of a random subset of $V$ that includes each $v \in V$ independently with probability $\theta_v$. Then, the following bounds hold:
    \begin{itemize}
        \item $\Ex{S \sim \D}{|S|} = \|\btheta\|_1$.
        \item $\Ex{S \sim \D}{\induced(S)} \le \|\btheta\|_1^2 + \|\btheta\|_1$.
        \item $\Ex{S \sim \D}{\cut(S, V\setminus S)} \ge |V| - 2\|\btheta\|_1$.
    \end{itemize}
\end{lemma}

\subsection{General unit-weight graphs}
In this section, we prove the $O(\sqrt{n}\cdot\OPT)$ upper bound in \Cref{thm:unit-ub-opt} by designing a distribution $\D\in\Delta(2^V)$ that satisfies \Cref{lemma:persuasive-binary}. Our proof closely relies on the properties developed in \cref{sec:weight-induced-size}. Before the proof, we first present a lemma on the integrality gap of $\OPT$.
Note that the program for computing $\OPT$ is essentially the relaxed program determining the minimum dominating set of a graph. Therefore, our lemma follows from the integrality gap upper bound for the dominating set problem, e.g., see \citep[Section 1.7]{WS11}.
\begin{lemma}[Integrality gap]
        \label[lemma]{lem:DS-integrality-gap}
        There exists a dominating set $\DS\subseteq V$ such that $|\DS|=O(\log n\cdot\OPT)$. Moreover, the set $\DS$ can be computed efficiently through randomized rounding of $\OPT$.
\end{lemma}
We are now ready to prove the upper bound theorem.
\begin{theorem}[Upper bound part of \Cref{thm:unit-ub-opt}]
\label{thm:unit-upper-general}
    In any unit-weight graph, there is a persuasive signaling scheme with a cost of $O\left(\sqrt{n}\cdot\OPT\right)$.
\end{theorem}
\begin{proof}[Proof of \cref{thm:unit-upper-general}]
    We start with constructing a family of distributions $\D_p\in\Delta(2^V)$ parametrized by $p\in[0,1]$, {and} then show that there exists a choice of $p$ such that the {binary} signaling scheme induced by $\D_p$ is persuasive and has a small cost. 

    Let $\btheta^\star\in[0,1]^V$ be a socially optimal solution that achieves $\|\btheta^\star\|_1=\OPT$. 
    We define $\D_A\in\Delta(2^V)$ to be the distribution of a random subset $\DStilde\subseteq V$ that includes each $v\in V$ independently with probability $\theta_v^\star$ --- 
    we call $\DStilde\sim\D_A$ an \emph{independent rounding} of $\theta^\star$.
    In addition, let $\DS$ be a dominating set satisfying \Cref{lem:DS-integrality-gap}. Then, $\D_B\in\Delta(2^V)$ is defined as the degenerate distribution on $\IS$ where $\IS$ is a maximal independent set of $V\setminus\DS$.
    
    We set the probability mass function of the distribution $\D_p$ to be 
    \[
        \forall S\subseteq V,\qquad \D_p(S)=(1-p)\cdot\D_A(S)+p\cdot\D_B(S).
    \]
    Recall from \Cref{lemma:persuasive-binary} that $\D_p$ induces a persuasive binary signaling scheme if it satisfies
    \begin{align}
         \frac{\Ex{S \sim \D_p}{\cut(S, V \setminus S)}}{\Ex{S \sim \D_p}{|V \setminus S|}} \ge \frac{\Ex{S \sim \D_p}{\induced(S)}}{\Ex{S \sim \D_p}{|S|}}.
         \label{eq:tmp-12345}
    \end{align}
    We analyze each term separately to obtain a sufficient condition for \cref{eq:tmp-12345} to hold. For brevity, we denote $\ISsize\triangleq|\IS|$.
    \begin{itemize}
        \item Applying the first property in \cref{lemma:IR-sol-rounding} to $\D_A$ gives
        \[
            \Ex{S \sim \D_p}{|S|}
        =   (1-p)\cdot\Ex{S \sim \D_A}{|S|}+p\cdot\Ex{S\sim\D_B}{|S|}
        =   (1-p)\cdot\OPT+p\cdot\ISsize,
        \]
        and thus, $\Ex{S \sim \D_p}{|V \setminus S|}
        =   n - \Ex{S \sim \D_p}{|S|}
        =   (1-p)\cdot(n-\OPT)+p\cdot\left(n-\ISsize\right)$.
                \item For the expected size of the cut, we can lower bound it as below.
        \begin{align*}
            \Ex{S \sim \D_p}{\cut(S, V \setminus S)}=&(1-p)\cdot\Ex{S \sim \D_A}{\cut(S, V \setminus S)}+p\cdot\Ex{S \sim \D_B}{\cut(S, V \setminus S)}\\
            {\ge}&(1-p)\cdot \left(n-2\OPT\right)+p\cdot ({n-|\DS|})\\
            {\ge}&(1-p)\cdot \left(n-2\OPT\right)+p\cdot \left(n- \iota\cdot\OPT\right).
        \end{align*}
        In the above equations, the second step uses the third property of \cref{lemma:IR-sol-rounding} and \cref{lemma:IS-cut}.
        In the last step, we denote $|\DS|$ by $\iota\cdot\OPT$ with $\iota=O(\log n)$ according to \cref{lem:DS-integrality-gap}.
        \item For the expected size of induced subgraphs, we upper bound it using the second property of \Cref{lemma:IR-sol-rounding} and the independence of $\IS$,
        \begin{align*}
            \Ex{S \sim \D_p}{\induced(S)}=&(1-p)\cdot\Ex{S \sim \D_A}{\induced(S)}+p\cdot\Ex{S \sim \D_B}{\induced(S)}\\
            \le& (1-p)\left(\OPT^2+\OPT\right)+p\ISsize.
        \end{align*}
    \end{itemize}
    Plugging the above bounds into \cref{eq:tmp-12345}, we obtain a sufficient condition for persuasiveness:
    \begin{align}
        \frac{(1-p) \left(n-2\OPT\right)+p \left(n-\iota\cdot\OPT\right)}
        {(1-p)(n-\OPT)+p\cdot(n-\ISsize)} \ge 
        \frac{(1-p)(\OPT^2+\OPT)+p\ISsize}
        {(1-p)\OPT+p\ISsize}.
        \label{eq:unit-suff-condition}
    \end{align}
    For the remaining proof, we assume $\OPT\le\sqrt{n}/C$ and $\ISsize\ge C\sqrt{n}\cdot\OPT$ for a constant $C=10$. The first assumption is WLOG because if $\OPT>\sqrt{n}/C$, \Cref{cor:degenerate-signal} guarantees a cost of $O(\frac{n^2}{n+2m})=O(n)$ by not sending any signals, which is already an $O(\sqrt{n})$ approximation to $\OPT$. 
    The second assumption is WLOG because if $\beta=|\IS|<C\sqrt{n}\cdot\OPT$, we can expand $\IS$ into a maximal independent set of $V$ by including some additional vertices from $\DS$. The size of the resulting maximal independent set is at most $|\IS|+|\DS|=\ISsize+\iota\cdot\OPT=O(\sqrt{n}+\log n)\cdot\OPT$. Therefore, sending signal $1$ to the resulting maximal independent set gives a persuasive binary signaling scheme with cost $O(\sqrt{n}\cdot\OPT)$.

    Under these two additional assumptions, there always exists $p\in[0,1]$ that satisfies \cref{eq:unit-suff-condition} --- for example, taking $p=1$ satisfies \cref{eq:unit-suff-condition} by making the left-hand side $>1$ and the right-hand side $=1$. However, since 
    \[
    \Cost(\D_p)\le\Ex{S\sim\D_p}{|S|}=(1-p)\cdot\OPT+p\cdot\ISsize,
    \]
    the larger $p$ is, the more costly the signaling scheme would be. Therefore, we seek the smallest $p$ that satisfies \cref{eq:unit-suff-condition}.
    To do so, we continue simplifying \cref{eq:unit-suff-condition} by subtracting $1$ from both sides:
    \begin{align*}
        \eqref{eq:unit-suff-condition}
        \Longleftrightarrow& \frac{(1-p) \left(-\OPT\right)+p \left(\ISsize-\iota\cdot\OPT\right)}
        {(1-p)(n-\OPT)+p\cdot(n-\ISsize)} \ge 
        \frac{(1-p)\OPT^2}
        {(1-p)\OPT+p\ISsize}.
    \end{align*}
    Dividing $(1-p)$ on both denominators and enumerators and setting $\kappa=p/(1-p)$ gives us:
    \begin{align*}
        \eqref{eq:unit-suff-condition}
        \Longleftrightarrow& \frac{-\OPT+\kappa \left(\ISsize-\iota\cdot\OPT\right)}
        {(n-\OPT)+\kappa(n-\ISsize)} \ge 
        \frac{\OPT^2}
        {\OPT+\kappa\ISsize}\\
        \Longleftrightarrow& 
        \left(-\OPT+\kappa (\ISsize-\iota\cdot\OPT)\right)\left(\OPT+\kappa\ISsize\right)\ge 
        \left((n-\OPT)+\kappa(n-\ISsize)\right)\OPT^2\\
        \Longleftrightarrow& 
        \kappa^2\ISsize(\ISsize-\iota\cdot\OPT)\ge 
        \OPT^2(n-\OPT+\kappa(n-\ISsize))+\OPT^2-\kappa((\ISsize-\iota\cdot\OPT)\OPT-\ISsize\OPT)\\
        \Longleftrightarrow&
        \kappa^2\ISsize(\ISsize-\iota\cdot\OPT)\ge 
        \OPT^2(n-\OPT+\kappa(n-\ISsize+\iota)+1).
    \end{align*}
    In the last equation above, we have $\ISsize-\iota\cdot\OPT\ge\frac{\beta}{2}$ {for all sufficiently large $n$, by} the assumption that $\ISsize\ge C\sqrt{n}\cdot\OPT$ {and the fact that $\iota = O(\log n)$}. In addition, we upper bound the right-hand side by dropping the negative terms. This gives the following sufficient condition:
    \begin{align*}
        \eqref{eq:unit-suff-condition}\Longleftarrow&
        \frac{\kappa^2\ISsize^2}{2}\ge (2+2\kappa)n\OPT^2,
    \end{align*}
    which can be satisfied by choosing $\kappa=\Theta(\sqrt{n}\OPT/\ISsize)$. In particular, under this choice, we have $\kappa=O(1)$ due to the assumption that $\ISsize\ge C\sqrt{n}\OPT$, so both the left- and right-hand sides are of the order $\Theta(n\OPT^2)$. This gives us the choice of $p$ via $\kappa=\frac{p}{1-p}$.

    Finally, {by \Cref{lemma:persuasive-binary}}, the cost of the resulting signaling scheme satisfies
    \begin{align*}
        \Cost(\D_p)\le\Ex{S\sim\D_p}{|S|}=(1-p)\cdot\OPT+p\cdot\ISsize
       \le\OPT+\kappa\beta\le O(\sqrt{n}\cdot\OPT).
    \end{align*}
\end{proof}

\subsection{Matching the cost of full revelation}
Towards proving \Cref{thm:unit-improve-ic}, in this section, we give a binary signaling scheme with a cost of exactly $\OPTIC$.

\begin{theorem}\label{thm:unit-match-ic}
    For any unit-weight graph $G$, there is a binary signaling scheme with cost $\OPTIC$.
\end{theorem}

Note that \Cref{thm:unit-match-ic} would be trivial if we allow signaling schemes with a large signal space. This is because, for any optimal stable solution $\btheta^\star \in [0, 1]^V$, labeling each vertex $v \in V$ with $\theta^\star_v$ gives a persuasive signaling scheme with cost $\OPTIC$. Indeed, the interesting aspect of \Cref{thm:unit-match-ic} is that, while $\btheta^\star$ might contain many different entries, it is possible to achieve the same cost, via signaling, using only two different values.

Recall that \Cref{lemma:IC-solution-decomp} states that any stable solution in a unit-weight graph can be written as the marginal of a distribution over independent sets. We first show that the lemma immediately implies \Cref{thm:unit-match-ic}. 

\begin{proof}[Proof of \Cref{thm:unit-match-ic}]
Let $\btheta^\star \in [0, 1]^V$ be an optimal stable solution for graph $G = (V, E)$ with $\|\btheta^\star\|_1 = \OPTIC$. We consider the binary signaling scheme defined by distribution $\D \in \Delta(2^V)$ obtained from $\btheta^\star$ via \Cref{lemma:IC-solution-decomp}.

By \Cref{lemma:persuasive-binary}, to verify the persuasiveness of the binary signaling scheme, it suffices to show that
\[
    \frac{\Ex{S \sim \D}{\cut(S, V \setminus S)}}{\Ex{S \sim \D}{|V \setminus S|}} \ge \frac{\Ex{S \sim \D}{\induced(S)}}{\Ex{S \sim \D}{|S|}}.
\]
Since every $S$ in the support of $\D$ is an independent set, we have $\Ex{S \sim \D}{\induced(S)} = \Ex{S \sim \D}{|S|}$. It remains to show that $\Ex{S \sim \D}{\cut(S, V \setminus S)} \ge \Ex{S \sim \D}{|V \setminus S|}$.

Again, since $S$ is always an independent set, the cut size $|\cut(S, V \setminus S)|$ is equal to $\sum_{v \in S}\deg(v)$.
Then, we have
\begin{align*}
    \Ex{S \sim \D}{|\cut(S, V \setminus S)|}
&=  \sum_{v \in V}\pr{S \sim \D}{v \in S}\cdot\deg(v)\tag{$S$ is an independent set}\\
&=  \sum_{v \in V}\theta^\star_v\cdot\deg(v)\tag{property of $\D$ from \Cref{lemma:IC-solution-decomp}}\\
&=  \sum_{v \in V}\sum_{u \in N(v)}\theta^\star_v
=  \sum_{v \in V}\sum_{u \in N(v)}\theta^\star_u\\
&\ge \sum_{v \in V}(1 - \theta^\star_v) \tag{$\btheta^\star$ is a feasible solution}\\
&=  n - \|\btheta^\star\|_1
=   \Ex{S \sim \D}{|V \setminus S|}.
\end{align*}
The fourth step holds since both the third line and the fourth line are equal to $\sum_{\{u, v\} \in E}(\theta^\star_u + \theta^\star_v)$. This proves persuasiveness.

By \Cref{lemma:persuasive-binary}, the cost of this scheme is exactly
\[
    \frac{\left(\Ex{S \sim \D}{|S|}\right)^2}{\Ex{S \sim \D}{\induced(S)}}
=   \Ex{S \sim \D}{|S|}
=   \sum_{v \in V}\pr{S \sim \D}{v \in S}
=   \sum_{v \in V}\theta^\star_v
=   \|\btheta^\star\|_1
=   \OPTIC.
\]
The third step above applies the property of $\D$ guaranteed by \Cref{lemma:IC-solution-decomp}.
\end{proof}

Now we prove \Cref{lemma:IC-solution-decomp} via rounding the stable solution through a correlated rounding method. The method is referred to as the \emph{competing exponential clocks} in the context of approximation algorithms (e.g.,~\citep{BNS13}).

\begin{proof}[Proof of \Cref{lemma:IC-solution-decomp}]
    Let $V' \coloneqq \{v \in V: \theta_v > 0\}$ be the support of the stable solution $\btheta \in [0, 1]^V$. We define $\D \in \Delta(2^V)$ as the distribution of the set $S \subseteq V$ generated by the following process:
    \begin{itemize}
        \item For each $v \in V'$, independently draw $X_v \sim \Exp(\theta_v)$.
        \item Choose the set $S$ as $\left\{v \in V': X_v < \min_{u \in N(v) \cap V'}X_v\right\}$.
    \end{itemize}
    It is clear that $S$ must be an independent set: if $S$ contains two neighbouring vertices $u$ and $v$, we must have both $X_u < X_v$ and $X_v < X_u$, a contradiction.

    For every $v \in V \setminus V'$, we clearly have $\pr{S \sim \D}{v \in S} = 0 = \theta_v$. To see that $\pr{S}{v \in S} = \theta_v$ holds for every $v \in V'$, we use the condition $\theta_v + \sum_{u \in N(v) \cap V'}\theta_u = 1$, which follows from stability. Then, the property of the exponential distribution shows that $\min_{u \in N(v) \cap V'}X_u$ follows the exponential distribution $\Exp(1 - \theta_v)$ and is independent of $X_u$. Therefore, we have
    \begin{align*}
        \pr{}{v \in S}
    &=  \pr{}{X_v < \min_{u \in N(v) \cap V'}X_u}\\
    &=  \Ex{X_v \sim \Exp(\theta_v)}{\exp(-(1 - \theta_v)X_v)}\\
    &=  \int_{0}^{+\infty}\theta_v\exp(-\theta_v x)\cdot \exp(-(1-\theta_v) x)~\rmd x
    =   \theta_v,
    \end{align*}
    as desired.
\end{proof}

\subsection{Strict improvement upon full revelation}
\label{sec:strict-improve-unweighted}
In this section, we build on \Cref{thm:unit-match-ic} and characterize conditions under which signaling results in a \emph{strict} improvement compared to $\OPTIC$ and the amount of improvement achieved. 

\begin{theorem}[Formal version of \Cref{thm:unit-improve-ic}]
\label{thm:unit-strict-improvement}
    On any unit-weight graph {that satisfies} $\OPTIC > \OPT$, there exists a binary signaling scheme $\D_\varphi$ with   
    $\Cost(\varphi)=\OPTIC-\eps(\OPTIC-\OPT)$, where
    \begin{align}
    \eps=\max\left\{\min\left\{\frac{1}{2},\ \frac{\PoS-1}{\PoS+\frac{2n}{\PoS+1}}\right\}
    , \frac{\PoS-1}{\PoS+(n-\OPT)}\right\} {> 0,}
    \label{eq:choice-of-eps}
    \end{align}
    {and $\PoS = \OPTIC/\OPT$ is the price of stability.}
\end{theorem}

To prove \Cref{thm:unit-strict-improvement}, we first present two technical lemmas that connect the ``wastefulness'' of the optimal stable solution to the gap $\OPTIC-\OPT$. Formally, we call a feasible solution $\btheta\in[0,1]^V$ \emph{wasteful} if it satisfies $W\btheta \ne \vecone$. Recall that feasibility requires $W\btheta \ge \vecone$ coordinate-wise. Therefore, a wasteful solution is one in which at least one of the tasks receives strictly more contribution than what is required. Our first lemma states that, even on general weighted graphs, {every non-wasteful feasible solution must be optimal. In particular, assuming $\OPTIC > \OPT$, the lemma below implies that the optimal stable solution must be wasteful, and this wastefulness allows us to design a signaling scheme with a cost strictly below $\OPTIC$.}
\begin{lemma}
\label[lemma]{lem:wasteful-general}
    In a general weighted graph, if there exists a non-wasteful feasible solution {$\btheta \in [0, 1]^V$}, it holds that {$\OPT= \|\btheta\|_1$}.
\end{lemma}
\begin{proof}
    Recall from the definition that $\OPT$ can be computed by the following linear program:
    \begin{align}
        \min_{\btheta\in\R^n}\  &\vecone^{\top} \btheta\label{eq:primal-lp}\tag{primal LP}\\
        \suchthat\  & W\btheta\ge\vecone;\nonumber\\
        &\btheta\ge\veczero.\nonumber
    \end{align}
    Since $W = W^{\top}$, the dual linear program of \eqref{eq:primal-lp} is given by:
    \begin{align}
        \max_{\dualvar\in\R^n}\  &\vecone^{\top} \dualvar\label{eq:dual-lp}\tag{dual LP}\\
        \suchthat\  & W \dualvar\le\vecone;\nonumber\\
        &\dualvar\ge\veczero.\nonumber
    \end{align}
    Since $\btheta$ is feasible and non-wasteful, we have $W\btheta=\vecone$, which implies that $\btheta$ is feasible solution for both \eqref{eq:primal-lp} and \eqref{eq:dual-lp}, and achieves the same objective value of $\vecone^{\top} \btheta = \|\btheta\|_1$. Therefore, $\btheta$ must be an optimal solution for \eqref{eq:primal-lp}, and we have $\OPT = \|\btheta\|_1$.
\end{proof}

The second lemma quantitatively captures the gap $\OPTIC-\OPT$ in terms of the ``wastefulness'' of the optimal {stable} solution {$\btheta^\star$}, i.e., the gap $\|W\btheta^\star\|_1-n$, in the case of unit-weight graphs. {Again, we state a more general result, and will apply it specifically to the optimal stable solution.}
\begin{lemma}
\label[lemma]{lem:wasteful-unit}
    Let $\btheta \in [0, 1]^V$ be a feasible solution on a unit-weight graph. Then,
    \[\|W\btheta\|_1\ge n+ \|\btheta\|_1-\OPT.\]
\end{lemma}
We prove \Cref{lem:wasteful-unit} in \Cref{app:wasteful-proof-unit}. The high-level idea of the proof is to convert $\btheta$, which is a feasible solution of \eqref{eq:primal-lp}, into a feasible solution $\dualvar^\star$ of \eqref{eq:dual-lp} by lowering the contributions that enter wasteful coordinates. Through our careful construction, we ensure the resulting decrease in the objective value (i.e., $\vecone^{\top} \btheta-\vecone^{\top} \dualvar^\star$) is upper bounded by the wastefulness of $\btheta$ (i.e., $\|W\btheta\|_1-n$). Finally, we establish the lemma by invoking weak duality to show that $\vecone^{\top} \dualvar^\star\le\OPT$.

Now we prove \Cref{thm:unit-strict-improvement} by constructing a signaling scheme that randomizes between an independent rounding of $\OPT$ and the signaling scheme in \Cref{thm:unit-match-ic}.

\begin{proof}[Proof of \Cref{thm:unit-strict-improvement}]
    Let $\btheta^\star \in [0, 1]^V$ be an optimal solution with cost $\OPT$, and $\btheta \in [0,1]^V$ be an optimal stable solution {with cost $\OPTIC$}. 
    Then, let $\D_A$ be the distribution defined by rounding $\btheta^\star$ independently, and $\D_B$ be the distribution obtained from rounding {$\btheta$} according to \Cref{lemma:IC-solution-decomp}.
    We define a parametrized family of distributions $\D_\eps\in\Delta(2^V)$ for {$\eps\in[0,1]$} as
    \[
        \forall S\subseteq V,\qquad \D_\eps(S)=\eps\cdot\D_A(S)+(1-\eps)\cdot\D_B(S).
    \]
    Recall from \cref{lemma:persuasive-binary} {that} the distribution $\D_\eps$ induces a persuasive binary signaling scheme if the following inequality holds:
    \begin{align}
         \frac{\Ex{S \sim \D_\eps}{\cut(S, V \setminus S)}}{\Ex{S \sim \D_\eps}{|V \setminus S|}} \ge \frac{\Ex{S \sim \D_\eps}{\induced(S)}}{\Ex{S \sim \D_\eps}{|S|}}.
         \label{eq:tmp-567}
    \end{align}
    We analyze each term separately to obtain a sufficient condition for \cref{eq:tmp-567} to hold.
    \begin{itemize}
        \item For the expected size of $|S|$ and $|V\setminus S|$, we combine \cref{lemma:IC-solution-decomp} with the first property of \cref{lemma:IR-sol-rounding} to get
        \begin{align*}
            \Ex{S \sim \D_\eps}{|S|}=&~\eps\Ex{S \sim \D_A}{|S|}+(1-\eps)\Ex{S\sim\D_B}{|S|}=\eps\OPT+(1-\eps)\OPTIC;\\
            \Ex{S \sim \D_\eps}{|V \setminus S|}=&~n-\Ex{S \sim \D_\eps}{|S|}=n-\OPTIC+\eps(\OPTIC-\OPT).
        \end{align*}
        \item For the expected size of the cut, we have
        \begin{align*}
            \Ex{S \sim \D_B}{\cut(S, V\setminus S)}
        &=  \Ex{S \sim \D_B}{\sum_{v \in S}\deg(v)} \tag{$S \sim \D_B$ is an independent set}\\
        &=  \sum_{v \in V}\pr{S \sim \D_B}{v \in S}\cdot\deg(v) \tag{linearity of expectation}\\
        &=  \sum_{v \in V}\theta_v\cdot\deg(v) \tag{property of $\D_B$ from \cref{lemma:IC-solution-decomp}}\\
        &=  \btheta^{\top}(W\vecone - \vecone) \tag{$W$ is the adjacency matrix of $G$}\\
        &=  \|W\btheta\|_1 - \|\btheta\|_1 \tag{$\btheta^{\top}W\vecone = (W\btheta)^{\top}\vecone = \|W\btheta\|_1$}\\
        &\ge (n+\OPTIC-\OPT)- \OPTIC\tag{\Cref{lem:wasteful-unit,lemma:IC-solution-decomp}}\\
        &= n-\OPT.
        \end{align*}
        Combining the above inequalities with the third property of \Cref{lemma:IR-sol-rounding} gives us
        \begin{align*}
            \Ex{S \sim \D_\eps}{\cut(S, V \setminus S)}&=\eps\Ex{S \sim \D_A}{\cut(S, V \setminus S)}+(1-\eps)\Ex{S \sim \D_B}{\cut(S, V \setminus S)}\\
            &\ge\eps \left(n-2\OPT\right)+(1-\eps) (n-\OPT)\\
            &=n-(1+\eps)\OPT.
        \end{align*}
        \item For the expected size of induced subgraphs, we upper bound it using the second property of \Cref{lemma:IR-sol-rounding} and the independence of subsets drawn from $\D_B$ (\Cref{lemma:IC-solution-decomp}).
        \begin{align*}
            \Ex{S \sim \D_\eps}{\induced(S)}&=\eps\Ex{S \sim \D_A}{\induced(S)}+(1-\eps)\Ex{S \sim \D_B}{\induced(S)}\\
            &\le \eps\left(\OPT^2+\OPT\right)+(1-\eps)\OPTIC.
        \end{align*}
    \end{itemize}
    Plugging the above bounds into \Cref{eq:tmp-567}, we obtain its sufficient condition:
    \begin{align}
        \frac{n-(1+\eps)\OPT}{n-\OPTIC+\eps(\OPTIC-\OPT)}
        \ge \frac{\eps\left(\OPT^2+\OPT\right)+(1-\eps)\OPTIC}{\eps\OPT+(1-\eps)\OPTIC}
        \label{eq:another-suff-cond}
    \end{align}
    We first show the existence of {$\eps\in (0,1]$} to make \cref{eq:another-suff-cond} hold. Note that when $\eps=0$, it reduces to
    \[
    \frac{n-\OPT}{n-\OPTIC}\ge\frac{\OPTIC}{\OPTIC}=1,
    \]
    which holds with a strict inequality because $\OPTIC>\OPT$. Since both sides are continuous in $\eps$, there should exist a small neighborhood of $\eps$ around $0$ that all satisfies \cref{eq:another-suff-cond}.
    
    Moreover, according to \Cref{lemma:persuasive-binary}, the cost of the signaling scheme induced by $\D_\eps$ is 
    \begin{align*}
        \Cost(\D_\eps)\le&\Ex{S\sim\D_\eps}{|S|}=\eps\OPT+(1-\eps)\OPTIC=\OPTIC-\eps(\OPTIC-\OPT).
    \end{align*}
    Therefore, to prove \cref{thm:unit-strict-improvement}, it remains to show that the choice of $\eps$ in \cref{eq:choice-of-eps} satisfies \cref{eq:another-suff-cond}.
    We start with simplifying the condition by subtracting $1$ on both sides:
    \begin{align*}
        \eqref{eq:another-suff-cond}\Longleftrightarrow\ & \frac{\OPTIC-\OPT-\eps\OPTIC}{n-\OPTIC+\eps(\OPTIC-\OPT)}
        \ge \frac{\eps\OPT^2}{\eps\OPT+(1-\eps)\OPTIC}.
    \end{align*}
    On the one hand, since lower bounding the left-hand side and upper bounding the right-hand side results in a sufficient condition, we use $\eps\OPT+(1-\eps)\OPTIC\ge\OPT$ and obtain
    \begin{align*}
        \eqref{eq:another-suff-cond}\Longleftarrow\ & \frac{\OPTIC-\OPT-\eps\OPTIC}{n-\OPT}
        \ge \frac{\eps\OPT^2}{\OPT}\\
        \Longleftrightarrow\ &\frac{\OPTIC-\OPT}{n-\OPT}\ge\eps
        \left(\OPT+\frac{\OPTIC}{n-\OPT}\right)\\
        \Longleftrightarrow\ &
        \eps\le \frac{\OPTIC-\OPT}{\OPT(n-\OPT)+\OPTIC}
        =\frac{\PoS-1}{\PoS+(n-\OPT)},
    \end{align*}
    where $\PoS=\OPTIC/\OPT$ is the price of stability. This justifies the second choice of $\eps$ in \Cref{eq:choice-of-eps}.

    On the other hand, under the condition of $\eps\le\frac{1}{2}$, we have $\eps\OPT+(1-\eps)\OPTIC\ge\frac{\OPT+\OPTIC}{2}$. We use this to provide an alternative sufficient condition to \cref{eq:another-suff-cond}. 
    \begin{align*}
        \eqref{eq:another-suff-cond}\Longleftarrow\ & \frac{\OPTIC-\OPT-\eps\OPTIC}{n}
        \ge \frac{\eps\OPT^2}{\frac{\OPT+\OPTIC}{2}}\\
        \Longleftrightarrow\ &\OPTIC-\OPT\ge\eps\left(
            \frac{2n\OPT^2}{\OPT+\OPTIC}+\OPTIC
        \right)\\
        \Longleftrightarrow\ &
        \eps\le\frac{\OPTIC-\OPT}{\frac{2n\OPT^2}{\OPT+\OPTIC}+\OPTIC}
        =\frac{\PoS-1}{\PoS+\frac{2n}{\PoS+1}},
    \end{align*}
    which justifies the first option of $\eps$ when $\eps\le\frac{1}{2}$.
\end{proof}

    \section{Lower Bounds for Unit-Weight Graphs}
\label{sec:unit-lower}

In this section, we prove the lower bound part of \Cref{thm:unit-ub-opt} by showing that on the double-star graph in \Cref{fig:double-star}, every persuasive signaling scheme must have an $\Omega(\sqrt{n})$ cost, while the optimal total workload is $O(1)$. Therefore, the $O(\sqrt{n})$ approximation guarantee in \Cref{thm:unit-ub-opt} cannot be significantly improved, even when non-binary schemes are allowed. 

Then, we prove a more general lower bound, which states that for any $n$ and $k \in [2, n]$, there is an instance on which $\OPT = \Theta(k)$, $\OPTIC = \Theta(n)$, and no persuasive signaling scheme can achieve a cost that is much better than $\min\{k\sqrt{n}, n\}$. Recall that by Theorems \ref{thm:unit-ub-opt}~and~\ref{thm:unit-improve-ic}, on this instance there is a persuasive binary signaling scheme with cost $O(\min\{\OPT\cdot\sqrt{n}, \OPTIC\}) = O(\min\{k\sqrt{n}, n\})$. Therefore, for a wide range of $\OPT$, between the two signaling schemes that we develop for unit-weight graphs, the better one is essentially optimal.

\subsection{Lower bound for the double-star graph}

\begin{theorem}[Lower bound part of \Cref{thm:unit-ub-opt}]
    On the double-star graph with $n = 2k + 2$ vertices, every persuasive signaling scheme has a cost of $\Omega(\sqrt{n})$.
\end{theorem}

\begin{proof}
    For clarity, we rename the vertices in the graph with $V = \{(c, 1), (c, 2)\} \cup \{(i, j): i \in [2], j \in [k]\}$, where $(c, i)$ is the center of the $i$-th star, and $(i, j)$ is the $j$-th leaf in the $i$-th star.

    Fix a persuasive signaling scheme for the graph. By \Cref{lem:simplify-signal}, it is without loss of generality to assume that the scheme is identity-independent and is specified by  $\D_\varphi \in \Delta(\signalspace^V)$ on signal space $\signalspace$. In other words, the signaling scheme samples $(s_{c,1}, s_{c,2}, s_{1,1}, \ldots, s_{1,k}, s_{2,1}, \ldots, s_{2,k})\in \signalspace^V$ from $\D_\varphi$, labels the two centers with $s_{c,1}$ and $s_{c,2}$, and labels the $j$-th leaf in the $i$-th star with $s_{i,j}$. For brevity, we use $x$ and $y$ as shorthands for $s_{c,1}$ and $s_{c,2}$, and let random variable $\Sum$ denote $x + y + \sum_{i=1}^{2}\sum_{j=1}^{k}s_{i,j}$. Then, the cost of the scheme is given by $\Ex{\D_\varphi}{\Sum}$.

    Let $C \coloneqq 3$. We will show that, under the assumption that $\Ex{\D_\varphi}{\Sum} \le \sqrt{k} / C$, the persuasiveness of $\D_\varphi$ implies $\Ex{\D_\varphi}{\Sum} = \Omega(\sqrt{k})$. Therefore, $\Ex{\D_\varphi}{\Sum} = \Omega(\sqrt{k}) = \Omega(\sqrt{n})$ always holds.

    By \Cref{lemma:persuasive-general}, the persuasiveness of $\D_\varphi$ implies that for every $\theta \in \signalspace$,
    \[
    0 \le \Delta_\theta=\Ex{\D_\varphi}{\sum_{v \in V}\1{s_v=\theta}\left(\sum_{u\in N(v)}s_u-(1-s_v)\right)}.
    \]
    Summing over all $\theta \in \signalspace$ gives
    \[
        0
    \le \sum_{\theta\in\signalspace}\Delta_\theta=\Ex{\D_\varphi}{\sum_{v \in V}\left(s_v+\sum_{u\in N(v)}s_u-1\right)}
    =\Ex{\D_\varphi}{\sum_{u \in V}s_u\left(1+|N(u)|\right)}-n.
    \]
    Since $|N(u)|=1$ when $u$ is a leaf and $|N(u)|=k+1$ when $u$ is a center, we have
    \[
        n
    \le \Ex{\D_\varphi}{\sum_{u \in V}s_u\left(1+|N(u)|\right)}
    =   2\Ex{\D_\varphi}{\sum_{v \in V}s_v}+k\cdot\Ex{\D_\varphi}{s_{c,1} + s_{c,2}}
    =   2\Ex{\D_\varphi}{\Sum} + k\cdot\Ex{\D_\varphi}{x+y},
    \]
    which further implies $\Ex{\D_\varphi}{x + y}\ge\frac{n-2\Ex{\D_\varphi}{\Sum}}{k}$. By the assumption that $\Ex{\D_\varphi}{\Sum} \le \sqrt{k} / C$, we have
    \[
        \Ex{\D_\varphi}{x + y} 
    \ge \frac{2k + 2 - 2\sqrt{k}/C}{k}
    \ge 2 - \frac{2}{C\sqrt{k}}.
    \]
    Let $\eps \coloneqq 1 / \sqrt{k}$. Since $2 - (x + y)$ is always non-negative, by Markov's inequality, we have
    \begin{align*}
        \pr{\D_\varphi}{x + y < 2 - \eps}
    &=  \pr{\D_\varphi}{2 - (x+y) > \eps}
    \le\frac{\Ex{\D_\varphi}{2 - (x + y)}}{\eps}\\
    &\le\frac{2 - \left(2 - \frac{2}{C\sqrt{k}}\right)}{1 / \sqrt{k}}
    =   \frac{2}{C}. \tag{$\Ex{\D_\varphi}{x + y} \ge 2 - 2/(C\sqrt{k})$}
    \end{align*}
    Equivalently, we have $\pr{\D_\varphi}{x + y \ge 2 - \eps} \ge 1 - 2/C$.

    As long as $k \ge 2$, we have $1 - \eps = 1 - 1/\sqrt{k} > 0$, so \Cref{lemma:persuasive-general} implies $\sum_{\theta\in \signalspace}\Delta_\theta\cdot\1{\theta \ge 1 - \eps}=0$, which is equivalent to
    \[
        \Ex{\D_\varphi}{\sum_{v \in V}\1{s_v\ge1-\eps}\left(s_v+\sum_{u\in N(v)}s_u-1\right)} = 0.
    \]
    The left-hand side above can be written as the sum of the following four terms:
    \begin{itemize}
        \item $T_1 \coloneqq \Ex{\D_\varphi}{\1{x+y\ge2-\eps}\cdot\left(2(x+y-1)+\sum_{i=1}^{2}\sum_{j=1}^{k}s_{i,j}\right)}$, the contribution from the two centers when $x + y \ge 2 - \eps$. Note that this condition implies both $x \ge 1 - \eps$ and $y \ge 1 - \eps$.
        \item $T_2 \coloneqq \Ex{\D_\varphi}{\1{x+y<2-\eps \wedge x\ge1-\eps}\left(x+y-1+\sum_{i=1}^{k} s_{1,i}\right)}$, the contribution from the center of the first star, when $x + y < 2 - \eps$.
        \item $T_3 \coloneqq \Ex{\D_\varphi}{\1{x+y<2-\eps \wedge y\ge1-\eps}\left(x+y-1+\sum_{i=1}^{k} s_{2,i}\right)}$, the contribution from the center of the second star, when $x + y < 2 - \eps$.
        \item $T_4 \coloneqq \Ex{\D_\varphi}{\sum_{v \in V \setminus \{(c, 1), (c, 2)\}}\1{s_{v}\ge1-\eps}\left(s_v+\sum_{u\in N(v)}s_{u}-1\right)}$, the contribution from the leaves.
    \end{itemize}
    Recall that we showed $\pr{\D_\varphi}{x + y \ge 2 - \eps} \ge 1 - 2/C$ earlier, so the first term $T_1$ satisfies
    \begin{align*}
        T_1
    &\ge\Ex{\D_\varphi}{\1{x+y\ge2-\eps}\cdot2(x+y-1)}
    \ge\Ex{\D_\varphi}{\1{x+y\ge2-\eps}\cdot2(2 - \eps -1)}\\
    &=  \pr{\D_\varphi}{x+y\ge2-\eps}\cdot(2 - 2\eps)
    \ge (1 - 2/C)\cdot2\cdot(1 - \eps)
    \ge 2\cdot(1 - 2/C - \eps).
    \end{align*}
    For the second term, note that when $x \ge 1 - \eps$ holds, $x + y - 1 + \sum_{i=1}^{k}s_{1,i} \ge x - 1 \ge -\eps$. Thus,
    \[
        T_2
    \ge -\eps\cdot\pr{\D_\varphi}{x + y < 2 - \eps \wedge x \ge 1 - \eps}
    \ge -\eps\cdot\pr{\D_\varphi}{x + y < 2 - \eps}
    \ge -2\eps/C.
    \]
    By symmetry, we also have $T_3 \ge -2\eps/C$.
    For the same reason, the last term can be lower bounded by $-\eps\cdot\Ex{\D_\varphi}{\sum_{v \in V \setminus \{(c, 1), (c, 2)\}}\1{s_{v}\ge1-\eps}}$.
    Note that
    \begin{align*}
        \Ex{\D_\varphi}{\Sum}
    &\ge\sum_{v \in V \setminus \{(c, 1), (c, 2)\}}\Ex{\D_\varphi}{s_v}
    \ge\sum_{v \in V \setminus \{(c, 1), (c, 2)\}}(1-\eps)\cdot\pr{\D_\varphi}{s_v \ge 1 - \eps}\\
    &=  (1 - \eps)\cdot\Ex{\D_\varphi}{\sum_{v \in V \setminus \{(c, 1), (c, 2)\}}\1{s_{v}\ge1-\eps}},
    \end{align*}
    so we have $T_4 \ge -\frac{\eps}{1-\eps}\Ex{\D_\varphi}{\Sum}$.
    
    Recall that $T_1 + T_2 + T_3 + T_4 = 0$. Therefore, combining the four lower bounds gives
    \[
        0
    \ge 2\cdot(1 - 2/C - \eps) - 2\eps/C - 2\eps/C - \frac{\eps}{1-\eps}\Ex{\D_\varphi}{\Sum},
    \]
    which is equivalent to 
    \[
        \Ex{\D_\varphi}{\Sum}
    \ge \frac{1-\eps}{\eps}\cdot[2\cdot(1 - 2/C - \eps) - 4\eps/C].
    \]
    By our choice of $C = 3$ and $\eps = 1/\sqrt{k}$, the right-hand side above is $\Omega(\sqrt{k}) = \Omega(\sqrt{n})$ for all sufficiently large $k$. This establishes the $\Omega(\sqrt{n})$ lower bound on the cost.
\end{proof}

\subsection{A more general lower bound via $k$-stars}\label{sec:unit-lb-general}
The previous lower bound gives an instance where $\OPT = O(1)$ and the optimal persuasive signaling scheme has an $\Omega(\sqrt{n})$ cost. Now, we present a more general lower bound in \Cref{thm:unit-lb-general} that holds for a wider range of $\OPT$. The proof of \Cref{thm:unit-lb-general} is deferred to \Cref{app:unit-lower}.

\begin{theorem}\label{thm:unit-lb-general}
    For any $n \ge k \ge 2$, there is a unit-weight graph with $\le n$ vertices, on which the following hold simultaneously:
    \begin{itemize}
        \item $\OPT = \Theta(k)$.
        \item $\OPTIC = \Omega(n)$.
        \item Every persuasive signaling scheme has a cost of $\Omega(\min\{k\sqrt{n}, n\})$.
    \end{itemize}
    When $n$ is sufficiently large, the $\Theta(\cdot)$ and $\Omega(\cdot)$ notations above hide universal constants that do not depend on $n$ and $k$.
\end{theorem}

When $k = O(\sqrt{n})$, we consider a generalized version of the double-star construction, in which there are $k$ disjoint stars, each with $\Theta(n/k)$ leaves. The centers of the $k$ stars form a clique. This gives an instance on which $\OPT = k$ and $\OPTIC = \Omega(n)$. By a similar analysis to that of the double-star graph, we show that every persuasive signaling scheme has an $\Omega(k\sqrt{n})$ cost on this $k$-star graph.

When $k \gg \sqrt{n}$, the above construction cannot be directly used for establishing an $\Omega(n)$ lower bound. This is because the $k$-clique formed by the centers already contains $\Omega(k^2) \gg n$ edges; by \Cref{cor:degenerate-signal}, an $O(n^2/m) = o(n)$ cost can be achieved. Instead, we apply this construction on $n' = (n / k)^2$ and $k' = n / k$ (so that $k' = O(\sqrt{n'})$ still holds), and construct $n / n' = k^2/n$ disjoint copies of such a graph. Intuitively, the optimal solutions in the large graph must be $n/n'$ times those for a single copy. This gives $\OPT = (n/n')\cdot k' = k$, $\OPTIC = (n/n')\cdot n' = n$, and the cost of every persuasive scheme is lower bounded by $(n/n')\cdot\Omega(n') = \Omega(n)$.

    \section{Upper Bounds for Weighted Graphs}
\label{sec:weighted-upper}

In this section, we prove upper bounds in weighted graphs using the ternary signaling approach outlined in \Cref{sec:weighted-overview}. Together, they give the upper bound part of \Cref{thm:weighted-ternary-opt}. We also prove \Cref{thm:weighted-improve-ic}, which states that we can achieve a cost strictly lower than $\OPTIC$ whenever $\OPTIR < \OPTIC$.

\subsection{Upper bound for approximating $\OPTIR$}
Before presenting the formal theorem statements, we first establish the following technical lemma, which shows that if the graph contains a large set that is ``almost independent'', we have a good approximation of $\OPTIR$.

\begin{lemma}\label[lemma]{lemma:weighted-upper-technical}
    Suppose that distribution $\D \in \Delta(2^V)$ satisfies
    \[
        \Ex{S \sim \D}{\induced(S)} \le (1 + \gamma)\cdot\Ex{S \sim \D}{|S|}.
    \]
    Then, there is a persuasive signaling scheme with a cost of
    \[
        O\left(\OPTIR + \gamma n + \frac{n\cdot\OPTIR}{\sqrt{\Ex{S \sim \D}{|S|}}}\right).
    \]
\end{lemma}

Note that if $\gamma = 0$, the condition $\Ex{S \sim \D}{\induced(S)} \le (1 + \gamma)\cdot\Ex{S \sim \D}{|S|}$ is equivalent to that $\D$ is supported over independent sets of $V$. When $\gamma > 0$, we allow some edges (with a small total weight) in the sub-graph induced by a random $S \sim \D$, at the cost of an additional $\gamma n$ term in the cost of the resulting signaling scheme.

We include a proof of \Cref{lemma:weighted-upper-technical} in \Cref{app:proof-weighted-technical}. The proof is based on optimizing the probabilities of sending signals to $\DStilde,\IStilde\sim\D$, and $V\setminus\DStilde$ as well as the values of the three signals, as described in \Cref{sec:weighted-overview}. In particular, let $\btheta^\star \in [0, 1]^V$ be an optimal solution subject to IR with cost $\|\btheta^\star\|_1 = \OPTIR$. We consider the following family of ternary signaling schemes with parameters $p,q\ge0$ and $\epsilon,\alpha\in[0,1]$:
    \begin{itemize}
        \item With probability $\frac{1}{1 + p + q}$, draw a random set $\DStilde \subseteq V$ that includes each $v \in V$ independently with probability $\theta^\star_v$. Label (the vertices in) $\DStilde$ with $1 - \eps$, and label $V \setminus \DStilde$ with $0$.
        
        \item With probability $\frac{p}{1 + p + q}$, draw a random set $\IStilde \sim \D$. Label $\IStilde$ with $1 - \eps$, and label $V \setminus \IStilde$ with $0$.
        
        \item With probability $\frac{q}{1 + p + q}$, draw a random set $\DStilde \subseteq V$ that includes each $v \in V$ independently with probability $\theta^\star_v$. Label $V \setminus \DStilde$ with $\alpha$, and label $\DStilde$ with $0$.
    \end{itemize}
    In the next section, we prove the upper bound side of \Cref{thm:weighted-ternary-opt} by instantiating \Cref{lemma:weighted-upper-technical} with specific choices of $\D$.

In \Cref{thm:weighted-upper-general}, we provide the $O(n^{3/4})$ approximation guarantee for general weighted graphs. The distribution $\D$ is chosen as the distribution of a random subset of the vertices sampled with a carefully chosen probability. The proof of \Cref{thm:weighted-upper-general} is in \Cref{app:proof-weighted-upper-general}.

\begin{theorem}[The first upper bound in \Cref{thm:weighted-ternary-opt}]\label{thm:weighted-upper-general}
    In any weighted graph, there exists a ternary signaling scheme that is persuasive and has cost
    $O\left(n^{3/4}\cdot\left(\OPTIR\right)^{1/2}\right)$.
\end{theorem}

When the edges in the graph have weights lower bounded by $\delta > 0$, we apply \Cref{lemma:weighted-upper-technical} with $\gamma = 0$ and choose $\D$ as a degenerate distribution at an independent set of the graph. The proof of \Cref{thm:weighted-upper-special} is in \Cref{app:proof-weighted-upper-special}.

\begin{theorem}[The second upper bound in \Cref{thm:weighted-ternary-opt}]\label{thm:weighted-upper-special}
    In any weighted graph, if every edge has a weight of at least $\delta$, there is a persuasive ternary signaling scheme with a cost of
        $O\left(\left(n\cdot \OPTIR\right)^{2/3}\cdot\delta^{-1/3}\right)$.
\end{theorem}

\subsection{Strict improvement upon full revelation}
In this section, we state and prove a more formal version of \Cref{thm:weighted-improve-ic} regarding strict improvement upon $\OPTIC$.
\begin{theorem}[Formal version of \Cref{thm:weighted-improve-ic}]\label{thm:weighted-strict-improvement}
    In any weighted graphs, if $\OPTIR<\OPTIC$, there exists a persuasive signaling scheme that uses at most $(n+1)$ signals and achieves a cost of
    \begin{align}
        \Cost=\OPTIC-\eps(\OPTIC-\OPTIR),
    \quad\text{where}\ 
    \eps=\frac{\PoS-1}{\PoS-1+\frac{\OPTIR(n-\OPTIR+1)}{\OPTIR+1}}> 0,
        \label{eq:choice-of-eps-weighted}
    \end{align}
    {and $\PoS = \OPTIC / \OPT$ is the price of stability.}
\end{theorem}
Similar to \cref{thm:unit-strict-improvement}, our proof relies on quantifying the \emph{wastefulness} of any feasible solution. We first present a quantitative version of \cref{lem:wasteful-general} in the case of weighted graphs.
\begin{lemma}
    \label[lemma]{lem:wasteful-weighted}
    Let $\btheta$ be a feasible solution {in a general weighted graph}. Then, we have
    \[
    \|W\btheta\|_1\ge n+\frac{\|\btheta\|_1-\OPT}{\OPT}.
    \]
\end{lemma}
The proof of \Cref{lem:wasteful-weighted} (deferred to \Cref{app:proof-wasteful-weighted}) mirrors that of \Cref{lem:wasteful-unit} and follows from a different construction of the feasible dual variables. Leveraging this lemma, \Cref{thm:weighted-strict-improvement} is achieved by randomizing between the independent rounding of the optimal IR solution and the degenerate scheme of deterministically sending the optimal stable solution. We prove this theorem in \Cref{app:proof-weighted-improve-ic}.

    \section{Lower Bounds for Weighted Graphs}
\label{sec:weighted-lower}

In this section, we prove the lower bound part of \Cref{thm:weighted-ternary-opt}. Even when all the edges are of weight $1/2$, there are graphs on which $\OPTIR$ is a constant, yet no signaling scheme achieves a cost lower than $n^{2/3}$. In addition, when restricted to binary signaling schemes, the gap becomes $\Omega(n)$.

We start with the lower bound against binary signaling schemes. In contrast to the unit-weight case, we cannot achieve any non-trivial approximation via a binary signaling scheme as before. This shows that the introduction of a third signal in \Cref{thm:weighted-ternary-opt} is necessary. The proof of the following theorem is in \Cref{app:proof-binary-fails}.

\begin{theorem}\label{thm:weighted-binary-signals-fail}
    There exists a family of weighted graphs on which (1) all edge weights are $1/2$; (2) $\OPTIR = O(1)$; (3) every persuasive binary signaling scheme has a cost of $\Omega(n)$.
\end{theorem}

The rest of this section is devoted to proving the lower bound part of \Cref{thm:weighted-ternary-opt}.

\begin{theorem}[Lower bound part of \Cref{thm:weighted-ternary-opt}]\label{thm:weighted-lower-special}
    There exists a family of weighted graphs on which (1) all edge weights are $1/2$; (2) $\OPTIR = O(1)$; (3) every persuasive signaling scheme has a cost of $\Omega(n^{2/3})$.
\end{theorem}

Recall that when the edge weights are $1/2$, our proof of \Cref{thm:weighted-upper-special} implies a signaling scheme with a cost of
\[
    O\left(\min\left\{\OPTIR\sqrt{n + m}, \frac{n^2}{m}\right\}\right),
\]
where $m$ is the total edge weight in the graph. In order to establish an $\Omega(n^{2/3})$ gap, we need $\OPTIR = \Theta(1)$ and $m = \Theta(n^{4/3})$. This naturally leads to the following construction, which is a modified version of the graph in \Cref{fig:binary-schemes-fail}:
\begin{itemize}
    \item The graph contains $n = k^3 + 2$ vertices: two centers and $k^3$ leaves.
    \item The $k^3$ leaves form $k^2$ disjoint cliques, each of size $k$.
    \item Each of the two centers is adjacent to every other vertex in the graph.
\end{itemize}
We indeed have $\OPTIR = 2 = O(1)$ (achieved by letting both centers play $1$) and $m = \Theta(k^4) = \Theta(n^{4/3})$. In the rest of this section, we show that any persuasive signaling scheme for this graph must have an $\Omega(k^2) = \Omega(n^{2/3})$ cost.

Our proof of \Cref{thm:weighted-lower-special} can be divided into the following three steps. We provide a sketch in the remainder of this section and offer more details in \Cref{app:projection-low-dim,app:test-function,app:dual-feasibility}.
\paragraph{Step 1. Dimensionality reduction.} By \Cref{lem:simplify-signal}, it suffices to consider an identity-independent persuasive signaling scheme specified by $\D_\varphi \in \Delta(\signalspace^V)$ with signal space $\signalspace$. According to \Cref{lemma:persuasive-general}, the persuasiveness of such a signaling scheme is characterized by $|\signalspace|$ constraints on the slack for each signal, which are linear in the $n$-dimensional probability simplex. As a result, the optimal signaling scheme with minimum cost can be characterized by a linear program. In \Cref{app:projection-low-dim}, we show that this is  equivalent to considering the projection of $\D_\varphi$ onto a lower-dimensional probability simplex supported on $k+2$ vertices --- two centers and a random clique of size $k$. As a result, we obtain a $k+2$-dimensional distribution $\D'$ and significantly simplify the constraints for persuasiveness.

\paragraph{Step 2. Choosing test functions.} In the second step, we aim to use the duality of linear programs to lower bound the cost of $\D'$. We consider two choices of test functions $f:\signalspace\to\R_{\ge 0}$: a constant function $f_1(\theta)=1$, and a function that puts more weight on larger signal values $f_2(\theta)=\frac{\theta}{1-\theta}$. Since the second test function is ill-defined at $\theta=1$, we additionally show that it is without loss of generality to consider signaling schemes without the signal $1$, as the restriction of distribution $\D'$ to $\signalspace\setminus\{1\}$ incurs at most a constant blow-up in the cost.
With this, the outputs of the two test functions essentially serve as two sets of dual variables --- they give rise to $\sum_{\theta\in\signalspace}f_1(\theta)\Delta_\theta\ge0$ and $\sum_{\theta\in\signalspace}f_2(\theta)\Delta_\theta=0.$ See more details in \Cref{app:test-function}.

\paragraph{Step 3. Verify the dual feasibility} In the final step, we consider the dual variables $\beta_1 f_1(\theta)-\beta_2f_2(\theta)$ with carefully chosen constants $\beta_1,\beta_2$, and verify the feasibility by establishing the inequality \[
    \Ex{\D'}{\Sum}\ge\sum_{\theta\in\signalspace}(\beta_1 f_1(\theta)-\beta_2f_2(\theta))\Delta_\theta+\Omega(1)\ge\Omega(1),
\]where $\Sum$ is the sum of the signals sent to the randomly selected clique. We prove the above inequality by an involved case analysis in \Cref{app:dual-feasibility}. As a result, this implies that $\Cost(\D_\varphi)\ge k^2\Ex{\D'}{\Sum}\ge\Omega(k^2)=\Omega(n^{\frac{2}{3}})$ and completes the proof of \Cref{thm:weighted-lower-special}.

    \section*{Acknowledgments}
This work was supported in part by the National Science Foundation under grant CCF-2145898, by the Office of Naval Research under grant N00014-24-1-2159, a C3.AI Digital Transformation Institute grant, and Alfred P. Sloan fellowship, and a Schmidt Science AI2050 fellowship.
    
    \newpage
    \bibliographystyle{plainnat}
    \bibliography{arxiv-ref}
    
    \appendix
    
    \section{Omitted Proofs from Section~\ref{sec:model}}
\subsection{Alternative Benchmarks}
Recall that in \Cref{sec:benchmarks-PoS}, we defined the benchmarks assuming agents all know the type profile $\bt=(t_1,\ldots,t_n)$. In this section, we discuss alternative benchmarks for the Bayesian setting where the agents only know the prior distribution $\prior$ of $\bt$ but not its actual realization. In this case, agents evaluate the quality of a collaborative solution in expectation under the type distribution $\prior$.
We start by defining the alternative benchmarks $\barOPT,\barOPT^{\IR},\barOPT^{\IC}$, then show how they can be easily achieved by a no-information signaling scheme and discuss the comparison with the benchmarks defined in \Cref{sec:benchmarks-PoS}.

We define benchmarks $\barOPT$ as the optimal total workload of a feasible action profile $\ba$, and $\barOPT^{\IR},\barOPT^{\IC}$ as the optimal workload under additional IR and stability constraints.
\begin{align*}
    \barOPT=\min_{\ba}\|\ba\|_1 \quad 
    \suchthat \ &
    \Ex{\bt\sim\prior}{\util(\ba;\bt)}\ge1,\ \ba\ge0;\\
    \barOPT^{\IR}=\min_{\ba}\|\ba\|_1 \quad \suchthat \ &
    \Ex{\bt\sim\prior}{\util(\ba;\bt)}\ge1,\ \ba\ge0,\ \ba\le1;\\
    \barOPT^{\IC}=\min_{\ba}\|\ba\|_1 \quad \suchthat \ &
    \Ex{\bt\sim\prior}{\util(\ba;\bt)}\ge1,\ \ba\ge0;\\
    &\forall i\in[n],\ a_i=\min\left\{x\ge0\ \left|\ \Ex{\bt\sim\prior}{\util_i(x,\ba_{-i};\bt)}\ge1\right.\right\}.
    \end{align*}

In is not hard to see that because $\util_i(\bt,\ba)=\Pi W\Pi^{-1}\ba$ for the permutation matrix that corresponds to $\bt$,
the benchmarks $\barOPT,\barOPT^{\IR},\barOPT^{\IC}$ for $W$ equals the corresponding benchmarks under the original definition for a different problem instance $\barW$ with entries 
\[\barW_{ij}=\Ex{\bt\sim \prior}{W_{t_i,t_j}}=\begin{cases}
    1,&i=j,\\
    \frac{\induced(V)-n}{n(n-1)},&i\neq j,
\end{cases}\]
in which the unique optimal solution for all three benchmarks is $\frac{n^2}{\induced(V)}$.
According to \Cref{cor:degenerate-signal}, one can achieve this cost by either incorporating a degenerate signaling scheme or not sending any signals at all. Therefore, we do not use these alternative benchmarks to evaluate our signaling schemes.

Next, we use the comparison of $\OPT$ and $\barOPT$ in unit-weight graphs as an example to show that the two sets of benchmarks are incomparable in general.
\begin{itemize}
    \item Consider a unit-weight graph that consists of a complete subgraph $K_{\frac{n}{2}}$ and $\frac{n}{2}$ isolated vertices.
    In this graph, we have $\OPT=\Theta(n)$ as all isolated vertices have to contribute unit effort. On the other hand, $\barOPT=\Theta(1)$ because $\induced(V)=\Theta(n^2)$. This shows the possibility of having $\barOPT\ll\OPT$.
    \item Consider a star graph with $n-1$ leaves. We have $\OPT=1$, achieved by having the center node contribute unit effort. However, $\barOPT=\Theta(n)$ as $\induced(V)=\Theta(n)$. Therefore, it is also possible to have $\barOPT\gg\OPT$.
\end{itemize}

\subsection{Proof of \Cref{lem:simplify-signal}}
\label{app:identity-independent}
Let $\varphi:\sym(V)\to\Delta(\signalspace^V)$ be a persuasive signaling scheme that is not necessarily identity-independent. Consider the following identity-independent signaling scheme $\tildephi$ that always chooses $\tildephi(\bt)=\D_{\tildephi}$ to be the mixture of $\{\varphi(\bt)\mid \bt\in\sym(V)\}$, i.e., 
\begin{align*}
    \D_{\tildephi}(\bs)=\frac{1}{|\sym(V)|}\sum_{\bt\in\sym(V)}\varphi(\bt)(\bs).
\end{align*}
We show that $\tildephi$ is persuasive. Let $\mu_i,\tildemu_i$ be the posterior distribution of agent $i$ under signaling schemes $\varphi$ and $\tildephi$, respectively.
Since $\tildephi(\bt)$ is defined to be independent of $\bt$, we have
\begin{align*}
    \tildemu_i(\bt,\bs_{-t_i}\mid s_{t_i})
    =\prior(\bt)\cdot\frac{\sum_{\bt'\in\sym(V)}\pr{\varphi(\bt')}{s_{t'_i},\bs_{-t'_i}}}{
    \sum_{\bs'_{-t'_i}}
    \sum_{\bt'\in\sym(V)}\pr{\varphi(\bt')}{s_{t_i'},\bs_{-t_i'}'}
    }
    =\frac{1}{|\sym(V)|}\sum_{\bt'\in\sym(V)}
    \mu_i(\bt',\bs_{-t'_i}\mid s_{t_i'}).
\end{align*}
where the last step follows from $\prior(\bt)=\frac{1}{|\sym(V)|}$ and the definition of $\mu_i$ in \Cref{eq:posterior}.
Therefore, for the expected quality under the posterior distribution $\tildemu_i$, we have
\begin{align*}
    \Ex{(\bt,\bs_{-t_i})\sim\tildemu_i(\cdot\mid \bs_{t_i}=\theta)}{u_{t_i}(\theta,\bs_{-t_i})}=&\frac{1}{|\sym(V)|}\Ex{\bt\sim\tau}{
    \sum_{\bt'\in\sym(V)}
    \Ex{(\bt',\bs_{-t'_i})\sim\mu_i(\cdot\mid \bs_{t_i'}=\theta)}{u_{t_i'}(\theta,\bs_{-t_i'})}
    }\ge1,
\end{align*}
where the inequality is because the persuasiveness of $\varphi$ implies that 
\begin{align*}
    \Ex{(\bt,\bs_{-t_i})\sim\mu_i(\cdot\mid \bs_{t_i}=\theta)}{u_{t_i}(\theta,\bs_{-t_i})}\ge1
\end{align*}
for any agent $i\in[n]$ and any $\theta\in\signalspace$, which holds with equality when $\theta>0$. We have thus established the persuasiveness of $\tildephi$.

Finally, as for the cost, we have
\begin{align*}
    \Cost(\tildephi)=\Ex{\bs\sim\D_{\tildephi}}{\|\bs\|_1}
    =&\frac{1}{|\sym(V)|}\sum_{\bt\in\sym(V)}\Ex{\bs\sim\varphi(\bt)}{\|\bs\|_1}
    =\Ex{\bt\sim\prior,\bs\sim\varphi(\bt)}{\|\bs\|_1}
    =\Cost(\varphi).
\end{align*}
    
    \section{Omitted proofs from Section~\ref{sec:structural}}
\subsection{Proof of \Cref{lemma:persuasive-general}}
\label{app:persuasive-general}
In this subsection, we finish the proof of \Cref{lemma:persuasive-general} by establishing \cref{eq:tmp-cond-exp-slack}, which is restated as follows:
\begin{align*}
    \Ex{(\bt,\bs_{-t_i})\sim\mu_i(\cdot\mid \bs_{t_i}=\theta)}{u_{t_i}(\theta,\bs_{-t_i})}
    =\theta+\frac{\Contrib_\theta}{\Numtheta_\theta}.
\end{align*}
We begin by substituting the expression for the posterior $\mu_i(\cdot\mid \bs_{t_i}=\theta)$ from \cref{eq:posterior}:
\begin{align*}
    &~\sum_{\bt,\bs_{-t_i}}
    \mu_i(\bt,\bs_{-t_i}\mid s_{t_i=\theta})
    \cdot u_{t_i}(\theta,\bs_{-t_i})\stepa{=}
    \sum_{\bt,\bs_{-t_i}}
    \frac{\prior(\bt)\D_{\varphi}({\theta,\bs_{-t_i}})}{
    \sum_{\bt',\bs_{-t_i}'}\prior(\bt')\D_{\varphi}({\theta,\bs_{-t_i}'})
    }
    \cdot u_{t_i}(\theta,\bs_{-t_i})\nonumber\\
    \stepb{=}&~\frac{\sum_{\bt,\bs}\prior(\bt)\D_{\varphi}(\bs)\cdot\indicator{s_{t_i}=\theta}\cdot u_{t_i}(\bs)}{
    \sum_{\bt',\bs'}\prior(\bt')\D_{\varphi}(\bs')\cdot\indicator{s'_{t_i}=\theta}
    }
    =\frac{\Ex{\bs\sim\D_\varphi}{\sum_{\bt}\prior(\bt)\cdot\indicator{s_{t_i}=\theta}\cdot u_{t_i}(\bs)}}{
    \Ex{\bs'\sim\D_\varphi}{\sum_{\bt'}\prior(\bt')\cdot\indicator{s'_{t_i}=\theta}}
    }\nonumber\\
    \stepc{=}&~\frac{\Ex{\bs\sim\D_\varphi}{\sum_{v\in V}\frac{1}{n}\cdot\indicator{s_{v}=\theta}\cdot u_{v}(\bs)}}{
    \Ex{\bs'\sim\D_\varphi}{\sum_{v\in V}\frac{1}{n}\cdot\indicator{s'_{v}=\theta}}
    }
    \stepd{=}\frac{\Ex{\bs\sim\D_\varphi}{\sum_{v\in V}\indicator{s_{v}=\theta}\cdot \left(\theta+\sum_{v'\in N(v)}W_{v,v'}s_{v'}\right)}}{
    \Ex{\bs'\sim\D_\varphi}{\sum_{v\in V}\indicator{s'_{v}=\theta}}
    }\\
    =&~\theta+\frac{\Contrib_\theta}{\Numtheta_\theta}.
\end{align*}
In the above equations, step (a) uses the expression of posterior $\mu_i$ in \cref{eq:posterior} together with the identity-independent property of scheme $\D_\varphi$. Step (b) follows from $\D_{\varphi}({\theta,\bs_{-t_i}})= \sum_{s_{t_i}}\D_\varphi(\bs)\cdot\indicator{s_{t_i}=\theta}$, and that $u_{t_i}(\theta,\bs_{-t_i})=u_{t_i}(\bs)$ when $s_{t_i}=\theta$.
Step (c) is because $\sum_{\bt:t_i=v}\prior(\bt)=\frac{1}{n}$ since the marginal distribution of $\prior$ on $t_i$ is uniform over $V$. Step (d) follows from $u_v(\bs)=\theta+\sum_{v'\in N(v)}W_{v,v'}s_{v'}$ when $s_v=\theta$. The final step uses the linearity of expectations.

\subsection{Proof of \cref{lemma:persuasive-binary}}
\label{app:persuasive-binary}

We prove \Cref{lemma:persuasive-binary} using the characterizations in \Cref{lemma:persuasive-general}. If the signaling scheme $\D_\varphi$ induced by sending $\alpha$ to a random subset of vertices $S\sim\D$ is persuasive, its slack (defined in \Cref{def:slack}) must satisfy $\Delta_0\ge0$ and $\Delta_\alpha=0$ for all $\alpha \in (0, 1]$. We have
\begin{align*}
    \Delta_\alpha&=\Contrib_\alpha-(1-\alpha)\Numtheta_\alpha\\
    &=\Ex{\bs \sim \D_\varphi}{\sum_{v \in V}\1{s_v = \alpha}\left(\sum_{v' \in N(v)}W_{v,v'}s_{v'}-(1-\alpha)\right)}\\
    &=\Ex{S\sim \D}{\sum_{v\in S}\left(\sum_{v'\in S\setminus\{v\}}W_{v,v'}\cdot\alpha-1+\alpha\right)}
    \tag{$s_{v'}=\alpha\cdot\1{v'\in S}$}\\
    &=\Ex{S\sim\D}{\alpha\cdot\sum_{v\in S}\sum_{v'\in S}W_{v,v'}-|S|}
    \tag{$W_{v,v}=1$}\\
    &=\alpha\cdot\Ex{S\sim\D}{\induced(S)}-\Ex{S\sim \D}{|S|}.
\end{align*}
Therefore, to make $\D_\varphi$ persuasive, we should set $\alpha=\frac{\Ex{S \sim \D}{|S|}}{\Ex{S \sim \D}{\induced(S)}}$.

Similarly, for $\Delta_0$, we have
\begin{align*}
    \Delta_0=\Contrib_0-\Numtheta_0
    =\Ex{S\sim \D}{\sum_{v\not\in S}\left(\sum_{v'\in S}\alpha\cdot W_{v,v'}-1\right)}
    =\alpha\cdot\Ex{S\sim \D}{\cut(S,V\setminus S)}-\Ex{S\sim \D}{|V\setminus S|}.
\end{align*}
Plugging in the choice of $\alpha$, to make $\Delta_0\ge0$, the distribution $\D$ should satisfy
\begin{align*}
    \frac{\Ex{S \sim \D}{\cut(S, V \setminus S)}}{\Ex{S \sim \D}{|V \setminus S|}} \ge \frac{\Ex{S \sim \D}{\induced(S)}}{\Ex{S \sim \D}{|S|}},
\end{align*}
which is exactly the inequality in \Cref{lemma:persuasive-binary}. Finally, for the cost of $\D_\varphi$, we have
\begin{align*}
    \Cost(\D_\varphi)=&\alpha\cdot\Ex{S\sim\D}{|S|}=\frac{\left(\Ex{S \sim \D}{|S|}\right)^2}{\Ex{S \sim \D}{\induced(S)}}.
\end{align*}
    
    \section{Details for Section~\ref{sec:discussion}}\label{sec:discussion-details}

\subsection{Conjectured impossibility of matching $\OPTIC$ using $O(1)$ signals.} \label{sec:discussion-details-match-ic}
We conjecture that, in general, $\omega(1)$ different signals are needed to achieve a cost below $\OPTIC$. On a technical level, even if it is possible to achieve a cost of $\OPTIC$ using $O(1)$ different signals, such a result is unlikely to follow from a ``natural'' proof strategy, because signaling schemes with small signal spaces do not \emph{compose} well.

On a graph $G$ that consists of several disjoint components $G_1, G_2, \ldots$, a natural proof strategy would first derive a signaling scheme $\varphi_i$ with a constant-size signal space for each $G_i$ (such that $\Cost(\varphi_i) = \OPTIC(G_i)$), and then appropriately combine them into a scheme for $G$. While the direct product of $\varphi_1, \varphi_2, \ldots$ does give a persuasive scheme for the larger graph $G$, the number of different signals also gets added up, and may scale with the size of $G$.

Another natural idea is to mimic our approach to designing binary schemes --- we view the signaling scheme as a randomized partitioning of $G_i$ into $O(1)$ parts\footnote{For instance, the distribution $\D \in \Delta(2^V)$ that specifies a binary scheme is essentially a randomized $2$-partition of $V$.}, and take the direct product of each $\D_i$ instead. Unfortunately, this breaks even for the case of binary schemes: Recall from \Cref{lemma:persuasive-binary} that $\D_i \in \Delta(2^{V_i})$ gives a persuasive binary scheme for $G_i$ if and only if
\begin{equation}\label{eq:simpson-individual}
    \frac{\Ex{S \sim \D_i}{\cut(S, V_i \setminus S)}}{\Ex{S \sim \D_i}{|V_i \setminus S|}} \ge \frac{\Ex{S \sim \D_i}{\induced(S)}}{\Ex{S \sim \D_i}{|S|}}.
\end{equation}
If we choose $\D$ as the direct product of $\D_1, \D_2, \ldots$, the corresponding condition reduces to
\begin{equation}\label{eq:simpson-combined}
    \frac{\sum_{i}\Ex{S \sim \D_i}{\cut(S, V_i \setminus S)}}{\sum_{i}\Ex{S \sim \D_i}{|V_i \setminus S|}} \ge \frac{\sum_{i}\Ex{S \sim \D_i}{\induced(S)}}{\sum_{i}\Ex{S \sim \D_i}{|S|}}.
\end{equation}
Perhaps counter-intuitively, even if \Cref{eq:simpson-individual} holds for every $i$, the combined inequality in \Cref{eq:simpson-combined} might \emph{not} hold!\footnote{This is similar to \emph{Simpson's paradox} in statistics. For unit-weight graphs, we avoided this issue by showing that the left-hand side of \Cref{eq:simpson-individual} is always lower bounded by $1$, which in turn upper bounds the right-hand side. This stronger condition \emph{does} imply \Cref{eq:simpson-combined} when multiple components are combined.}

Based on this technical insight, we found a simple instance on which a numerical optimization over ternary signaling schemes suggests that three different signals are not enough to achieve a cost of $\OPTIC$. We consider a graph that consists of multiple connected components $G_1, G_2, \ldots$ with vertex sets $V_1, V_2, \ldots$. Each graph $G_i$ admits a persuasive binary signaling scheme with cost $\OPTIC(G_i)$ induced by a distribution $\D_i \in \Delta(2^{V_i})$ that satisfies \Cref{eq:simpson-individual}. Towards ensuring that the direct product of $\D_1, \D_2, \ldots$ does not give a persuasive signaling scheme, we choose these components such that: (1) \Cref{eq:simpson-individual} is tight for each $i$; (2) the two sides of \Cref{eq:simpson-individual} take different values for different values of $i$.

Concretely, we consider the graph
that consists of:
\begin{itemize}
    \item \textbf{Component 1:} Three vertices that form a path with two unit-weight edges.
    \item \textbf{Component 2:} Two additional vertices connected by an edge of weight $1/2$.
\end{itemize}
In the optimal stable solution, the vertex in the middle of the $3$-path plays $1$, while both vertices in the other component play $2/3$. Recall that a binary signaling scheme is specified by a distribution over $2^V$, which is a $31$-dimensional simplex on this specific graph. Therefore, the optimal binary signaling scheme is characterized by a constrained continuous optimization problem on this simplex, in which the constraint is given by \Cref{lemma:persuasive-binary}. We relax this optimization problem to multiple linear programs, and numerically verified that no binary signaling scheme achieves a cost of $\OPTIC$ or lower.

The same idea gives another graph
that consists of:
\begin{itemize}
    \item \textbf{Component 1:} Three vertices that form a path with two unit-weight edges.
    \item \textbf{Component 2:} Two vertices connected by an edge of weight $2/3$.
    \item \textbf{Component 3:} Three vertices that form a triangle with three edges of weights $1/2$, $3/4$, and $3/4$.
\end{itemize}
A similar but more complicated numerical optimization on ternary signaling schemes indicates that no ternary scheme achieves a cost of $\OPTIC$ or lower.

Despite this negative evidence, it \emph{might} still be possible to achieve an $O(\OPTIC)$ cost using $O(1)$ signals. However, such a constant factor approximation, in conjunction with our proof of \Cref{thm:weighted-improve-ic}, would not guarantee a strict improvement upon $\OPTIC$.

\subsection{Conjectured $\Omega(n^{3/4})$ lower bound}\label{sec:discussion-details-approx}
We conjecture that the $n^{3/4}$ ratio is tight on the following instance: The graph consists of two centers connected by a unit-weight edge, along with $n - 2$ leaves that form a clique with edge weight $n^{-3/4}$. In addition, each center is connected to each leaf by an edge of weight $1/2$.

When only the centers play $1$ each, we obtain a feasible IR solution with cost $2 = O(1)$, whereas our approach at best gives an $O(n^{3/4})$ cost on this graph. For a sanity check, note that the total edge weight in the graph is $m = \Theta(n^2\cdot n^{-3/4} + n) = \Theta(n^{5/4})$. Our proof of \Cref{thm:weighted-ternary-opt} (more specifically, \Cref{thm:weighted-upper-general}) gives a persuasive signaling scheme with cost upper bounded by a constant factor times
\[
    \min\left\{(\OPTIR)^{2/3}m^{1/3}n^{1/3} + \OPTIR\cdot\sqrt{n}, \frac{n^2}{m}\right\},
\]
yet both terms reduce to $\Theta(n^{3/4})$ when we plug in $\OPTIR = \Theta(1)$ and $m = \Theta(n^{5/4})$.

The graph above resembles the hard instances for the other cases on which we manage to prove tight lower bounds, so why is it harder to prove an $\Omega(n^{3/4})$ bound? Recall from \Cref{sec:lower_overview} that we prove lower bounds by carefully choosing a dual solution $f: [0, 1] \to \R$ to the LP that characterizes the optimal scheme. On the lower bounds that we managed to prove, we could choose $f$ as a linear combination of the constant function and the function $f(\theta) = \theta / (1 - \theta)$. This choice gives a simple closed-form expression, and satisfies concavity and Lipschitz continuity (except when $\theta$ is close to $1$). These make the choice of $f$ relatively amenable to formal proofs. However, in the hard instance defined above, a numerical computation suggests that the optimal dual solution is much more ill-behaved --- it is non-convex, non-concave, and has a large derivative around $1/2$. Therefore, we expect that, even if the instance above does witness an $\Omega(n^{3/4})$ bound, the proof would involve a much more complicated choice of $f$, and a more involved analysis.
    
    \section{Omitted Proofs from Section~\ref{sec:unit-weight-upper}}\label{sec:unit-weight-upper-omitted}

\subsection{Technical Lemmas in Section~\ref{sec:weight-induced-size}}
We prove \Cref{lemma:IS-cut,lemma:IR-sol-rounding} stated in \Cref{sec:weight-induced-size}.

\begin{proof}[Proof of \Cref{lemma:IS-cut}]
    Since $\DS$ and $V \setminus (\IS \cup \DS)$ partitions $V \setminus \IS$, we have
    \[
        \cut(\IS, V \setminus \IS)
    =   \cut(\IS, \DS) + \cut(\IS, V \setminus (\IS \cup \DS)).
    \]
    Since $\DS$ is a dominating set of $G$, every vertex in $\IS \subseteq V \setminus \DS$ has a neighbour in $\DS$, so we have $\cut(\IS, \DS) \ge |\IS|$. Since $\IS$ is maximal with respect to the induced sub-graph of $V \setminus \DS$, every vertex $v \in V \setminus (\DS \cup \IS)$ has a neighbour in $\IS$; otherwise, $v$ can be added to $\IS$ to give a larger independent set. This gives $\cut(\IS, V \setminus (\IS \cup \DS)) \ge |V \setminus (\IS \cup \DS)|$. Therefore, we conclude that
    \[
        \cut(\IS, V \setminus \IS)
    \ge |\IS| + |V \setminus (\IS \cup \DS)|
    =   |V| - |\DS|.\qquad\qedhere
    \]
\end{proof}

\begin{proof}[Proof of \Cref{lemma:IR-sol-rounding}]
    The first bound follows easily from the fact that each vertex is included independently with probability $\theta_v^\star$. For the second bound, note that
    \[
        \Ex{S \sim \D}{\induced(S)}
    \le \sum_{u, v \in V}\pr{S \sim \D}{u, v \in S}.
    \]
    By definition of $\D$, $\pr{S \sim \D}{u, v \in S}$ is given by $\theta_u\theta_v$ if $u \ne v$, and $\theta_u$ if $u = v$.
    It follows that
    \[
        \Ex{S \sim \D}{\induced(S)}
    \le \sum_{u, v \in V: u \ne v}\theta_u\theta_v + \sum_{u \in V}\theta_u
    \le \left(\sum_{u \in V}\theta_u\right)^2 + \left(\sum_{u \in V}\theta_u\right)
    =   \|\btheta\|_1^2 + \|\btheta\|_1.
    \]

   For the third bound, note that
    \[
        \Ex{S \sim \D}{\cut(S, V \setminus S)}
    =   \sum_{u, v \in V}W_{u,v}\cdot\pr{S \sim \D}{u \notin S \wedge v \in S}
    =   \sum_{u \in V}\sum_{v \in N(u)}W_{u,v}\cdot(1 - \theta_u)\cdot \theta_v.
    \]
    For each $u \in V$, the condition $W\btheta \ge \vecone$ implies $\sum_{v \in N(u)}W_{u,v}\theta_v \ge 1-\theta_u$,
    and it follows that
    \begin{align*}
        \Ex{S \sim \D}{\cut(S, V \setminus S)}
    &=  \sum_{u \in V}(1 - \theta_u)\sum_{v \in N(u)}W_{u,v}\theta_v
    \ge\sum_{u \in V}(1 - \theta_u)^2
    \ge \sum_{u \in V}(1 - 2\theta_u)
    =   |V| - 2\|\btheta\|_1.
    \end{align*}
\end{proof}

\subsection{Proof of \Cref{lem:wasteful-unit}}
\label{app:wasteful-proof-unit}
\begin{proof}[Proof of \Cref{lem:wasteful-unit}]
    Recall from the proof of \Cref{lem:wasteful-general} that $\OPT$ is characterized by the linear program in \eqref{eq:primal-lp}, and its dual is given by \eqref{eq:dual-lp}. In the following, we construct a feasible solution $\dualvar^\star$ to the dual LP, such that $\vecone^{\top} (\btheta-\dualvar^\star)\le \|W\btheta\|_1-n$. By the weak duality of LP, we have $\vecone^{\top}\dualvar^\star \le \OPT$. It follows that
    \[
        \|\btheta\|_1 - \OPT
    \le \vecone^{\top}\btheta - \vecone^{\top}\dualvar^\star
    \le \|W\btheta\|_1 - n,
    \]
    which is equivalent to the desired bound.

    Now we show how to construct $\dualvar^\star$. Without loss of generality, the vertices of the graph are labeled with $V = [n]$. Let $\btheta^{(0)}=\btheta$.
    For $i = 1, 2, \ldots, n$, we examine the $i$-th coordinate of $W\btheta^{(i-1)}$. If $(W\btheta^{(i-1)})_i\le1$, we set $\btheta^{(i)}=\btheta^{(i-1)}$ and move on to the next coordinate. If $(W\btheta^{(i-1)})_i>1$, we choose vector $\btheta^{(i)}$ such that:
    \begin{itemize}
        \item $\theta^{(i)}_i + \sum_{j \in N(i)}\theta^{(i)}_j = 1$, and $\theta^{(i)}_j \in \left[0, \theta^{(i-1)}_j\right]$ holds for every $j \in N(i)\cup\{i\}$. This is possible since $\theta^{(i-1)}_i + \sum_{j \in N(i)}\theta^{(i-1)}_j
    =   (W\btheta^{(i-1)})_i
    >   1$.
        \item $ \theta^{(i)}_j = \theta^{(i-1)}_j$ for each $j \notin N(i) \cup \{i\}$.
    \end{itemize}
    If $(W\btheta^{(i-1)})_i\le1$, we immediately obtain $(W\btheta^{(i)})_i = (W\btheta^{(i-1)})_i\le1$; otherwise, our choice of $\btheta^{(i)}$ guarantees
    \[
        (W\btheta^{(i)})_i
    =   \theta^{(i)} + \sum_{j \in N(i)}\theta^{(i)}_j = 1,
    \]
    Thus, in both cases, the dual constraint $W\dualvar \le \vecone$ is satisfied by $\dualvar = \btheta^{(i)}$ on the $i$-th coordinate.
    Furthermore, it is clear from our procedure that $\veczero \le \btheta^{(n)} \le \btheta^{(n-1)} \le \cdots \le \btheta^{(0)} = \btheta$ holds coordinate-wise. Since all entries in $W$ are non-negative, $W\btheta^{(n)} \le W\btheta^{(n-1)} \le \cdots \le W\btheta^{(0)} = W\btheta$ also holds coordinate-wise.
    
    Therefore, if we let $\dualvar^\star=\btheta^{(n)}$, it holds for every $i \in [n]$ that
    \[
        (W\dualvar^\star)_i
    =   (W\btheta^{(n)})_i
    \le (W\btheta^{(i)})_i
    =   1.
    \]
    In other words, $\dualvar^\star$ is a feasible solution to \eqref{eq:dual-lp}. Moreover, for each $i \in [n]$, if $\btheta^{(i)}$ is different from $\btheta^{(i-1)}$, we have
    \[
        \vecone^{\top}\left(\btheta^{(i-1)}-\btheta^{(i)}\right)
    =   \left(\theta^{(i-1)}_i + \sum_{j \in N(i)}\theta^{(i-1)}_j\right) - \left(\theta^{(i)}_i + \sum_{j \in N(i)}\theta^{(i)}_j\right)
    =   (W\btheta^{(i-1)})_i-1
    \le (W\btheta)_i-1.
    \]
    If $\btheta^{(i)}$ remains the same as $\btheta^{(i-1)}$, we trivially have
    \[
        \vecone^{\top}\left(\btheta^{(i-1)}-\btheta^{(i)}\right)
    =   0
    \le (W\btheta)_i-1,
    \]
    where the last step follows from the feasibility of $\btheta$. Therefore, we have $\vecone^{\top}\left(\btheta^{(i-1)}-\btheta^{(i)}\right) \le (W\btheta)_i-1$ in both cases, and the desired bound follows from
    \begin{align*}
         \vecone^{\top}\left(\btheta-\dualvar^\star\right)= \sum_{i=1}^n \vecone^{\top}\left(\btheta^{(i-1)}-\btheta^{(i)}\right)\le\sum_{i=1}^{n} [(W\btheta)_i-1]=\|W\btheta\|_1-n.
    \end{align*}
\end{proof}

\section{Omitted Proofs from Section~\ref{sec:unit-lower}}
\label{app:unit-lower}

\begin{proof}[Proof of \Cref{thm:unit-lb-general}]
    We start with the case that $k = O(\sqrt{n})$. The $k \gg \sqrt{n}$ case would easily follow from the lower bound for the first case.

    \paragraph{The $k = O(\sqrt{n})$ case.} Concretely, we assume that $k \le \sqrt{n} / 2$. Consider the following graph: There are $k$ star graphs, each with a center connected to $l \coloneqq \lfloor\frac{n}{k}\rfloor - 1$ leaves. The $k$ centers of the stars form a clique. For this graph, we have $\OPT=k$ (achieved when each of the $k$ centers contributes $1$) and $\OPTIC = \Omega(n)$ (e.g., achieved when the center in one of the stars and the leaves in all the other stars contribute $1$ each).
    We will show that every persuasive signaling scheme must have a cost of $\Omega(k\sqrt{n}) = \Omega(\min\{k\sqrt{n}, n\})$.

Again, we assume that there is a signaling scheme with cost $\le k\sqrt{n} / C$ where $C \coloneqq 20$, and we will still derive a lower bound of $\Omega(k\sqrt{n})$ on the cost under this assumption. By \Cref{lem:simplify-signal}, without loss of generality, the signaling scheme is identity-independent and is specified by a distribution $\D_\varphi \in \Delta(\signalspace^V)$. In other words, the scheme samples $((s_{c,1},\ldots,s_{c,k}),(s_{i,j})_{i\in[k],j\in[l]})\in \signalspace^V$ from $\D_\varphi$, labels the center of the $i$-th star with $s_{c,i}$ and labels the $j$-th leaf in the $i$-th star with $s_{i,j}$. For brevity, we shorthand $x_i$ for $s_{c,i}$, and let $\Sum \coloneqq \sum_{i=1}^{k}x_i + \sum_{i=1}^{k}\sum_{j=1}^{l}s_{i,j}$ denote the sum of the signals.

We first use the persuasiveness of $\D_\varphi$ to derive a lower bound on $\Ex{\D_\varphi}{\sum_{i=1}^{k}x_i}$, the expected total signal received by the $k$ centers. By \Cref{lemma:persuasive-general}, it holds for every $\theta \in \signalspace$ that
\[
0 \le 
\Delta_\theta=\Ex{\D_\varphi}{\sum_{v \in V}\1{s_v=\theta}\left(\sum_{u\in N(v)}s_u-(1-s_v)\right)}.
\]
Summing over $\theta \in \signalspace$ gives
\[
0 \le \sum_{\theta\in\signalspace}\Delta_\theta=\Ex{\D_\varphi}{\sum_{v \in V}\left(s_v+\sum_{u\in N(v)}s_u-1\right)}
=\Ex{\D_\varphi}{\sum_{u \in V}s_u\left(1+|N(u)|\right)}-k(l+1).
\]
Since $|N(u)|=1$ when $u$ is a leaf and $|N(u)|=l+k-1$ when $u$ is a center, we have
\[
k(l+1)
\le
\Ex{\D_\varphi}{\sum_{u \in V}s_u\left(1+|N(u)|\right)}=
2\Ex{\D_\varphi}{\Sum}+(l+k-2)\cdot\Ex{\D_\varphi}{\sum_{i=1}^k x_i}.
\]
Applying the assumption that $\Ex{\D_\varphi}{\Sum} \le k\sqrt{n}/C$ gives the desired lower bound on $\Ex{\D_\varphi}{\sum_{i=1}^k x_i}$:
\[
    \Ex{\D_\varphi}{\sum_{i=1}^k x_i}\ge\frac{k(l+1)-2\Ex{\D_\varphi}{\Sum}}{l+k-2}
\ge \frac{k(l+1)-2k\sqrt{n}/C}{l+k}.
\]

Let $y \coloneqq \frac{\sqrt{n}}{\sqrt{n}+2k}$ and define the random variable $X \coloneqq \sum_{i=1}^k \1{x_i \ge y}$. Next, we use the lower bound on $\Ex{\D_\varphi}{\sum_{i=1}^k x_i}$ to lower bound the expectation of $X$ by $\Omega(k)$. Applying the inequality $\1{x_i \ge y} \ge \frac{x_i - y}{1 - y}$ gives
\begin{align*}
    \Ex{\D_\varphi}{X}
&\ge \frac{\Ex{\D}{\sum_{i=1}^k x_i}-ky}{1-y}
\ge \frac{\frac{k(l+1) - 2k\sqrt{n}/C}{l + k} - ky}{1-y}\\
&=  \frac{kl(1-y) + k - 2k\sqrt{n}/C - yk^2}{(1-y)(l+k)}\\
&\ge\frac{kl}{l + k} - \frac{2k\sqrt{n}}{C(1-y)(l+k)} - \frac{yk^2}{(1-y)(l+k)}.
\end{align*}
Plugging $y = \frac{\sqrt{n}}{\sqrt{n} + 2k}$ into the last term above gives
\[
    \frac{kl}{l + k} - \frac{n + 2k\sqrt{n}}{C(l+k)} - \frac{k\sqrt{n}}{2(l+k)}.
\]
Recall that $k \le \sqrt{n} / 2$, $l = \lfloor n/k\rfloor - 1 = (n/k)\cdot(1 + o_n(1))$ and $C = 20$. For all sufficiently large $n$, the three terms above can be bounded as follows:
\[
    \frac{kl}{l+k}
=   \frac{k}{1 + k/l}
=   \frac{k}{1 + k^2/n\cdot(1 + o_n(1))}
\ge \frac{k}{1 + 1/4 + o_n(1)} \ge \frac{k}{2},
\]
\[
    \frac{n + 2k\sqrt{n}}{C(l+k)} \le \frac{2n}{Cl} = \frac{k}{10}(1+o_n(1))\le \frac{k}{8},
\]
and
\[
    \frac{k\sqrt{n}}{2(l+k)}
=   \frac{k}{2(\sqrt{n}/k+k/\sqrt{n})}\cdot(1+o_n(1))
\le \frac{k}{2\cdot(2 + 1/2)}\cdot(1+o_n(1))
\le \frac{k}{4}.
\]
Therefore, we conclude that $\Ex{\D_\varphi}{X} \ge k/2 - k/8 - k/4 = k/8$.

Since $k - X$ is always non-negative, by Markov's inequality, we have
\[
    \pr{\D_\varphi}{X\le\frac{k}{10}}=\pr{\D_\varphi}{k-X\ge\frac{9k}{10}}\le\frac{\Ex{\D_\varphi}{k-X}}{9k/10}\le
    \frac{7k/8}{9k/10}
    =   \frac{35}{36},
\]
which implies $\pr{\D_\varphi}{X\ge\frac{k}{10}}\ge\frac{1}{36}$.

Applying \Cref{lemma:persuasive-general} again gives
\begin{equation}\label{eq:unit-lb-general}
\begin{split}
    0
=   \sum_{\theta \in \signalspace}\Delta_\theta\cdot\1{\theta \ge y}
&=  \Ex{\D_\varphi}{\sum_{v}\1{s_v\ge y}\left(\sum_{u\in N(v)}s_u-(1-s_v)\right)}\\
&\ge\Ex{\D_\varphi}{\sum_{i=1}^k \1{x_i \ge y}\cdot\left(\sum_{i=1}^k x_i-1\right)+\sum_{i=1}^k\sum_{j=1}^{l}\1{s_{i,j}\ge y}(y-1)}\\
&=  \Ex{\D_\varphi}{X\left(\sum_{i=1}^k x_i-1)\right)}-(1-y)\Ex{\D_\varphi}{\sum_{i=1}^{k}\sum_{j=1}^{l}\1{s_{i,j}\ge y}}.
\end{split}
\end{equation}
In the second line above, we drop the contributions from the leaves when $x_i \ge y$, and drop the contribution from the incident center when a leaf satisfies $s_{i,j}\ge y$.

In the rest of the proof, we will first lower bound the $\Ex{\D_\varphi}{X\left(\sum_{i=1}^k x_i-1)\right)}$ term. By \Cref{eq:unit-lb-general}, this gives a lower bound on $(1-y)\Ex{\D_\varphi}{\sum_{i=1}^{k}\sum_{j=1}^{l}\1{s_{i,j}\ge y}}$, which, in turn, lower bounds the cost.

Note that
\[
    X\cdot\left(\sum_{i=1}^k x_i-1\right)
\ge X\cdot\left(\sum_{i=1}^k y\cdot\1{x_i \ge y}-1\right)
=   X\cdot(yX - 1),
\]
and the minimum of the quadratic function $x \mapsto yx^2 - x$ is $-1/(4y)$. By the assumption that $k \le \sqrt{n} / 2$, we have $y = \frac{\sqrt{n}}{\sqrt{n}+2k} \in [1/2, 1]$, which implies $X\cdot\left(\sum_{i=1}^k x_i-1\right) \ge -1/(4y) \ge -1/2$. When $X \ge k / 10$, we have a stronger lower bound of $X \cdot \left(\sum_{i=1}^{k}x_i - 1\right) \ge (k/10)\cdot [y\cdot(k/10)-1]$. It follows that
\begin{align*}
    \Ex{\D_\varphi}{X\left(\sum_{i=1}^k s_i-1)\right)}
&\ge\Ex{\D_\varphi}{\1{X\ge\frac{k}{10}}\cdot\frac{k}{10}\cdot\left(\frac{ky}{10}-1\right) + \1{X<\frac{k}{10}}\cdot (-1/2)}\\
&\ge\Omega(k^2)\cdot\pr{\D_\varphi}{X\ge\frac{k}{10}} - \frac{1}{2}
\ge \Omega(k^2).
\end{align*}
The last step follows from the inequality $\pr{\D_\varphi}{X \ge k / 10} \ge 1/36$ that we proved earlier. Plugging the above into \Cref{eq:unit-lb-general} gives $(1-y)\Ex{\D_\varphi}{\sum_{i=1}^{k}\sum_{j=1}^{l}\1{s_{i,j}\ge y}} \ge \Omega(k^2)$.

Finally, we conclude that
\begin{align*}
    \Ex{\D_\varphi}{\Sum} \ge
    y\cdot \Ex{\D_\varphi}{\sum_{i=1}^{k}\sum_{j=1}^{l}\1{s_{i,j}\ge y}}\ge\Omega(k^2)\cdot\frac{y}{1-y}\ge\Omega(k\sqrt{n}),
\end{align*}
where the last step follows from $y = \frac{\sqrt{n}}{\sqrt{n} + 2k}$. This completes the proof for the $k \le \sqrt{n}/2$ case.

\paragraph{The $k \gg \sqrt{n}$ case.} We handle the $k > \sqrt{n} / 2$ case by a reduction to the first case. Let $k' \coloneqq \lfloor n / (4k)\rfloor$ and $n' = 4(k')^2$. Since $k' \le \sqrt{n'} / 2$, our proof for the first case gives a graph $G'$ with $n'$ vertices, on which $\OPT = k'$, $\OPTIC = \Omega(n')$, and every persuasive signaling scheme has an $\Omega(k'\sqrt{n'}) = \Omega(n')$ cost.

Let $m \coloneqq \lfloor n / n'\rfloor = \Theta(k^2/n)$. Consider the graph $G$ that consists of $m$ disjoint copies of $G'$, denoted by $G'_1$ through $G'_m$. It is clear that the benchmarks $\OPT$ and $\OPTIC$ are additive on a graph consisting of multiple connected components. Thus, on the graph $G$, we have $\OPT = m \cdot k' = \Theta(k)$ and $\OPTIC = m \cdot \Omega(n') = \Omega(n)$.

It remains to show that the lower bound on the cost of persuasive signaling schemes also composes. Suppose that $\D_\varphi$ is a persuasive signaling scheme for graph $G$, with a cost of $C$. Consider the following induced signaling scheme $\D'_\varphi$ for graph $G'$: Pick $i \in [m]$ uniformly at random, draw $s \sim \D_\varphi$, and choose the signal as the restriction of $s$ to $G'_i$. Let $\Delta$ and $\Delta'$ denote the slacks induced by the signaling schemes $\D_\varphi$ and $\D'_\varphi$, respectively. Note that $\D'_\varphi$ is still persuasive, since for every $\theta\in[0,1]$, we have
\begin{align*}
    \Delta'_{\theta} = \frac{1}{m}\sum_{i=1}^m\Ex{\D_\varphi}{\sum_{v\in G'_i}\1{s_v=\theta}\left(\sum_{u\in N(v)}s_u-(1-s_v)\right)}=\frac{\Delta_\theta}{m} \ge 0,
\end{align*}
and the inequality is tight for all $\theta > 0$.
Moreover, the cost of the signaling scheme $\D'_\varphi$ is given by $C / m$.

Now, applying the lower bound for graph $G'$ shows that $C / m \ge \Omega(k'\sqrt{n'}) = \Omega(n')$. Therefore, we have $C \ge \Omega(mn') = \Omega(n)$. In other words, every persuasive signaling scheme must have an $\Omega(n) = \Omega(\min\{k\sqrt{n}, n\})$ cost on graph $G$. This completes the proof.
\end{proof}
    
    \section{Omitted Proofs from Section~\ref{sec:weighted-upper}}
\label{app:weighted-upper}

\subsection{Proof of \Cref{lemma:weighted-upper-technical}}
\label{app:proof-weighted-technical}
\begin{proof}[Proof of \Cref{lemma:weighted-upper-technical}.]
    Let $\btheta^\star \in [0, 1]^V$ be an optimal solution subject to IR with cost $\|\btheta^\star\|_1 = \OPTIR$. Consider the following signaling scheme with parameters $p, q \ge 0$ and $\eps, \alpha \in [0, 1]$.
    \begin{itemize}
        \item With probability $\frac{1}{1 + p + q}$, draw a random set $\DStilde \subseteq V$ that includes each $v \in V$ independently with probability $\theta^\star_v$. Label (the vertices in) $\DStilde$ with $1 - \eps$, and label $V \setminus \DStilde$ with $0$.
        
        \item With probability $\frac{p}{1 + p + q}$, draw a random set $\IStilde \sim \D$. Label $\IStilde$ with $1 - \eps$, and label $V \setminus \IStilde$ with $0$.
        
        \item With probability $\frac{q}{1 + p + q}$, draw a random set $\DStilde \subseteq V$ that includes each $v \in V$ independently with probability $\theta^\star_v$. Label $V \setminus \DStilde$ with $\alpha$, and label $\DStilde$ with $0$.
    \end{itemize}

    \paragraph{Properties of $\DStilde$ and $\IStilde$.} We will use the following properties of $\DStilde$ and $\IStilde$:
    \begin{itemize}
        \item $\Ex{}{|\DStilde|} = \|\btheta^\star\|_1 = \OPTIR$.
        \item $\Ex{}{\induced(\DStilde)} \le (\OPTIR)^2 + \OPTIR$.
        \item $\Ex{}{\cut(\DStilde, V \setminus \DStilde)} \ge n - 2\OPTIR$.
        \item $\Ex{}{\induced(\IStilde)} \le (1 + \gamma)\Ex{}{|\IStilde|}$.
    \end{itemize}
    The first three bounds follow from \Cref{lemma:IR-sol-rounding}. The last follows from the assumption on $\D$.

    \paragraph{Conditions for persuasiveness.} By \Cref{def:slack} and \Cref{lemma:persuasive-general}, in order for the signaling scheme above to be persuasive, we need to satisfy:
    \begin{itemize}
        \item For $\theta \in \{\alpha, 1 - \eps\}$,
        \begin{equation}\label{eq:condition-general}
        (1-\theta)\cdot\Numtheta_\theta=\Contrib_\theta
        \end{equation}
        \item For $\theta = 0$, \Cref{eq:condition-general} holds with ``$=$'' replaced with ``$\le$''.
    \end{itemize}
    In the rest of the proof, we carefully pick the parameters of the signaling scheme to satisfy the conditions above, while ensuring that the resulting cost satisfies the desired upper bound.

    \paragraph{Pick $\alpha$ to handle $\theta = \alpha$.} We first examine \Cref{eq:condition-general} when $\theta = \alpha$. Since we send $\alpha$ to $V \setminus \DStilde$ with probability $\frac{q}{1+p+q}$, the left-hand side of \Cref{eq:condition-general} is given by
    \[
        (1 - \alpha)\cdot\frac{q}{1+p+q}\Ex{}{|V \setminus \DStilde|} = (1 - \alpha)\cdot\frac{q}{1+p+q}(n - \OPTIR).
    \]
    The right-hand side, on the other hand, is given by
    \[
        \frac{q}{1+p+q}\cdot\alpha\cdot\Ex{}{\induced(V\setminus \DStilde) - (n - |\DStilde|)}
    =   \frac{q}{1+p+q}\cdot\alpha\cdot\left[\Ex{}{\induced(V\setminus \DStilde)} - (n - \OPTIR)\right].
    \]
    Therefore, for $\theta = \alpha$, \Cref{eq:condition-general} reduces to
    \[
       (1 - \alpha)(n - \OPTIR) = \alpha\cdot[\Ex{}{\induced(V\setminus \DS)} - (n - \OPTIR)],
    \]
    which holds if we pick $\alpha = \frac{n - \OPTIR}{\Ex{}{\induced(V \setminus \DStilde)}}$. Note that this choice of $\alpha$ is valid, since we have
    \[
    \Ex{}{\induced(V \setminus \DStilde)} \ge \Ex{}{|V \setminus \DStilde|} = n - \OPTIR \ge 0,
    \]
    which implies $\alpha \in [0, 1]$.

    \paragraph{Pick $\eps$ to handle $\theta = 1 - \eps$.} When $\theta = 1 - \eps$, the left-hand side of~\Cref{eq:condition-general} is
    \[
        \eps\cdot\left[\frac{1}{1+p+q}\cdot\Ex{}{\left|\DStilde\right|} + \frac{p}{1+p+q} \cdot \Ex{}{|\IStilde|}\right]
    =   \frac{1}{1+p+q}\cdot\eps\cdot\left(\OPTIR + p\Ex{}{|\IStilde|}\right).
    \]
    The right-hand side is equal to
    \[
        (1 - \eps)\cdot\left[\frac{1}{1+p+q}\cdot\Ex{}{\induced(\DStilde) - |\DStilde|} + \frac{p}{1+p+q}\cdot\Ex{}{\induced(\IStilde) - |\IStilde|}\right].
    \]
    Then, \Cref{eq:condition-general} at $\theta = 1 - \eps$ is equivalent to
    \[
        \eps\cdot\left(\OPTIR + p\Ex{}{|\IStilde|}\right)
    =   (1 - \eps)\cdot\left[\Ex{}{\induced(\DStilde) - |\DStilde|} + p\cdot\Ex{}{\induced(\IStilde) - |\IStilde|}\right],
    \]
    which is satisfied if we pick
    \[
        \eps = \frac{\Ex{}{\induced(\DStilde)} - \OPTIR + p\cdot\Ex{}{\induced(\IStilde) - |\IStilde|}}{\Ex{}{\induced(\DStilde)} + p\cdot\Ex{}{\induced(\IStilde)}} \in [0, 1].
    \]

    \paragraph{Pick $q$ to handle $\theta = 0$ case.} Finally, at $\theta = 0$, the left-hand side of \Cref{eq:condition-general} is equal to
    \begin{align*}
        &~\frac{1}{1+p+q}\cdot\left[\Ex{}{|V\setminus\DStilde|} + p\cdot\Ex{}{|V\setminus \IStilde|} + q\cdot\Ex{}{|\DStilde|}\right]\\
    =   &~\frac{1}{1+p+q}\cdot\left[(n - \OPTIR) + p\left(n - \Ex{}{|\IStilde|}\right) + q\OPTIR\right].
    \end{align*}
    The right-hand side is given by
    \[
        \frac{1}{1+p+q}\cdot\left[
        (1 - \eps)\cdot\Ex{}{\cut(\DStilde, V \setminus \DStilde)} + p\cdot(1-\eps)\cdot\Ex{}{\cut(\IStilde, V \setminus \IStilde)} + q\cdot\alpha\cdot \Ex{}{\cut(\DStilde, V\setminus\DStilde)}\right].
    \]
    Recall that we have $\Ex{}{\cut(\DStilde, V \setminus \DStilde)} \ge n - 2\OPTIR$. We also trivially relax the $\Ex{}{\cut(\IStilde, V \setminus \IStilde)}$ term to $0$. Recall that at $\theta = 0$, it suffices for the left-hand side of \Cref{eq:condition-general} to be upper bounded by the right-hand side. This can be guaranteed if we have
    \[
        (n - \OPTIR) + p\left(n - \Ex{}{|\IStilde|}\right) + q\OPTIR
    \le (1 - \eps)\cdot(n - 2\OPTIR) + q\cdot\alpha\cdot (n - 2\OPTIR),
    \]
    which is equivalent to
    \[
        q\cdot[\alpha(n-2\OPTIR) - \OPTIR]
    \ge (n - \OPTIR) + p\left(n - \Ex{}{|\IStilde|}\right) - (1 - \eps)\cdot(n - 2\OPTIR).
    \]
    We claim that it is without loss of generality to assume that $\alpha(n-2\OPTIR) - \OPTIR \ge \alpha n / 2$; otherwise, as we show at the end of the proof, a cost of $O(\OPTIR)$ can be trivially achieved. Under this additional assumption, it is, in turn, sufficient to satisfy
    \[
        \frac{1}{2}q\alpha n \ge \eps n + (1 - 2\eps)\OPTIR + p\left(n - \Ex{}{|\IStilde|}\right).
    \]
    Therefore, we can always pick $q$ to satisfy the condition at $\theta = 0$, such that $q\alpha n$ is at most $O(\OPTIR + (\eps + p)n)$.

    \paragraph{Upper bound the cost by optimizing $p$.} The cost of our signaling scheme is clearly
    \[
        \frac{1}{1+p+q}\cdot\left[(1-\eps)\Ex{}{|\DStilde|} + p\cdot(1-\eps)\Ex{}{|\IStilde|} + q\cdot\alpha\Ex{}{|V \setminus \DStilde|}\right]
    \le \OPTIR + p\Ex{}{|\IStilde|} + q\alpha n.
    \]
    By our choice of $q$, the $q\alpha n$ term is at most $O(\OPTIR + (\eps + p)n)$, so the cost is also upper bounded by $O(\OPTIR + (\eps + p)n)$.

    Now, recall our choice of
    \[
        \eps = \frac{\Ex{}{\induced(\DStilde)} - \OPTIR + p\cdot\Ex{}{\induced(\IStilde) - |\IStilde|}}{\Ex{}{\induced(\DStilde)} + p\cdot\Ex{}{\induced(\IStilde)}}.
    \]
    Also recall that $\Ex{}{\induced(\DStilde)} \le (\OPTIR)^2 + \OPTIR$ and $\Ex{}{\induced(\IStilde)} \le \Ex{}{|\IStilde|}\cdot (1 + \gamma)$. Note that the denominator above is trivially lower bounded by $p \cdot \Ex{}{\induced(\IStilde)} \ge p\Ex{}{|\IStilde|}$. This gives 
    \[
        \eps \le \frac{(\OPTIR)^2 + p\cdot\left(\Ex{}{|\IStilde|}\cdot (1 + \gamma) - \Ex{}{|\IStilde|}\right)}{p\Ex{}{|\IStilde|}}
    =   \gamma + \frac{(\OPTIR)^2}{p\Ex{}{|\IStilde|}}.
    \]
    It follows that our cost is at most
    \[
        O(\OPTIR + (\eps + p)n)
    \preceq \OPTIR + \gamma n + pn + \frac{n(\OPTIR)^2}{p\Ex{}{|\IStilde|}}
    \preceq \OPTIR + \gamma n + \frac{n\cdot\OPTIR}{\sqrt{\Ex{}{|\IStilde|}}},
    \]
    where the second step holds if we set $p = \frac{\OPTIR}{ \sqrt{\Ex{}{|\IStilde|}}}$.

    \paragraph{When $\alpha(n-2\OPTIR) - \OPTIR \ge \alpha n / 2$ does not hold.} Finally, we show that if the assumption $\alpha(n-2\OPTIR) - \OPTIR \ge \alpha n / 2$ is violated, we can achieve an $O(\OPTIR)$ cost easily.

    Note that we must have $\OPTIR > \alpha n / 6$ in this case; otherwise, we would have
    \[
        \alpha(n-2\OPTIR) - \OPTIR
    \ge \alpha[n- 2 \cdot (\alpha n / 6)] - \alpha n / 6
    \ge \alpha n - \alpha n / 3 - \alpha n / 6
    =   \alpha n / 2,
    \]
    a contradiction. The second step above holds since $\alpha \in [0, 1]$. 
    Also, we may assume that $\OPTIR \le n / 2$; otherwise, any persuasive signaling scheme would give a cost of at most $n = O(\OPTIR)$.
    
    Now, recall that $\alpha = \frac{n - \OPTIR}{\Ex{}{\induced(V \setminus \DStilde)}}$. Then, $\OPTIR > \alpha n / 6$ and $\OPTIR \le n / 2$ together imply
    \[
        \OPTIR > \frac{n(n - \OPTIR)}{6\cdot \Ex{}{\induced(V \setminus \DStilde)}}
    \ge \frac{n^2}{12\cdot \Ex{}{\induced(V \setminus \DStilde)}}
    \ge \frac{n^2}{12(n + 2m)},
    \]
    where $m$ is the total edge weight in the graph. Then, by \Cref{cor:degenerate-signal}, we can achieve a cost of
    \[
        \frac{n^2}{n + 2m}
    \le 12\OPTIR = O(\OPTIR).
    \]
\end{proof}

\subsection{Proof of \Cref{thm:weighted-upper-general}}
\label{app:proof-weighted-upper-general}

\begin{proof}[Proof of \Cref{thm:weighted-upper-general}.]
    Let $m$ be the total edge weight in the graph, and $\gamma > 0$ be a parameter to be chosen later. Define $\D \in \Delta(2^V)$ as the distribution of the random set obtained from including each $v \in V$ independently with probability $r \coloneqq \min\left\{\frac{\gamma n}{2m}, 1\right\}$. Clearly, we have $\Ex{S \sim \D}{|S|} = rn$ and 
    \[
        \Ex{S \sim \D}{\induced(S)}
    =   rn + 2r^2m
    =   \left(1 + \frac{2rm}{n}\right)\cdot\Ex{S \sim \D}{|S|}
    \le (1 + \gamma)\Ex{S \sim \D}{|S|}.
    \]

    By \Cref{lemma:weighted-upper-technical}, there is a persuasive signaling scheme with a cost of, up to a constant factor, at most
    \[
        \OPTIR + \gamma n + \frac{n\cdot\OPTIR}{\sqrt{\min\left\{\frac{\gamma n^2}{2m}, n\right\}}}
    \preceq \OPTIR + \gamma n + \OPTIR\cdot \max\left\{\sqrt{m / \gamma}, \sqrt{n}\right\}.
    \]
    Under this choice of $\gamma = (\OPTIR)^{2/3}m^{1/3}n^{-2/3}$, that above cost reduces to
    \[
        (\OPTIR)^{2/3}m^{1/3}n^{1/3} + \OPTIR \cdot \sqrt{n}.
    \]
    In addition, recall from \Cref{cor:degenerate-signal} that we can easily achieve a cost of $O(n^2/m)$. Therefore, we can always achieve the minimum between these two and obtain an upper bound of
    \begin{align*}
        \min\left\{(\OPTIR)^{2/3}m^{1/3}n^{1/3} + \OPTIR \cdot \sqrt{n}, \frac{n^2}{m}\right\}
    &\le \min\left\{(\OPTIR)^{2/3}m^{1/3}n^{1/3}, \frac{n^2}{m}\right\}  + \OPTIR \cdot \sqrt{n}\\
    &\preceq n^{3/4}\cdot(\OPTIR)^{1/2} + \OPTIR\cdot\sqrt{n}.
    \end{align*}
    Note that the second term above dominates the first only if $\OPTIR = \Omega(\sqrt{n})$, at which point both terms are greater than the trivial cost of $n$. Therefore, we can further simplify the bound to $O\left(n^{3/4}\cdot(\OPTIR)^{1/2}\right)$.
\end{proof}

\subsection{Proof of \Cref{thm:weighted-upper-special}}
\label{app:proof-weighted-upper-special}
Before proving the theorem, we first present a lemma that lower bounds the size of the maximum independent set in a graph.

\begin{lemma}\label[lemma]{lemma:IS-lower-bound}
    A graph with $n$ vertices and $m$ edges contains an independent set of size $\Omega\left(\min\{n^2/m, n\}\right)$.
\end{lemma}

\begin{proof}
    Consider the greedy algorithm that keeps adding the vertex with the smallest degree, and then removing the vertex along with its neighbors. As long as the number of remaining vertices, $n'$, is at least $n / 2$, the smallest degree is at most $\frac{2m}{n'} \le 4m / n$, so each iteration removes at most $4m/n + 1$ vertices. The maximal independent set obtained from this greedy approach has size at least $\frac{n / 2}{4m/n+1} = \Omega(\min\{n^2/m, n\})$. 
\end{proof}

Now we are ready to prove \Cref{thm:weighted-upper-special}.
\begin{proof}[Proof of \Cref{thm:weighted-upper-special}.]
    Let $m$ denote the total edge weight of the graph. Since the edge weights are lower bounded by $\delta$, there are at most $m / \delta$ edges in the graph. By \Cref{lemma:IS-lower-bound}, the graph contains an independent set $\IS$ of size $\Omega(\min\{n^2 / (m / \delta), n\}) = \Omega(\min\{\delta n^2/m, n\})$.

    Let $\D \in \Delta(2^V)$ be the degenerate distribution at $\IS$. We have $\Ex{S \sim \D}{\induced(S)} = \Ex{S \sim \D}{|S|} = |\IS|$. Then, applying \Cref{lemma:weighted-upper-technical} with $\gamma = 0$ gives a persuasive signaling scheme with a cost of
    \[
        O\left(\OPTIR + \frac{n \cdot \OPTIR}{\sqrt{\min\{\delta n^2/m, n\}}}\right)
    =   O\left(\OPTIR\cdot\left(\sqrt{n} + \sqrt{m / \delta}\right)\right).
    \]

    Again, we use the fact that we can achieve a cost of $O(n^2/m)$ (\Cref{cor:degenerate-signal}). This give an upper bound of
    \[
        O\left(\min\left\{\OPTIR\cdot\left(\sqrt{n} + \sqrt{m / \delta}\right), \frac{n^2}{m}\right\}\right)
    =   O\left(\left(n\cdot\OPTIR\right)^{2/3}\delta^{-1/3} + \OPTIR \cdot \sqrt{n}\right).
    \]
    Note that the $\OPTIR \cdot\sqrt{n}$ term dominates the first term only if $\OPTIR = \Omega(\sqrt{n})$, at which point the first term is already $\Omega(n)$. Therefore, the upper bound can be simplified to $O\left((n\cdot\OPTIR)^{2/3}\delta^{-1/3}\right)$.
\end{proof}

\subsection{Proof of \Cref{lem:wasteful-weighted}}
\label{app:proof-wasteful-weighted}
\begin{proof}[Proof of \Cref{lem:wasteful-weighted}.]
    For the sake of contradiction, assume $\|W\btheta\|_1<n+\frac{\|\btheta\|-\OPT}{\OPT}$. Therefore, it must be the case that for all $i\in[n]$,
    \begin{align*}
        1\le (W\btheta)_i<1+\frac{\|\btheta\|_1-\OPT}{\OPT}=\frac{\|\btheta\|_1}{\OPT}.
    \end{align*}
    Recall from the proof of \Cref{lem:wasteful-general} that $\OPT$ is characterized by the linear program in \cref{eq:primal-lp}, and its dual is given by \cref{eq:dual-lp}.
    We construct dual variables {$\dualvar^\star$} as follows, where $R_{\max}$ is the largest value among all coordinates of $W\btheta$:
    \[
    \dualvar^\star\triangleq \frac{1}{R_{\max}}\btheta,\qquad
        R_{\max}\triangleq\max_{i\in[n]}(W\btheta)_i<\frac{\|\btheta\|_1}{\OPT}.
    \]
    Note that {$\dualvar^\star$} is a feasible solution of \cref{eq:dual-lp} because
    \[
    W\dualvar^\star=\frac{1}{R_{\max}}(W\btheta)\le\vecone.
    \]
     Connecting to the optimal objective value $\OPT=\min_{\btheta}\vecone^{\top}\btheta$ of \cref{eq:primal-lp}, we have
    \begin{align*}
        \OPT=&\min_{\btheta}\vecone^{\top}\btheta
        \ge\max_{\dualvar} \vecone^{\top}\dualvar
        \ge \vecone^{\top} \dualvar^\star\tag{weak duality}\\
        = &\frac{1}{R_{\max}}\left(\vecone^{\top} \btheta\right)=\frac{1}{R_{\max}}\cdot\|\btheta\|_1
        \tag{definition of $\dualvar^\star$}\\
        >&\OPT,\tag{$R_{\max}<\frac{\|\btheta\|_1}{\OPT}$}
    \end{align*}
    a contradiction! Therefore, we must have $\|W\btheta\|_1\ge n+\frac{\|\btheta\|_1-\OPT}{\OPT}$.
\end{proof}

\subsection{Proof of \Cref{thm:weighted-strict-improvement}}
\label{app:proof-weighted-improve-ic}
\begin{proof}[Proof of \Cref{thm:weighted-strict-improvement}.]
    Let $\btheta^\star \in [0, 1]^V$ be an optimal IR solution with cost $\OPTIR$, $\btheta \in [0,1]^V$ be the optimal stable solution {with cost $\OPTIC$,}
    and {$\alpha \in [0, 1]$} be a signal value to be defined later. {By \Cref{lem:wasteful-general}, the fact that $\|\btheta\|_1 = \OPTIC > \OPTIR \ge \OPT$ implies that $\btheta$ must be wasteful.}
    
    {Let} $\D_A\in\Delta(\{0,\alpha\}^V)$ be the distribution defined by first sampling {a random set $S \subseteq V$} by including each vertex {$v$} independently with probability $\theta_v^\star$, {and} then labeling each $v\in V$ with $s_v=\alpha\cdot\1{v\in S}$. Let $\D_{\btheta^\star} \in \Delta(2^V)$ denote the distribution of such a random set $S$. In addition, we define $\D_B\in\Delta([0,1]^V)$ {as} the degenerate distribution supported on {$\{\btheta\}$}. Let $$\signalspace=\{\theta_v: v\in V\}\cup\{\alpha\}$$ be the signal space including all distinct values in $\btheta$ and $\alpha$. Note that we must have $0\in\signalspace$, because a stable solution can only be wasteful when there exists {$v$ such that} $\theta_v=0$; {otherwise, the stability condition implies $W\btheta = \vecone$, i.e., $\btheta$ is not wasteful. Thus,} we have $|\signalspace|\le n+1$.
    
    We define a parametrized family of distributions $\D_\eps\in\Delta(\signalspace^V)$ for $\eps\in[0,1]$ as
    \[\forall s \in \signalspace^V,\qquad \D_\eps(s)=\eps\cdot\D_A(s)+(1-\eps)\cdot\D_B(s).\]
    Recall from \cref{lemma:persuasive-general} {that} the signaling scheme $\D_\eps$ is persuasive if for all $\theta\in\signalspace$, we have
    \begin{align}
        \Delta_\theta=\Ex{\bs \sim \D_\varphi}{\sum_{v \in V}\1{s_v = \theta} \left(\sum_{v'\in N(v)}W_{v,v'}s_{v'}-(1-\theta)\right)}\ge0,
    \end{align}
    and the above is tight for all $\theta\in \signalspace\setminus\{0\}$.
    We now analyze $\Delta_\theta$ for each $\theta\in\signalspace$.
    \begin{itemize}
        \item \textbf{Case 1.} $\theta\in\signalspace\setminus\{0,\alpha\}$.
        Since $\btheta$ is stable, for any $v$ such that $\theta=\theta_v>0$, the feasibility condition must be tight, i.e.,
        \[
        u_v(\btheta)=\theta+\sum_{v'\in N(v)}W_{v,v'}\theta_{v'}=1.
        \]
        This, together with the fact that $\D_B$ is a degenerate distribution on $\btheta$, implies
        \[
        \Delta_\theta^{\D_B}=\sum_{v\in V}\1{\theta_v = \theta} \left(\sum_{v'\in N(v)}W_{v,v'}\theta_{v'}-(1-\theta)\right)=0.
        \]
        Therefore, we have \[
        \Delta_\theta^{\D_\eps}=(1-\eps)\Delta_\theta^{\D_B}=0.
        \]
        \item \textbf{Case 2.} $\theta=\alpha$. If $\alpha=\theta_v$ for some $v$, the analysis in the previous case gives us $\Delta_\alpha^{\D_B}=0$. {On the other hand,}
        \begin{align*}
            \Delta_\alpha^{\D_A}=&\Ex{\bs \sim \D_A}{\sum_{v \in V}\1{s_v = \alpha} \left(\sum_{v'\in N(v)}W_{v,v'}s_{v'}-(1-\alpha)\right)}\\
            =&\Ex{S\sim\D_{\btheta^\star}}{\sum_{v \in S}\left(\alpha\sum_{v'\in S}W_{v,v'}-1\right)}\tag{$s_v=\alpha\cdot\1{v\in S}$ }\\
            =&\alpha\cdot\Ex{S\sim\D_{\btheta^\star}}{\induced(S)}-\OPTIR\tag{$\Ex{S\sim\D_{\btheta^\star}}{|S|}=\OPTIR$}.
        \end{align*}
        Therefore, to guarantee $\Delta_\alpha^{\D_A}=0$, we set 
        \begin{align}
            \alpha=\frac{\OPTIR}{\Ex{S\sim\D_{\btheta^\star}}{\induced(S)}}.
            \label{eq:alpha-choice}
        \end{align}
        Note that $\alpha\le1$ because $\Ex{S\sim\D_{\btheta^\star}}{\induced(S)}\ge\Ex{S\sim\D_{\btheta^\star}}{|S|}=\OPTIR$.
        In addition, we have
        \[
        \Delta_\alpha^{\D_\eps}=\eps\Delta_\alpha^{\D_A}+(1-\eps)\Delta_\alpha^{\D_B}=0.
        \]
        \item \textbf{Case 3.} $\theta=0$. We first analyze the slack in $\D_A$. 
        \begin{align*}
            \Delta_0^{\D_A}=&\Ex{\bs \sim \D_A}{\sum_{v \in V}\1{s_v = 0} \left(\sum_{v'\in N(v)}W_{v,v'}s_{v'}-1\right)}\\
            =&\Ex{S\sim\D_{\btheta^\star}}{\sum_{v \in V\setminus S}\left(\alpha\sum_{v'\in S}W_{v,v'}-1\right)}\tag{$s_v=\alpha\cdot\1{v\in S}$ }\\
            =&\alpha\cdot\Ex{S\sim\D_{\btheta^\star}}{\cut(S,V\setminus S)}-(n-\OPTIR)\tag{$\Ex{S\sim\D_{\btheta^\star}}{|V\setminus S|}=n-\OPTIR$}\\
            =&\OPTIR\cdot
            \frac{\Ex{S\sim\D_{\btheta^\star}}{\cut(S,V\setminus S)}}{{\Ex{S\sim\D_{\btheta^\star}}{\induced(S)}}}-(n-\OPTIR)\tag{choice of $\alpha$ in \cref{eq:alpha-choice}}\\
            \ge&\OPTIR\cdot
            \frac{n-2\OPTIR}{(\OPTIR)^2+\OPTIR}-(n-\OPTIR)\tag{last two properties in \cref{lemma:IR-sol-rounding}}\\
            =&-\frac{\OPTIR(n-\OPTIR+1)}{\OPTIR+1}.
        \end{align*}
        For $\Delta_0^{\D_B}$, since $(W\btheta)_v-1=0$ for all coordinates where $\theta_v>0$, we have
        \begin{align*}
            \Delta_0^{\D_B}
        &=  \sum_{v\in V}\1{\theta_v = 0} \left(\theta_v + \sum_{v'\in N(v)}W_{v,v'}\theta_{v'}-1\right)\\
        &=  \sum_{v\in V}\1{\theta_v = 0} \left((W\btheta)_v-1\right)\\
        &=  \sum_{v\in V}\1{\theta_v = 0} \left((W\btheta)_v-1\right) + \sum_{v\in V}\1{\theta_v > 0} \left((W\btheta)_v-1\right) \tag{$\theta_v > 0 \implies (W\btheta)_v = 1$}\\
        &=  \|W\btheta\|_1-n\\
        &\ge\frac{\|\btheta\|_1-\OPT}{\OPT}=\frac{\OPTIC-\OPT}{\OPT},
        \end{align*}
        where the inequality follows from \cref{lem:wasteful-weighted}.
        Therefore, for the mixed distribution $\D_\eps$, we have
        \begin{align*}
            \Delta_0^{\D_\eps}=&~\eps\Delta_0^{\D_A}+(1-\eps)\Delta_0^{\D_B}\\
        \ge&~-\eps\cdot \frac{\OPTIR(n-\OPTIR+1)}{\OPTIR+1}
        +(1-\eps)\frac{\OPTIC-\OPT}{\OPT}.
        \end{align*}
        Therefore, for the signaling scheme $\D_\eps$ to be persuasive, it suffices to choose $\eps\in[0,1]$ that such that 
        \begin{align*}
            &-\eps\cdot\frac{\OPTIR(n-\OPTIR+1)}{\OPTIR+1}
        +(1-\eps)\cdot\frac{\OPTIC-\OPT}{\OPT}\ge0\\
        \Longleftrightarrow\ &
        \eps \le \frac{\OPTIC-\OPT}{\OPTIC-\OPT
        +\frac{\OPT\cdot\OPTIR(n-\OPTIR+1)}{\OPTIR+1}}=\frac{\PoS-1}{\PoS-1+\frac{\OPTIR(n-\OPTIR+1)}{\OPTIR+1}},
        \end{align*}
        which justifies the choice of $\eps$ in \Cref{eq:choice-of-eps-weighted}. The cost of this signaling scheme is 
        \begin{align*}
            \Cost(\D_\eps)=&\eps\alpha\Ex{S\sim\D_A}{|S|}+(1-\eps)\OPTIC
            \le\OPTIC-\eps(\OPTIC-\OPTIR),
        \end{align*}
        where we have used $\alpha \le 1$ and $\Ex{S\sim\D_A}{|S|}=\OPTIR$.
    \end{itemize}
\end{proof}

\section{Omitted Proofs from Section~\ref{sec:weighted-lower}}
\label{app:weighted-lower}
\subsection{Proof of \Cref{thm:weighted-binary-signals-fail}}
\label{app:proof-binary-fails}
\begin{proof}[Proof of \Cref{thm:weighted-binary-signals-fail}]
    We prove the theorem using the example from \Cref{fig:binary-schemes-fail}: Set $n = 3k + 2$ for some integer $k$. Let $C_1, C_2, \ldots, C_k$ be $k$ cliques, each of size $3$. Let $v_1,v_2$ be two additional vertices. Each $v_i$ is connected to each other and all the vertices in $C_{1:k}$. All edges have weight $1/2$. The optimal solution is $2$, achieved when both $v_1$ and $v_2$ play $1$ while the other vertices play $0$.

    Fix any binary signaling scheme $\D\in\Delta(2^{[n]})$. Consider the distribution $\D'$ over $\Vec{x}=(x_3,x_2,x_1,x_0,y)\in\{0, 1, \ldots, k\}^4\times \{0, 1, 2\}$ induced by $\D$ through the following procedure: 
    \begin{itemize}
        \item Sample $S\sim\D$.
        \item For $i \in \{0,1,2,3\}$, let $x_i$ be the number of cliques in $C_1, \ldots, C_k$ that contain exactly $i$ vertices in $S$.
        \item Let $y=|S\cap\{v_1,v_2\}|$.
    \end{itemize}
    Note that $\D'$ is a sufficient statistics of $\D$ in terms of the persuasiveness of the induced binary signaling scheme: When $\D$ is the degenerate distribution at $S \subseteq [n]$, the resulting $\Vec{x}=(x_3,x_2,x_1,x_0,y)$ satisfies
    \begin{align}
    \label{eq:suff-stat}
        \begin{cases}
            |S|=3x_3+2x_2+x_1+y\\
            \cut(S,V\setminus S)=3x_3+3x_2+2x_1-\frac{y}{2}(3x_3+x_2-x_1-3x_0)+\frac{\1{y=1}}{2}\\
            \induced(S)=|S| + 3x_3+x_2+y(3x_3+2x_2+x_1)+\1{y=2}
        \end{cases}
    \end{align}
    By \Cref{lemma:persuasive-binary}, we have
    \begin{align*}
        \frac{\Ex{\D}{\cut(S,V\setminus S)}}{n-\Ex{\D}{|S|}}\ge
        \frac{\Ex{\D}{\induced(S)}}{\Ex{\D}{|S|}}=\frac{1}{\alpha},
    \end{align*}
    where $\alpha$ is the value of the non-zero signal.
    
    Substituting \Cref{eq:suff-stat} into the above condition, subtracting $1$ from both sides, and using the fact that $n=3(x_0+x_1+x_2+x_3)+2$, we obtain
    \begin{align}
        &\frac{\Ex{\D'}{3x_3+2x_2+\frac{x_1y}{2}+y-3x_0(1-\frac{y}{2})-\frac{y}{2}(3x_3+x_2)-\left(2-\frac{\1{y=1}}{2}\right)}}{n-\Ex{\D}{|S|}}\nonumber\\
        &\quad\ge\frac{\Ex{\D'}{3x_3+x_2+y(3x_3+2x_2+x_1)+\1{y=2}}}{\Ex{\D}{|S|}}.\label{eq:raw-ineq-stable}
    \end{align}

    We compare the numerators of both sides of \Cref{eq:raw-ineq-stable} when $(x_3, x_2, x_1, x_0, y)$ is deterministic. When $y \in \{0, 2\}$, we have $2\cdot\1{y = 2} \ge y$, and it follows that
    \[
        2\cdot[3x_3 + x_2 + y(3x_3 + 2x_2 + x_1) + \1{y = 2}]
    \ge 3x_3 + 2x_2 + \frac{x_1y}{2} + y,
    \]
    which is, in turn, lower bounded by the numerator on the left-hand side of \Cref{eq:raw-ineq-stable}.
    When $y = 1$ and $x_1 + x_2 + x_3 \ge 1$, we have
    \begin{align*}
        2\cdot[3x_3 + x_2 + y(3x_3 + 2x_2 + x_1) + \1{y = 2}]
    &=  12x_3 + 6x_2 + 2x_1\\
    &=  (11x_3 + 5x_2) + x_1 + (x_1 + x_2 + x_3)\\
    &\ge3x_3 + 2x_2 + \frac{x_1y}{2} + y,
    \end{align*}
    which is again lower bounded by the numerator on the left-hand side of \Cref{eq:raw-ineq-stable}.
    Finally, in the remaining case that $x_1 = x_2 = x_3 = 0$, $x_0 = k$, and $y = 1$, the numerator on the right-hand side of \Cref{eq:raw-ineq-stable} is $0$, while the numerator on the left is given by
    \[
        y - 3x_0(1 - y/2) - (2 - \1{y=1}/2)
    =   - \frac{3}{2}k - \frac{1}{2}
    <   0.
    \]

    Let $\LHS$ and $\RHS$ denote the left- and right-hand sides of \Cref{eq:raw-ineq-stable}, respectively. Then, the above case analysis shows that
    \[
        2\cdot \frac{\Ex{\D}{|S|}}{n-\Ex{\D}{|S|}}\cdot \RHS \ge\LHS\ge\RHS,
    \]
    which implies $\Ex{\D}{|S|} \ge n / 3 =\Omega(n)$.

    Recall that by \Cref{lemma:persuasive-binary}, the cost of the binary scheme is given by $\alpha\cdot \Ex{\D}{|S|}$. Therefore, it is left to show $\alpha=\Omega(1)$. To prove $\alpha=\Omega(1)$, it suffices to upper bound $\frac{\Ex{\D}{\induced(S)}}{\Ex{\D}{|S|}}$ by a constant. We have
    \begin{align*}
        \frac{\Ex{\D}{\induced(S)}}{\Ex{\D}{|S|}} - 1&=\frac{\Ex{\D'}{3x_3+x_2+y(3x_3+2x_2+x_1)+\1{y=2}}}{\Ex{\D}{|S|}}\\
        &\le\frac{\Ex{\D'}{9x_3+5x_2+2x_1+\1{y=2}}}{\Ex{\D'}{3x_3+2x_2+x_1+y}}\le 3.
    \end{align*}
    Therefore, $\alpha= \frac{\Ex{\D}{|S|}}{\Ex{\D}{\induced(S)}} \ge\frac{1}{4}=\Omega(1)$. This lower bounds the cost of the binary scheme by $\alpha\cdot\Ex{\D}{|S|}=\Omega(n)$.
\end{proof}

\subsection{Projection to a low-dimensional space}
\label{app:projection-low-dim}
 Consider an identity-independent persuasive signaling scheme specified by $\D_\varphi \in \Delta(\signalspace^V)$ with signal space $\signalspace$. The scheme naturally induces a distribution $\D'$ over $(x, y, \alpha_1, \alpha_2, \ldots, \alpha_k) \in \signalspace^{k+2}$ in the following way:
\begin{itemize}
    \item First, we sample a set of signals from $\D_\varphi$.
    \item Let $x$ and $y$ be the signals sent to the two center vertices.
    \item Choose one of the $k^2$ cliques uniformly at random, and set $\alpha_1, \ldots, \alpha_k$ to the signals sent to the $k$ vertices in that clique.
\end{itemize}
Define random variable $\Sum \coloneqq \sum_{i=1}^{k}\alpha_i$. It is easy to verify that the cost of the signaling scheme is given by
\[
    \Ex{(x, y, \alpha) \sim \D'}{x + y + k^2 \cdot \Sum},
\]
so our goal is to prove that $\Ex{\D'}{\Sum} = \Omega(1)$.

Recall from \Cref{def:slack} that $\Delta_{\theta}$ is the following quantity:
\begin{align*}
\Delta_\theta=\Contrib_\theta-(1-\theta)\cdot\Numtheta_\theta,
\end{align*}
which is the total amount of slack at signal value $\theta$. \Cref{lemma:persuasive-general} states that a valid scheme must satisfy $\Delta_0 \ge 0$ and $\Delta_{\theta} = 0$ for every $\theta > 0$.

Our first step is to re-write each $\Delta_{\theta}$ as an expectation over the distribution of $(x, y, \alpha_1, \ldots, \alpha_k)$. This reduces the dimensionality of the signaling scheme from $n = k^3 + 2$ to $k + 2$, and slightly simplifies the notations in the remainder of the proof.

\begin{lemma}[informal]\label[lemma]{lemma:slack-contribution}
    The slacks $(\Delta_{\theta})_{\theta \in \signalspace}$ can be equivalently defined as the expected outcome of the following procedure:
    \begin{itemize}
        \item Draw $(x, y, \alpha_1, \ldots, \alpha_k) \sim \D'$.
        \item $\Delta_x \gets \Delta_x + \frac{y + k^2\Sum}{2} - (1 - x)$.
        \item $\Delta_y \gets \Delta_y + \frac{x + k^2\Sum}{2} - (1 - y)$.
        \item For every $i \in [k]$, $\Delta_{\alpha_i} \gets \Delta_{\alpha_i} + k^2\cdot\left[\frac{\Sum + x + y + \alpha_i}{2} - 1\right]$.
    \end{itemize}
\end{lemma}

Informally, the $k^2$ factors in the above account for the fact that every clique is only sampled with probability $1/k^2$.

Below is a more formal statement of the lemma above.

\begin{lemma}[Formal version of \Cref{lemma:slack-contribution}]\label[lemma]{lemma:slack-contribution-formal}
    Let $\D_\varphi \in \Delta(\signalspace^V)$ be the distribution that specifies a signaling scheme for the graph. Formally, the scheme draws
    \[
        s = (s_{c,1}, s_{c,2}, s_{1,1},\ldots,s_{1,k}, s_{2,1},\ldots,s_{2,k},\ldots, s_{k^2,1},\ldots,s_{k^2,k})
    \sim \D_\varphi,
    \]
    labels the two centers with $s_{c,1}$ and $s_{c,2}$, and labels the $j$-th vertex in the $i$-th clique with $s_{i,j}$.

    Let $\D' \in \Delta(\signalspace^{k+2})$ be the projection of $\D_\varphi$, i.e., $\D'$ is the distribution of
    \[
        (s_{c,1}, s_{c,2}, s_{i,1}, s_{i,2}, \ldots, s_{i,k})
    \]
    when $s \sim \D_\varphi$ and $i$ is drawn uniformly at random from $[k^2]$.

    Then, for any $\theta \in \signalspace$, we have
    \begin{align*}
        &~\Ex{s \sim \D_\varphi}{\sum_{v \in V}\1{s_v = \theta}\cdot\frac{1}{2}\sum_{u \in N(v)}s_u}
    -   (1 - \theta)\cdot\Ex{s \sim \D_\varphi}{\sum_{v \in V}\1{s_v = \theta}}\\
    =   &~\Ex{(x,y,\alpha)\sim\D'}{\1{x = \theta}\cdot\left(\frac{y+k^2\Sum}{2} - (1 - x)\right)} + \Ex{(x,y,\alpha)\sim\D'}{\1{y = \theta}\cdot\left(\frac{x+k^2\Sum}{2} - (1 - y)\right)}\\
       &~\qquad+\sum_{i=1}^{k}\Ex{(x,y,\alpha)\sim\D'}{\1{\alpha_i=\theta}\cdot k^2\cdot\left(\frac{\Sum+x+y+\alpha_i}{2} - 1\right)}.
    \end{align*}
\end{lemma}

\begin{proof}
    It suffices to prove the identity for deterministic values of $s_{c,1}$, $s_{c,2}$, and $(s_{i,j})_{(i,j)\in[k^2]\times[k]}$, i.e., when $\D_\varphi$ is degenerate. The general case then follows from taking an expectation. To this end, we show that each $s_{c,1}$, $s_{c,2}$, and $s_{i,j}$ contributes the same amount to both sides of the equation. For simplicity, we will assume that the $k^3+2$ entries of $s$ are distinct; the general case follows from the same argument.
    \paragraph{Contribution from $s_{c,1}$ and $s_{c,2}$.}
    When $\theta = s_{c,1}$, the left-hand side is given by
    \[
        \frac{1}{2}\sum_{u \in N((c,1))}s_u - (1 - s_{c,1})
    =   \frac{s_{c,2}}{2} + \frac{1}{2}\sum_{i=1}^{k^2}\sum_{j=1}^{k}s_{i,j} - (1 - s_{c,1}).
    \]
    The right-hand side is
    \begin{align*}
        \Ex{(x,y,\alpha)\sim\D'}{\frac{y+k^2\Sum}{2} - (1 - x)}
    &=  \frac{s_{c,2}}{2} - (1 - s_{c,1}) + \frac{k^2}{2}\Ex{(x,y,\alpha)\sim\D'}{\Sum}\\
    &=  \frac{s_{c,2}}{2} - (1 - s_{c,1}) + \frac{k^2}{2}\cdot\frac{1}{k^2}\sum_{i=1}^{k^2}\sum_{j=1}^{k}s_{i,j},
    \end{align*}
    which is equal to the left-hand side.

    By symmetry, the contributions from $s_{c,2}$ to both sides are also equal.
    \paragraph{Contribution from $s_{i,j}$.} When $\theta = s_{i,j}$, the left-hand side reduces to
    \[
        \frac{1}{2}\sum_{u \in N((i,j))}s_u - (1 - s_{i,j})
    =   \frac{s_{c,1} + s_{c,2}}{2} + \sum_{j' \in [k]\setminus\{j\}}\frac{s_{i,j'}}{2} - 1 + s_{i,j}
    =   \frac{s_{c,1} + s_{c,2} + s_{i,j}}{2} + \sum_{j' \in [k]}\frac{s_{i,j'}}{2} - 1,
    \]
    while the right-hand side is given by
    \[
        \Ex{(x,y,\alpha)\sim\D'}{\1{\alpha_j = s_{i,j}}\cdot k^2\cdot\left(\frac{\Sum+x+y+\alpha_j}{2} - 1\right)}.
    \]
    Note that $\alpha_j = s_{i,j}$ holds only when $\alpha_1, \ldots, \alpha_k$ are equal to $s_{i,1}, \ldots, s_{i,k}$, which happens with probability $1/k^2$ by definition of $\D'$. Thus, the expression above can be simplified to
    \[
        \frac{1}{k^2}\cdot k^2\cdot\left(\frac{\sum_{j'=1}^ks_{i,j'} + s_{c,1} + s_{c,2} + s_{i,j}}{2} - 1\right)
    =   \frac{s_{c,1}+s_{c,2}+s_{i,j}}{2} + \sum_{j' \in [k]}\frac{s_{i,j'}}{2} - 1,
    \]
    which is exactly the left-hand side. This completes the proof.
\end{proof}

\subsection{Choice of the test function}
\label{app:test-function}
Recall that our goal is to show that any distribution $\D'$ over $(x, y, \alpha_1, \ldots, \alpha_k)$ induced by a valid signaling scheme $\D_\varphi$ must satisfy $\Ex{\D'}{\Sum} = \Omega(1)$. In order for $\D'$ to be induced by a valid scheme, it must satisfy $\Delta_0 \ge 0$ and $\Delta_\theta = 0$ for every $\theta \in \signalspace\setminus\{0\}$, where $(\Delta_\theta)_{\theta \in \signalspace}$ are obtained from $\D'$ via \Cref{lemma:slack-contribution}. While the number of constraints might be large, we will carefully take only a few linear combinations of them, such that they are sufficient for lower bounding $\Ex{\D'}{\Sum}$.

In particular, we note that for every function $f:\signalspace\to\R$ with $f(0) \ge 0$, we must have
\[
    \sum_{\theta \in \signalspace}f(\theta)\cdot\Delta_{\theta} \ge 0.
\]
The hope is that we can choose a few simple functions $f$ such that, after plugging $x$, $y$, and $\alpha_i$ into the inequality above, we obtain a good lower bound on $\Ex{\D'}{\Sum}$.

\paragraph{The constant function.} We start with the most simple choice of $f(\theta) \equiv 1$. By \Cref{lemma:slack-contribution}, the contribution of the combination $(x, y, \alpha_1, \ldots, \alpha_k)$ to $\sum_{\theta \in \signalspace}f(\theta)\cdot\Delta_{\theta} = \sum_{\theta \in \signalspace}\Delta_{\theta}$ is given by:
\begin{align*}
    &~\left[\frac{y + k^2\Sum}{2} - (1 - x)\right] + \left[\frac{x + k^2\Sum}{2} - (1 - y)\right] + k^2\cdot\sum_{i=1}^{k}\left[\frac{\Sum + x + y + \alpha_i}{2} - 1\right]\\
=   &~\frac{3}{2}(x+y) + k^2\Sum - 2 + k^3\cdot\left[\frac{\Sum + x + y}{2} - 1\right] + \frac{k^2}{2}\Sum\\
=   &~\left(\frac{1}{2}k^3 + \frac{3}{2}k^2\right)\Sum + \left(\frac{1}{2}k^3 + \frac{3}{2}\right)(x + y) - (k^3 + 2).
\end{align*}
By linearity, the condition $\sum_{\theta \in \signalspace}\Delta_{\theta} \ge 0$ can be written as
\[
    \left(\frac{1}{2}k^3 + \frac{3}{2}k^2\right)\Ex{\D'}{\Sum} + \left(\frac{1}{2}k^3 + \frac{3}{2}\right)\Ex{\D'}{x + y} \ge (k^3 + 2).
\]
Note that this condition alone is not enough for lower bounding $\Ex{\D'}{\Sum}$, since the above can be easily satisfied by setting $\Sum \equiv 0$ and $x, y \equiv 1$.

\paragraph{Getting rid of signal $1$.} Later, we will consider another choice of $f(\theta) = \frac{\theta}{1 - \theta}$, which is ill-defined at $\theta = 1$. To avoid this issue, we now take a detour and show that sending signal $1$ is essentially useless, so it is without loss of generality that $1 \notin \signalspace$.

Note that whenever signal $1$ is sent to a vertex, all the neighbouring vertices must receive signal $0$. This is because, according to \Cref{lemma:slack-contribution}, whenever $\Delta_1$ is changed, the change is always non-negative. Therefore, to ensure $\Delta_1 = 0$, there cannot be any non-zero contribution to vertices that receive signal $1$. In other words, all neighbours of such vertices must receive signal $0$.

Therefore, the support of $\D'$ must be contained in $(\signalspace \setminus \{1\})^{k+2} \cup\{e_1, e_2, \ldots, e_{k+2}\}$, where $e_i$ is the vector with $1$ at the $i$-th coordinate and zeros elsewhere. Using \Cref{lemma:slack-contribution}, we can verify that the contribution of each $e_i$ to $\Delta_0$ is non-positive, while $\Delta_\theta$ is unaffected for $\theta > 0$. In particular, when either $x = 1$ or $y = 1$, $\Delta_0$ increases by
\[
    \left(\frac{1}{2} - 1\right) + k \cdot k^2 \cdot \left(\frac{1}{2} - 1\right)
=   -\frac{k^3+1}{2} < 0.
\]
When one of the $\alpha_i$'s is equal to $1$, $\Delta_0$ gets increased by
\[
    2\cdot\left(\frac{k^2}{2} - 1\right) + (k - 1)\cdot k^2\cdot\left(\frac{1}{2} - 1\right)
=   - \frac{1}{2}k^3 + \frac{3}{2}k^2 - 2
\le 0.
\]

Therefore, we may let $\D''$ be the restriction of $\D'$ to $(\signalspace\setminus\{1\})^{k+2}$. As argued above, $\D''$ is still valid in the sense that $(\Delta_\theta)_{\theta \in \signalspace\setminus\{1\}}$ induced by $\D''$ (according to \Cref{lemma:slack-contribution}) still satisfies $\Delta_0 \ge 0$ and $\Delta_{\theta} = 0$ for all $\theta > 0$.

We still need to show that the restriction to $(\signalspace\setminus\{1\})^{k+2}$ does not significantly increase the cost, i.e., $\Ex{\D''}{\Sum}$ should be $O(1)\cdot\Ex{\D'}{\Sum}$. To this end, it suffices to prove that $\D'((\signalspace\setminus\{1\})^{k+2})$ is lower bounded by $\Omega(1)$, so that the normalization does not blow up the cost. We start by claiming that
\[
    \D'(\{e_3, e_4, \ldots, e_{k+2}\})
=   \pr{(x, y, \alpha) \sim \D'}{\alpha_1 = 1 \vee \alpha_2 = 1 \vee\cdots\vee \alpha_k = 1}
\le \frac{1}{10}.
\]
This holds because, otherwise, we have $\Ex{(x, y, \alpha) \sim \D'}{\Sum}
\ge \frac{1}{10}\cdot 1
=   \Omega(1)$, and we are done.
Furthermore, we claim that
\[
    \D'(\{e_1, e_2\})
=   \pr{(x, y, \alpha) \sim \D'}{x = 1 \vee y = 1}
\le \frac{1}{10}.
\]
Earlier, we showed that when either $x = 1$ or $y = 1$, $\sum_{\theta\in\signalspace}\Delta_\theta$ is decreased by $\frac{k^3+1}{2}$. Furthermore, in general, each fixed value of $(x, y, \alpha_1, \ldots, \alpha_k)$ increases $\sum_{\theta\in\signalspace}\Delta_{\theta}$ by
\[
    \left(\frac{1}{2}k^3 + \frac{3}{2}k^2\right)\Sum + \left(\frac{1}{2}k^3 + \frac{3}{2}\right)(x + y) - (k^3 + 2)
\le \left(\frac{1}{2}k^3 + \frac{3}{2}k^2\right)\Sum + 1.
\]
This gives
\begin{align*}
    0 \le \sum_{\theta \in \signalspace}\Delta_\theta
\le &~-\pr{(x,y,\alpha)\sim\D'}{x=1\vee y=1}\cdot\frac{k^3+1}{2}\\ &~+ \Ex{(x,y,\alpha)\sim\D'}{\1{x\ne 1 \wedge y \ne 1}\cdot\left[\left(\frac{1}{2}k^3 + \frac{3}{2}k^2\right)\Sum + 1\right]}\\
\le &~-\pr{(x,y,\alpha)\sim\D'}{x=1\vee y=1}\cdot\frac{k^3+1}{2}+ \left(\frac{1}{2}k^3 + \frac{3}{2}k^2\right)\Ex{(x,y,\alpha)\sim\D'}{\Sum} + 1,
\end{align*}
which, together with $\pr{}{x=1\vee y=1} \ge 1/10$, would imply $\Ex{}{\Sum} = \Omega(1)$.
Therefore, without loss of generality, we may assume that $\D'((\signalspace \setminus \{1\})^{k+2}) \ge 1 - \frac{1}{10} - \frac{1}{10} = \frac{4}{5}$. Therefore,
\[
    \Ex{(x, y, \alpha) \sim \D''}{\Sum}
\le \frac{5}{4}\Ex{(x, y, \alpha) \sim \D'}{\Sum}.
\]
In particular, proving $\Ex{\D''}{\Sum} = \Omega(1)$ would imply $\Ex{\D'}{\Sum} = \Omega(1)$.

\paragraph{Another choice of $f$.} Now, we assume that $1 \notin \signalspace$, and examine $\sum_{\theta \in \signalspace}f(\theta)\cdot\Delta_\theta$ when $f(\theta) = \frac{\theta}{1 - \theta}$. Again, \Cref{lemma:slack-contribution} implies that the contribution of a fixed set of values $(x, y, \alpha_1, \ldots, \alpha_k)$ to $\sum_{\theta \in \signalspace}\frac{\theta}{1-\theta}\cdot\Delta_{\theta}$ is given by
\[
    \frac{x}{1-x}\cdot\left[\frac{y + k^2\Sum}{2} - (1 - x)\right]
+   \frac{y}{1-y}\cdot\left[\frac{x + k^2\Sum}{2} - (1 - y)\right]
+   \sum_{i=1}^{k}\frac{\alpha_i}{1-\alpha_i}\cdot k^2\left[\frac{\Sum + x + y + \alpha_i}{2} - 1\right].
\]
The sum of the first two terms above can be simplified into
\[
    \frac{xy}{2}\cdot\left(\frac{1}{1-x} + \frac{1}{1-y}\right) + \frac{k^2\Sum}{2}\cdot\left(\frac{x}{1-x} + \frac{y}{1-y}\right) - (x + y),
\]
while the last summation can be re-written as
\[
    k^2\sum_{i=1}^{k}\frac{\alpha_i}{1-\alpha_i}\cdot\left[\frac{\Sum+x+y-1}{2}-\frac{1-\alpha_i}{2}\right]
=   \frac{k^2}{2}(\Sum + x + y - 1)\sum_{i=1}^{k}\frac{\alpha_i}{1 - \alpha_i} - \frac{k^2}{2}\Sum.
\]

\subsection{Verify the dual feasibility}
\label{app:dual-feasibility}
We will show that, for any finite signal space $\signalspace \subset [0, 1)$ and every fixed choice of $(x, y, \alpha_1, \ldots, \alpha_k) \in \signalspace^{k+2}$, it holds that
\begin{equation}\label{eq:dual-feasibility}
    \Sum + \beta_2 \cdot \sum_{\theta \in \signalspace}\frac{\theta}{1-\theta}\cdot\Delta_{\theta} \ge \Omega(1) + \beta_1\cdot\sum_{\theta \in \signalspace}\Delta_{\theta},
\end{equation}
where $\beta_2 = \frac{1}{4k}$, $\beta_1 = \frac{1}{2k^2}$, and the $\Omega(1)$ notation hides a positive universal constant that does not depend on $(x, y, \alpha_1, \ldots, \alpha_k)$, as long as $k$ is sufficiently large.

Assuming that \cref{eq:dual-feasibility} holds, taking an expectation and rearranging shows that
\[
    \Ex{\D'}{\Sum}
\ge \Omega(1) + \beta_1\cdot\Ex{}{\sum_{\theta\in\signalspace}\Delta_{\theta}} - \beta_2 \cdot\Ex{}{\sum_{\theta\in\signalspace}\frac{\theta}{1-\theta}\cdot\Delta_{\theta}}
\ge \Omega(1),
\]
which implies the $\Omega(k^2) = \Omega(n^{2/3})$ lower bound on the cost and proves \Cref{thm:weighted-lower-special}.

Now we plug $\beta_1 = \frac{1}{4k}$, $\beta_2 = \frac{1}{2k^2}$, as well as the expressions for $\sum_{\theta\in\signalspace}\Delta_{\theta}$ and $\sum_{\theta\in\signalspace}\frac{\theta}{1-\theta}\cdot\Delta_{\theta}$ into \cref{eq:dual-feasibility}. It is sufficient to prove that
\begin{align*}
    &~\Sum + \frac{xy}{8k}\left(\frac{1}{1-x} + \frac{1}{1-y}\right) + \frac{k}{8}\Sum\cdot\left(\frac{x}{1-x} + \frac{y}{1-y}\right) - \frac{x+y}{4k} + \frac{k(\Sum+x+y-1)}{8}\sum_{i=1}^{k}\frac{\alpha_i}{1-\alpha_i} - \frac{k}{8}\Sum\\
\ge &~\Omega(1) + \left(\frac{k}{4} + \frac{3}{4}\right)\Sum + \left(\frac{k}{4} + \frac{3}{4k^2}\right)(x+y) - \left(\frac{k}{2} + \frac{1}{k^2}\right).
\end{align*}
Since $x + y \le 2 = O(1)$, every additive term of form $(x+y)/\poly(k)$ or $1/\poly(k)$ can be absorbed into the $\Omega(1)$ gap (for all sufficiently large $k$). Thus, it suffices to show
\begin{equation}\label{eq:simplified-dual}
\begin{split}
    &~\frac{k}{2} + \frac{xy}{8k}\left(\frac{1}{1-x} + \frac{1}{1-y}\right) + \frac{k}{8}\Sum\cdot\left(\frac{x}{1-x} + \frac{y}{1-y}\right) + \frac{k(\Sum+x+y-1)}{8}\sum_{i=1}^{k}\frac{\alpha_i}{1-\alpha_i}\\
\ge &~\Omega(1) + \left(\frac{3}{8}k-\frac{1}{4}\right)\Sum + \frac{k}{4}(x+y).
\end{split}
\end{equation}

We will consider the following three different cases:
\begin{itemize}
    \item \textbf{Case 1:} $x + y < 1$ and $\Sum + x + y - 1 \ge 0$. Since $x, y, \alpha_1, \ldots, \alpha_k \in [0, 1)$, we have
    \[
        \frac{x}{1-x} + \frac{y}{1-y} \ge x+y
    \quad\text{and}\quad
    \sum_{i=1}^{k}\frac{\alpha_i}{1 - \alpha_i} \ge \sum_{i=1}^{k}\alpha_i = \Sum.
    \]
    Furthermore, we relax the $\frac{xy}{8k}\left(\frac{1}{1-x}+\frac{1}{1-y}\right)$ term in \cref{eq:simplified-dual} to $0$. Then, it suffices to prove the following:
    \begin{equation}\label{eq:dual-feasibility-case-1}
        \frac{k}{2} + \frac{k}{8}\Sum\cdot(x+y) + \frac{k(\Sum + x + y - 1)}{8}\cdot\Sum \ge \Omega(1) + \left(\frac{3}{8}k-\frac{1}{4}\right)\Sum + \frac{k}{4}(x+y).
    \end{equation}
    For fixed $\Sum$, both sides of \cref{eq:dual-feasibility-case-1} are affine in $x + y$, so it suffices to verify it at $x + y = 1$ and $x + y = \max\{1 - \Sum, 0\}$. At $x + y = 1$, \cref{eq:dual-feasibility-case-1} gets reduced to
    \[
        \frac{k}{8}\Sum^2 - \left(\frac{k}{4}-\frac{1}{4}\right)\Sum + \frac{k}{4} \ge \Omega(1),
    \]
    which is true since the left-hand side is equal to $\frac{k}{8}(\Sum - 1)^2 + \frac{\Sum}{4} + \frac{k}{8} \ge \frac{k}{8} \ge \frac{1}{8}$.

    At $x + y = \max\{1 - \Sum, 0\}$, if $\Sum \le 1$ (so that $x + y = 1 - \Sum$), \cref{eq:dual-feasibility-case-1} is equivalent to
    \[
        \frac{\Sum}{4} + \frac{k}{4}\left(1 - \frac{\Sum^2}{2}\right) \ge \Omega(1),
    \]
    which is true since $\Sum \le 1$ implies $1 - \frac{\Sum^2}{2} \ge 1/2$.
    In the other case that $\Sum > 1$, we have $x + y = 0$, and this reduces \cref{eq:dual-feasibility-case-1} to
    \[
        \frac{k}{2} + \frac{k}{8}\Sum(\Sum-1) - \left(\frac{3}{8}k-\frac{1}{4}\right)\Sum \ge \Omega(1),
    \]
    which always holds since the left-hand side is equal to $\frac{k}{8}(\Sum-2)^2 + \frac{\Sum}{4} \ge \frac{\Sum}{4} > \frac{1}{4}$.
    
    \item \textbf{Case 2:} $x + y < 1$ and $\Sum + x + y - 1 < 0$. Again, we may apply the relaxation
    \[
        \frac{xy}{8k}\left(\frac{1}{1-x} + \frac{1}{1-y}\right) + \frac{k}{8}\Sum\cdot\left(\frac{x}{1-x} + \frac{y}{1-y}\right)
    \ge \frac{k}{8}\Sum\cdot(x+y).
    \]
    However, since the factor $\Sum+x+y-1$ is now negative, the last term on the left-hand side of \cref{eq:simplified-dual} (namely, $\frac{k(\Sum+x+y-1)}{8}\sum_{i=1}^{k}\frac{\alpha_i}{1-\alpha_i}$) is minimized when $\alpha$ is a permutation of $(\Sum, 0, 0, \ldots, 0)$. Thus, we need to prove the following inequality:
    \begin{equation}\label{eq:dual-feasibility-case-2}
        \frac{k}{2} + \frac{k}{8}\Sum(x+y) + \frac{k(\Sum + x + y - 1)}{8}\cdot\frac{\Sum}{1 - \Sum} \ge \Omega(1) + \left(\frac{3}{8}k-\frac{1}{4}\right)\Sum + \frac{k}{4}(x+y).
    \end{equation}

    Again, it suffices to verify the above at $x + y = 0$ and $x + y = 1 - \Sum$, respectively. At $x + y = 0$, \cref{eq:dual-feasibility-case-2} reduces to
    \[
        \frac{k}{2}\cdot(1 - \Sum) + \frac{1}{4}\Sum \ge \Omega(1),
    \]
    which is true since $\frac{k}{2}\cdot(1 - \Sum) + \frac{1}{4}\Sum \ge \frac{1}{4}\cdot(1 - \Sum) + \frac{1}{4}\Sum = \frac{1}{4}$.
    
    At $x + y = 1 - \Sum$, we get exactly the same inequality as the ``$x + y = \max\{1 - \Sum, 0\}$ and $\Sum \le 1$'' part of Case 1, which has already been verified.
    
    \item \textbf{Case 3:} $x + y \ge 1$ (and thus, $\Sum + x + y - 1 \ge 0$). In this case, we claim that the minimum of both $\frac{xy}{8k}\left(\frac{1}{1-x}+\frac{1}{1-y}\right)$ and $\frac{x}{1-x} + \frac{y}{1-y}$ are achieved at $x = y \ge 1/2$. Write $x = \mu + \delta$ and $y = \mu - \delta$ for $\mu = \frac{x+y}{2} \in [1/2, 1)$. We have
    \[
        xy\cdot\left(\frac{1}{1-x} + \frac{1}{1-y}\right)
    =   (\mu^2 - \delta^2)\cdot\frac{2-(\mu+\delta)-(\mu-\delta)}{(1-\mu-\delta)(1-\mu+\delta)}
    =   (2-2\mu)\cdot\left[1 + \frac{\mu^2 - (1 - \mu)^2}{(1 - \mu)^2 - \delta^2}\right],
    \]
    which is minimized at $\delta = 0$, since $\mu \in [1/2, 1)$ guarantees that $2 - 2\mu > 0$ and $\mu^2 - (1 - \mu)^2 \ge 0$.

    Similarly, $\frac{x}{1-x} + \frac{y}{1-y}$ can be written as
    \[
        \frac{1}{1-x} + \frac{1}{1-y} - 2
    =   \frac{2 - 2\mu}{(1 - \mu)^2 - \delta^2} - 2,
    \]
    which is also minimized at $\delta = 0$.
    
    Thus, it suffices to prove \cref{eq:simplified-dual} for the $x = y$ case, i.e., for all $x \in [1/2, 1)$,
    \[
        \frac{k}{2} + \frac{1}{4k}\cdot\frac{x^2}{1-x} + \frac{k}{4}\Sum\cdot\frac{x}{1-x} + \frac{k(\Sum + 2x - 1)}{8}\cdot\Sum \ge \Omega(1) + \left(\frac{3}{8}k - \frac{1}{4}\right)\Sum + \frac{k}{2}x.
    \]
    Consider the function
    \begin{align*}
        g(x, s)
    &\coloneqq
        \frac{k}{2} + \frac{1}{4k}\cdot\frac{x^2}{1-x} + \frac{k}{4}s\cdot\frac{x}{1-x} + \frac{k(s + 2x - 1)}{8}\cdot s - \left(\frac{3}{8}k - \frac{1}{4}\right)s - \frac{k}{2}x\\
    &=  \frac{k}{8}s^2 - \left(\frac{1}{2}k - \frac{1}{4} - \frac{k}{4}x - \frac{k}{4}\cdot\frac{x}{1-x}\right)s + \frac{k}{2}(1-x) + \frac{1}{4k}\cdot\frac{x^2}{1-x}.
    \end{align*}
    We need to prove that $g(x, s) \ge \Omega(1)$ holds for all $x \in [1/2, 1)$ and $s \ge 0$.
    For any fixed $x$, $g(x, s)$ is quadratic in $s$, and we have
    \[
        \inf_{s \ge 0}g(x, s)
    =   \begin{cases}
        g(x, 0), & s^*(x) < 0,\\
        g(x, s^*(x)), & s^*(x) \ge 0,
    \end{cases}
    \]
    where
    \[
    s^*(x) \coloneqq \frac{1}{k/4}\left(\frac{1}{2}k - \frac{1}{4} - \frac{k}{4}x - \frac{k}{4}\cdot\frac{x}{1-x}\right)
    =   2 - x - \frac{x}{1-x} - \frac{1}{k}.
    \]
    The case that $s^*(x) < 0$ is easy, since for any $x \ge 1/2$,
    \[
        g(x, 0)
    =   \frac{k}{2}(1-x) + \frac{1}{4k}\cdot\frac{x^2}{1-x}
    \ge 2\sqrt{\frac{k}{2}(1-x)\cdot\frac{1}{4k}\cdot\frac{x^2}{1-x}}
    =   \frac{2x}{\sqrt{8}}
    \ge \frac{1}{2\sqrt{2}}
    =   \Omega(1).
    \]

    To handle the $s^*(x) \ge 0$ case, we note that
    \[
        \inf_{s \ge 0}g(x,s)
    =   g(x,s^*(x))
    =   g(x, 0) - \frac{k}{8}[s^*(x)]^2,
    \]
    since the minimum of the quadratic function $f(x) = ax^2-bx+c$ where $a > 0$ is given by $c - \frac{b^2}{4a} = f(0) - a(x^*)^2$, achieved at $x^* = \frac{b}{2a}$.
    We will prove that whenever $s^*(x) \ge 0$,
    \begin{equation}\label{eq:g-zero-vs-g-opt}
        g(x, 0)
    \ge \frac{k}{4}[s^*(x)]^2.
    \end{equation}
    This then implies
    \[
        \inf_{s \ge 0}g(x,s)
    =   g(x, 0) - \frac{k}{8}[s^*(x)]^2
    \ge g(x, 0) - g(x, 0) / 2
    = \Omega(1),
    \]
    where the last step follows from our previous argument for $g(x, 0) = \Omega(1)$.

    Since $2 - 1/k - x - x/(1-x) = s^*(x) \ge 0$, we have $x + x / (1-x) \le 2$, which further implies $x \le 2 - \sqrt{2} \le 3/5$. Then, $g(x, 0)$ can be lower bounded as follows:
    \[
        g(x, 0)
    =   \frac{k}{2}(1-x) + \frac{1}{4k}\cdot\frac{x^2}{1-x}
    \ge \frac{k}{2}(1-x) \ge \frac{k}{5}.
    \]
    Furthermore, $x \in [1/2, 1)$ implies that
    \[
        s^*(x)
    =   2 - \frac{1}{k} - x - \frac{x}{1-x}
    \le 2 - 0 - \frac{1}{2} - 1
    =   \frac{1}{2}.
    \]
    Therefore, \cref{eq:g-zero-vs-g-opt} follows from
    \[
        g(x, 0)
    \ge \frac{k}{5}
    > \frac{k}{4}\cdot(1/2)^2
    \ge \frac{k}{4}[s^*(x)]^2.
    \]
    This finishes the proof for Case~3.
\end{itemize}

\end{document}